\documentclass[%
 reprint,
amsmath,amssymb,
aip,
floatfix,
]{revtex4-2}
\usepackage{graphicx}
\usepackage{dcolumn}
\usepackage{hyperref}
\usepackage{url}
\hypersetup{colorlinks,allcolors=black}

\usepackage{hyperref}
\usepackage[mathlines]{lineno}

\usepackage{xfrac}
\usepackage{amsmath}
\usepackage{amsthm}
\DeclareMathOperator{\sign}{sign}
\newcommand{\Mathematica}{\textit{Mathematica\textsuperscript{\resizebox{!}{0.8ex}{\textregistered}}}}
\def\8{\infty}
\def\oh{\sfrac{1}{2}}

\def\oq{\sfrac{1}{4}}

\newcommand*{\I}{\imath}%

\def\eps{\epsilon}

\def\dal{\partial_{\alpha}}
\def\dbe{\partial_{\beta}}
\def\dga{\partial_{\gamma}}

\def\const{\textit{ const }}
\def\undertext#1{\vtop{\hbox{#1}\kern 1pt \hrule}}

\def\abs#1{\left| #1\right|}

\def\VEV#1{\left\langle #1\right\rangle}

\def\tr{\hbox{tr}\,}

\def\dbyd#1#2{\frac{d#1}{d#2}}
\def\pp#1{\frac{\partial}{\partial#1}}
\def\pbyp#1#2{\frac{\partial#1}{\partial#2}}
\def\ff#1{\frac{\delta}{\delta#1}}
\def\fbyf#1#2{\frac{\delta#1}{\delta#2}}

\def\bea{\begin{eqnarray} && &&}
\def\eea{\end{eqnarray}}

\let\oldexp\exp
\renewcommand{\exp}[1]{\oldexp\left(#1\right)}

\def\IINT#1{\int_{#1}}

\def \Schr{Schr\"odinger}

\def\KO{Kolmogorov}
\def\RE{\textbf{Rey}}
\def\NS{Navier-Stokes}

\def\BS{Biot-Savart}

\def\val{v_{\alpha}}
\def\vbe{v_{\beta}}
\def\vga{v_{\gamma}}

\def\oga{\omega_{\gamma}}

\newcommand{\Mod}[1]{\ (\mathrm{mod}\ #1)}

\def\XXint#1#2#3{{\setbox0=\hbox{$#1{#2#3}{\int}$}
     \vcenter{\hbox{$#2#3$}}\kern-.5\wd0}}

\renewcommand{\Re}{\textbf{Re }}
\renewcommand{\Im}{\textbf{Im }}
\def\Binom#1#2{\dbinom{#1}{#2}}

\newcommand{\tpmod}[1]{{\@displayfalse\Mod{#1}}}
\usepackage[english]{babel}

\usepackage{mathtools}

\DeclarePairedDelimiter\floor{\lfloor}{\rfloor}

\newcommand{\pctPDFRot}[2]{
\begin{figure}
    \centering
    \includegraphics[angle=-90,width=0.45\textwidth]{figures/#1.pdf}
    \caption{#2}
    \label{fig::#1}
\end{figure}
}

\newcommand{\pctPDF}[2]{
\begin{figure}
    \centering
    \includegraphics[width=0.45\textwidth]{figures/#1.pdf}
    \caption{#2}
    \label{fig::#1}
\end{figure}
}
\newcommand{\pctWPDF}[3]{
\begin{figure}
    \centering
    \includegraphics[width=#1\textwidth]{figures/#2.pdf}
    \caption{#3}
    \label{fig::#2}
\end{figure}
}
\makeatletter
\def\@email#1#2{%
 \endgroup
 \patchcmd{\titleblock@produce}
  {\frontmatter@RRAPformat}
  {\frontmatter@RRAPformat{\produce@RRAP{*#1\href{mailto:#2}{#2}}}\frontmatter@RRAPformat}
  {}{}
}%
\usepackage{lmodern}
\newenvironment{dedication}
  {
   \vspace*{\stretch{1}}
   \itshape             
   \raggedleft          
  }
  {\par 
   \vspace{\stretch{3}} 
  }

\newtheorem{lemma}{Lemma}

\newtheorem*{theorem}{Quantum Trace Theorem}
\begin{document}

\preprint{APS/123-QED}

\title{Quantum Solution of Classical Turbulence.\\
Decaying Energy Spectrum}

\author{Alexander Migdal}
\affiliation{
\quad Department of Physics, New York University Abu Dhabi,\\
Saadiyat Island, 
            Abu Dhabi,\\
            PO Box 129188, 
           Abu Dhabi,
            United Arab Emirates
}%
\altaffiliation[Also at ]{IAS, Princeton, NJ starting September 2024.
}%
\begin{dedication}
   To dear friend and brilliant physicist, Sreeni--\\my constant source of inspiration and support.
  \end{dedication}
\date{\today}
\begin{abstract}
This paper presents a recent advancement that transforms the problem of decaying turbulence in the Navier-Stokes equations in $3+1$ dimensions into a Number Theory challenge: finding the statistical limit of the Euler ensemble. We redefine this ensemble as a Markov chain, establishing its equivalence to the quantum statistical theory of $N$ fermions on a ring, interacting with an external field associated with random fractions of $\pi$. Analyzing this theory in the turbulent limit, where $N \to \infty$ and $\nu \to 0$, we discover the solution as a complex trajectory (instanton) that acts as a saddle point in the path integral over the density of these fermions.

By computing the contribution of this instanton to the vorticity correlation function, we obtain an analytic formula for the observable energy spectrum—a complete solution of decaying turbulence derived entirely from first principles without the need for approximations or fitted dimensionless parameters. Our analysis reveals the full spectrum of critical indices in the velocity correlation function in coordinate space, determined by the poles of the Mellin transform, which we prove to be a meromorphic function. Real and complex poles are identified, with the complex poles reflecting dissipation and uniquely determined by the famous complex zeros of the Riemann zeta function.

Universal functions of the scaling variables supersede the traditional turbulent scaling laws (K41, Heisenberg, and multifractal). These functions for the energy spectrum, energy decay rate, and velocity correlation significantly deviate from power laws but closely match the results from grid turbulence experiments \cite{GridTurbulence_1966, Comte_Bellot_Corrsin_1971} and recent DNS data \cite{SreeniDecaying} within experimental error margins.
\end{abstract}


\maketitle

\pctPDF{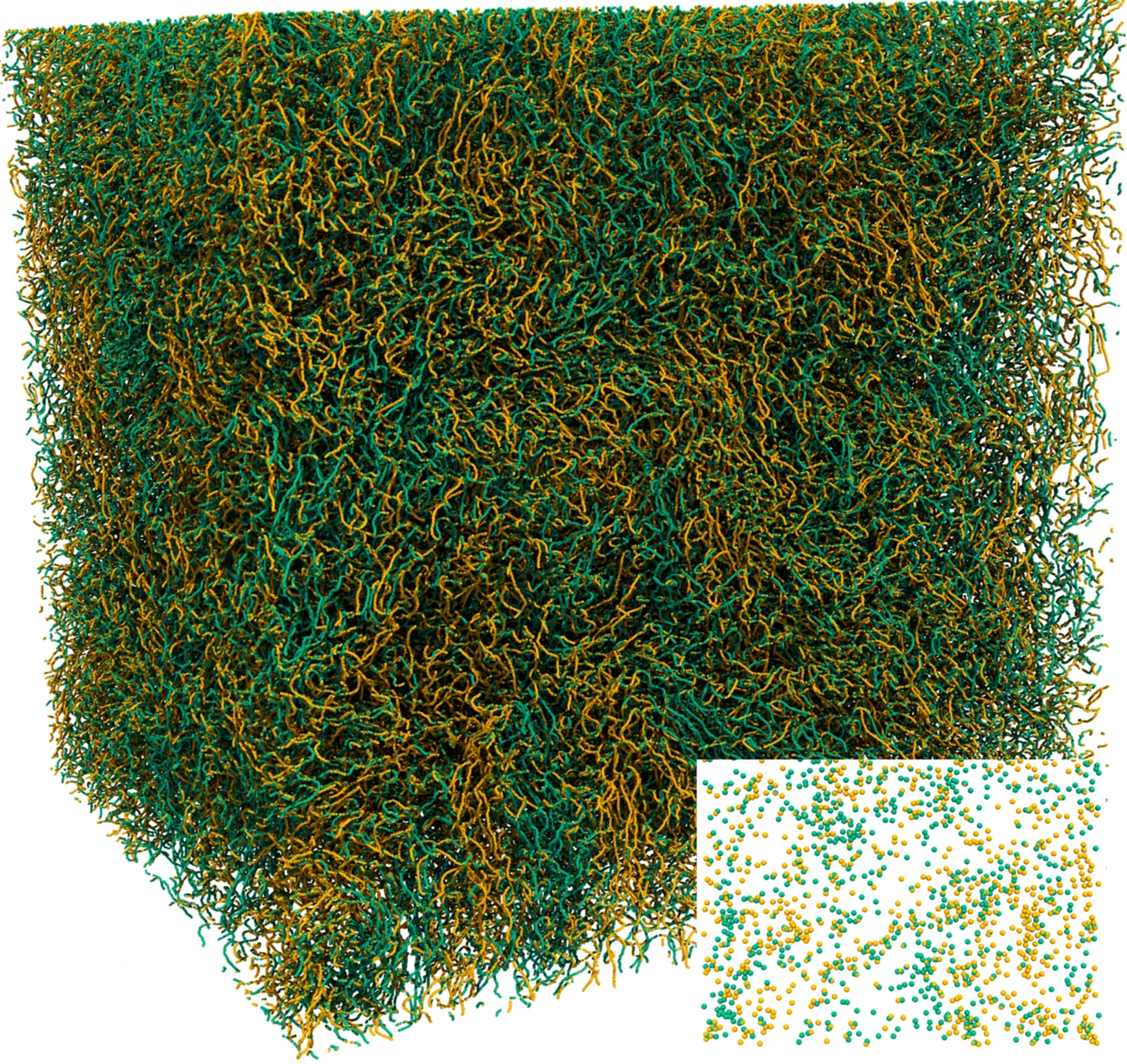}{The three-dimensional vortices in quantum turbulent flow from "Polanco, J.I., Müller, N.P. \& Krstulovic, G. Vortex clustering, polarisation and circulation intermittency in classical and quantum turbulence". Nat Commun 12, 7090 (2021). Licensed under a Creative Commons Attribution (CC BY) license. }

\section{Prologue}

Turbulence appears as an overwhelmingly complex phenomenon. As depicted in Fig.\ref{fig::VortexClustering} from a recent simulation \cite{Polanco2021}, vortex lines of various shapes and sizes are entangled much like spaghetti. This visual complexity raises the question: How can such complexity be described analytically? Yet, it also sparks hope for a simplified statistical description.

With its myriad interacting particles, molecular dynamics similarly presents an intricate challenge. However, despite its complexity, a straightforward statistical description emerges that grows increasingly precise with the escalating complexity of the dynamical system. Maxwell, Boltzmann, and Gibbs demonstrated that Newton's mechanics uniformly cover the energy surface over time, laying the groundwork for statistical mechanics—a robust theory, albeit sometimes computationally challenging, as in critical phenomena.

Why, then, should \NS{} turbulence be any different?

Regrettably, to date, no known analog of the Gibbs distribution exists for turbulent flows. Therefore, a foundational element of turbulence theory must be to devise a substitute for the Gibbs distribution.

Hopf initiated this exploration in 1952 (see the recent review in \cite{Hopf19}), formulating a functional equation that the probability distribution of the turbulent velocity field must satisfy. Through iterative application of this equation to the nonlinear term in the \NS{} equation, one can generate an expansion in inverse powers of viscosity. The core challenge of turbulence theory is solving the Hopf equation in the opposite limit of low viscosity.

This beautiful equation is mathematically as intricate as the vortex spaghetti depicted earlier. Such complexity places turbulence high within the hierarchy of physics theories, nestled between critical phenomena and the quark confinement problem.
\pctPDF{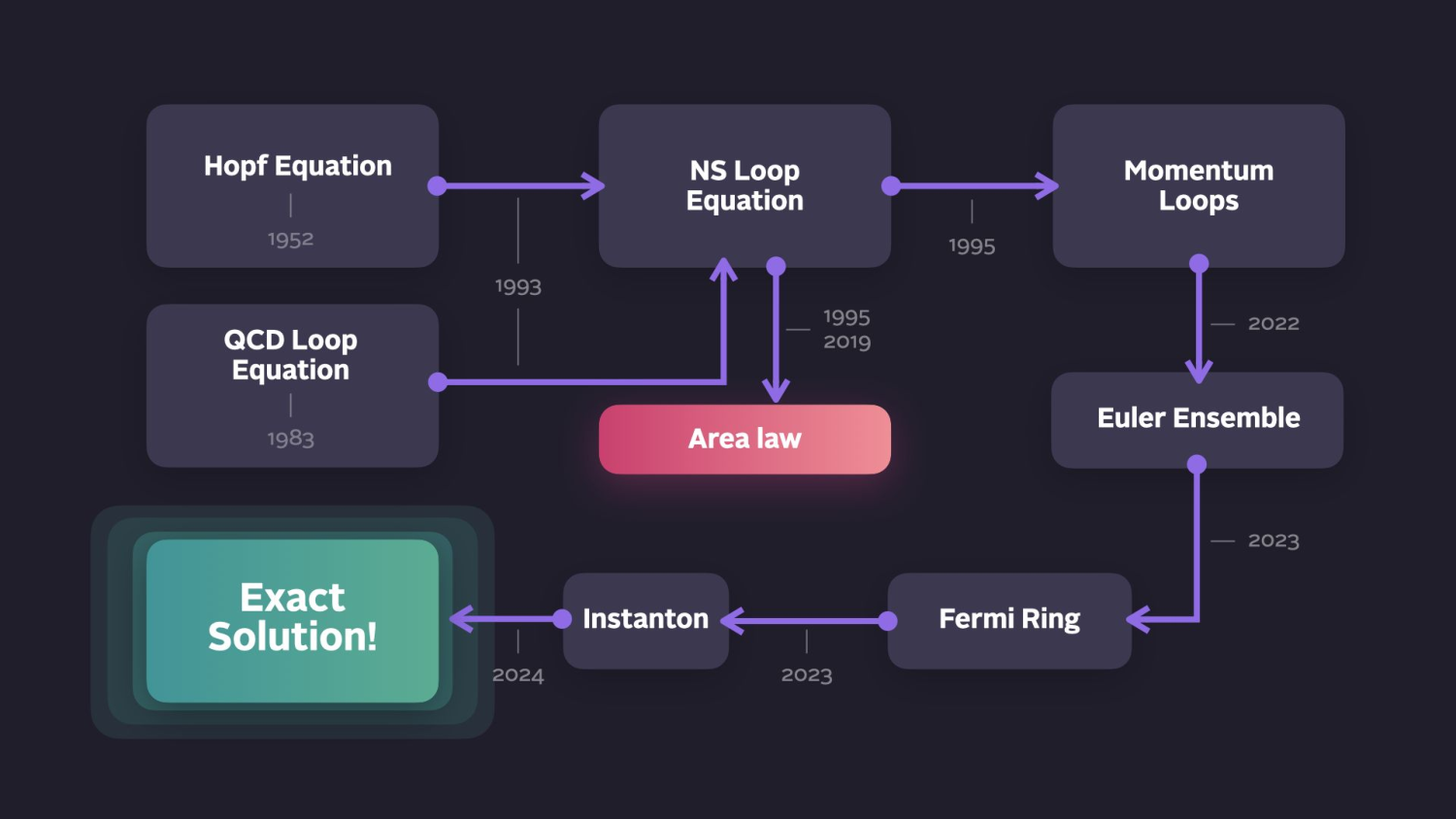}{The path to the microscopic theory.}

Our theory, initiated in 1993 (see Fig.\ref{fig::Dual-Turbulence} for the historical outline), proposes a simpler variant of the Hopf equation—the loop equation—which suffices to define the statistics.

The loop equation corresponds to the \Schr{} equation in loop space. This profound analogy is not a poetic metaphor but a precise mathematical equivalence with significant implications, such as quantum interference effects affecting the scaling laws of classical turbulence.

Using the loop equation, we have identified a new instance of duality between the strong coupling phase of one theory (a fluctuating velocity field in three dimensions) and the weak coupling phase of another (a one-dimensional quantum ring of Fermi particles).

This weak coupling limit can be analytically solved, providing explicit formulas for observables in decaying turbulence, such as the energy dissipation decay index \(n(t) = t \frac{\mathcal{E}(t)}{E(t)}\).

Experimental data from decaying grid turbulence (see Fig. \ref{fig::GridTurbulenceFit}) corroborate our prediction that \(n(\infty) = \frac{5}{4}\), within a 2\% experimental error margin. We anticipate that more precise future measurements will validate this prediction more accurately.

Recent DNS and experiments \cite{SreeniDecaying, GregXi2} compute the energy spectrum and the velocity correlation function,  matching our theory and challenging the Kolmogorov scaling. The accuracy is lower here due to large experimental errors. We found the way to reduce experimental errors by computing effective index for the second moment of velocity difference using numerical Fourier transform. This method dramatically improved the quality of the fit for the effective index in our theory, compared to traditional numerical differentiation of the energy spectrum.  
New large-scale experiments (both real and numerical) are welcome to verify our theory.

\section{Definitions and notations}
We use the units where the constant fluid density $\rho=1$. The 3D vectors and the dot products are denoted like this: $\vec v(\vec{r},t),  \vec{k} \cdot \vec{r} $. We also use the Einstein tensor notation with summation over repeated Greek indexes $ v_\alpha, r_\beta, \omega_{\alpha\beta} = e_{\alpha\beta\gamma} \dbe \vga$, and so on.  In these cases, these indexes run from $1$ to $d$, where $d$ is the dimension of space. In the remaining cases, we work only in 3D space. 
We always consider the Euclidean metric and Cartesian coordinates, so we make no distinction between upper and lower tensor and vector indexes. 
The parameterization of the loops is chosen to run from $0$ to $1$ with periodicity implied beyond these limits. The space loops $\vec{C}(\xi)$ are assumed to be continuous and periodic but not differentiable. In addition, the loop can have some extra periods, in which case this loop represents several separate loops. In particular, the same geometric loop $\vec{C}$ can be covered several times with $\vec{C}_n(\xi) = C(n \xi)$. The momentum loops $\vec{P}(\xi,t)$ depend on time, whereas the spacial loops are static. The momentum loops can have discontinuities $\Delta \vec{P}(\theta,t) = \vec{P}(\theta + 0,t) - \vec{P}(\theta - 0,t)$. The momentum loops are independent of the space loops, being the vector functions of $\xi, t$. The circulation $\Gamma =\oint_C d \vec{r} \cdot \vec v(\vec{r}, t) = \int_0^1 d \xi \vec{C}'(\xi)\cdot \vec v(\vec{C}( \xi), t) $. The whole theory, starting with circulation, is parametric invariant $\xi \Rightarrow f(\xi), f'(\xi) >0$. The operators are denoted like $\hat{a}, \hat{a}^\dagger, \hat{X}$.
We use three types of symbols for differentials.
$d$  is an ordinary differential of variable like $d \xi$ or vector variable like $d \vec{C}(\theta) = d \theta \vec{C}'(\theta)$. $D$ symbolizes a measure for a path integral $D\alpha$; it can be strictly defined by the limit of the multidimensional integral over the discrete values or as a limit of a multidimensional integral over all Fourier harmonics (this is cleaner). We do not need to specify the Gaussian measure, as it is well-known in physics and math. Finally, $\delta$ is a functional variation, and $\ff{}$ is a functional derivative like in $\ff{\phi(\xi)}\Lambda_N[\phi] = \fbyf{\Lambda_N[\phi]}{\phi(\xi)} $ with the $L_2$ norm in functional space.
We also use notations $\partial_t = \pp{t}, \; \partial_\beta = \pp{r_\beta}$ for the time derivative and for components of the spatial derivatives.
\newpage
\section{Summary}
Here are the main results reported in this paper.
\begin{itemize}
\item We review the theory of  \NS{} loop equation, its relation to the Hopf functional equation, and the representation of the loop functional in terms of momentum loop.
\item We present the solution of the loop equation in the inviscid limit of the three-dimensional \NS{} theory in terms of the Euler ensemble. This ensemble consists of a one-dimensional ring of Ising spins in an external field related to random fractions of $\pi$.
   \item
   The continuum limit of this solution, $N \to \infty$, corresponds to the inviscid limit of the decaying turbulence in the \NS{} equation. Effective turbulent viscosity is $\tilde{\nu} = \nu N^2 \to \const{}$. 
   \item
   We derived an analytic formula for energy spectrum and dissipation in finite system \eqref{EbyH}, \eqref{SbyH}, \eqref{answer} and investigated it in Appendix  \ref{RiemanZeros}.
   \item The energy spectrum decays asymptotically as $(\tilde{\nu})^{\sfrac{3}{2}}t^{-\sfrac{1}{2}}\kappa^{-\sfrac{7}{2}} $
   where $ \kappa = k \sqrt{\tilde{\nu} t} $. 
   \item The turbulent kinetic energy decays as $E(t) \propto t^{-\sfrac{5}{4}}$
   \item 
   Both effective indexes $n(t) = -\dbyd{\log E}{\log t}, \mu(\kappa) = \pbyp{\log E(k,t)}{\log k}$ are nontrivial functions of the logarithm of scale and time, approaching $n(\infty) =\sfrac{5}{4}, \mu(\infty) = -\sfrac{7}{2}$, (see Fig.\ref{fig::NPlot}, Fig.\ref{fig::MuIndex}). 
   \item The 1966 experimental values \cite{GridTurbulence_1966, Comte_Bellot_Corrsin_1971}, Fig.\ref{fig::GridTurbulenceFit} of $n(\infty) \approx 1.25 \pm 0.02 $ agree with our prediction.
   \item 
   The results of DNS for the energy spectrum in decaying turbulence (Fig. \ref{fig::DecayingSpectrum} of the review paper \cite{SreeniDecaying}) shows the energy spectrum decaying faster than  $k^{-\sfrac{5}{3}}$. The shape of a log-parabolic curve in DNS matches our prediction Fig.\ref{fig::Hk53}. 
   \item 
   In the section \ref{DNS}, we verify our prediction for the second velocity moment $\VEV{\Delta v^2}(r)$ by numerical Fourier transform of the raw spectrum data from \cite{SreeniDecaying, GDSM24}. 
   The K41 $\frac{2}{3}$ scaling law is ruled out by this DNS, but there is a match within the DNS errors of the scale-dependent effective index with our theory in a wide range of distances ( Fig.\ref{fig::MoveIndex14} ).
   \item 
   We computed the spectrum of singularity indexes $p_n$ in the product $\vec v(0,t)\cdot \vec v(r,t)\sim \sum C_n(t)r^{p_n}$, similar but different from the OPE in CFT. Some of these singularity indexes come in complex conjugate pairs related to the zeros of the Riemann $\zeta$ function \eqref{SpecPoles},\eqref{Epoles},\eqref{VVSpectrum}).
\end{itemize}

\pctPDFRot{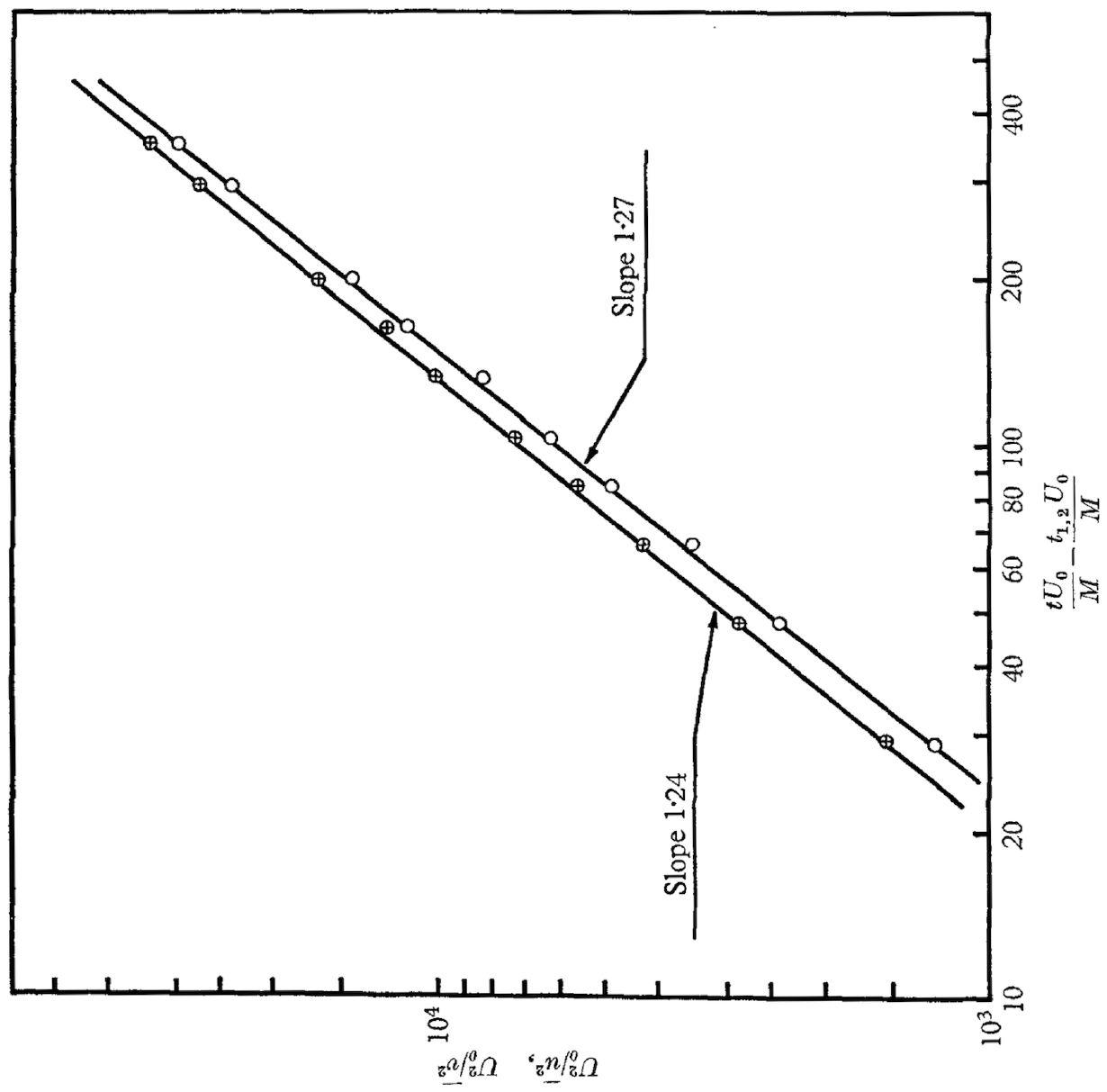}{The experimental data \cite{GridTurbulence_1966, Comte_Bellot_Corrsin_1971} for decaying turbulence behind the oscillating grid. Reproduced with permission from "Comte-Bellot G, Corrsin S. The use of a contraction to improve the isotropy of grid-generated turbulence". Journal of Fluid Mechanics. 1966; 25(4):657-682, © 1966 Cambridge University Press. The two lines correspond to log-log plots of $1/v_\perp^2, 1/v_\parallel^2$ for the flow behind the grid. The total inverse kinetic energy i$ K = 2/(v_\perp^2+v_\parallel^2)$, would have the  the mean slope $(1.24 + 1.27)/2 \approx 1.255 \pm 0.02$ in agreement with our prediction $n(\infty)=\frac{5}{4}$. This plot was shown by K.R.Sreenivasan at the ICTS in Dec'23 in Bangalore, India \cite{DecayingTalk} as the most reliable experimental data. Other measurements and simulations, at different conditions reviewed in \cite{SreeniDecaying}, provide mismatching data influenced by initial and boundary conditions. The only large-scale DNS \cite{SreeniDecaying, GDSM24} covering three time decades also matches our predictions.}
\pctPDF{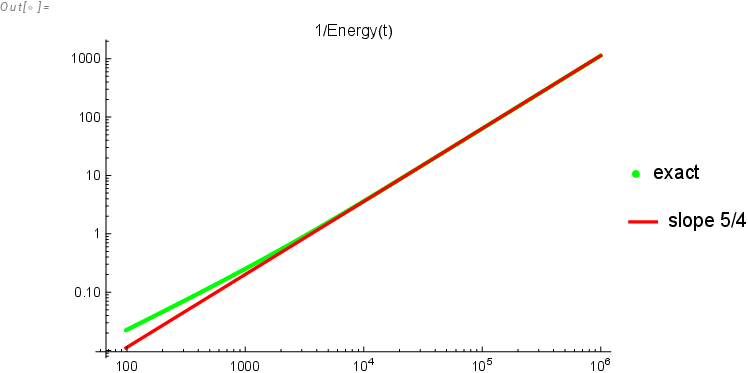}{Green curve is the inverse energy $1/E(t)$ as a function of time. It slowly approaches from above its asymptotic law $1/E(t) \propto t^{\sfrac{5}{4}}$ shown in red.}

\pctPDFRot{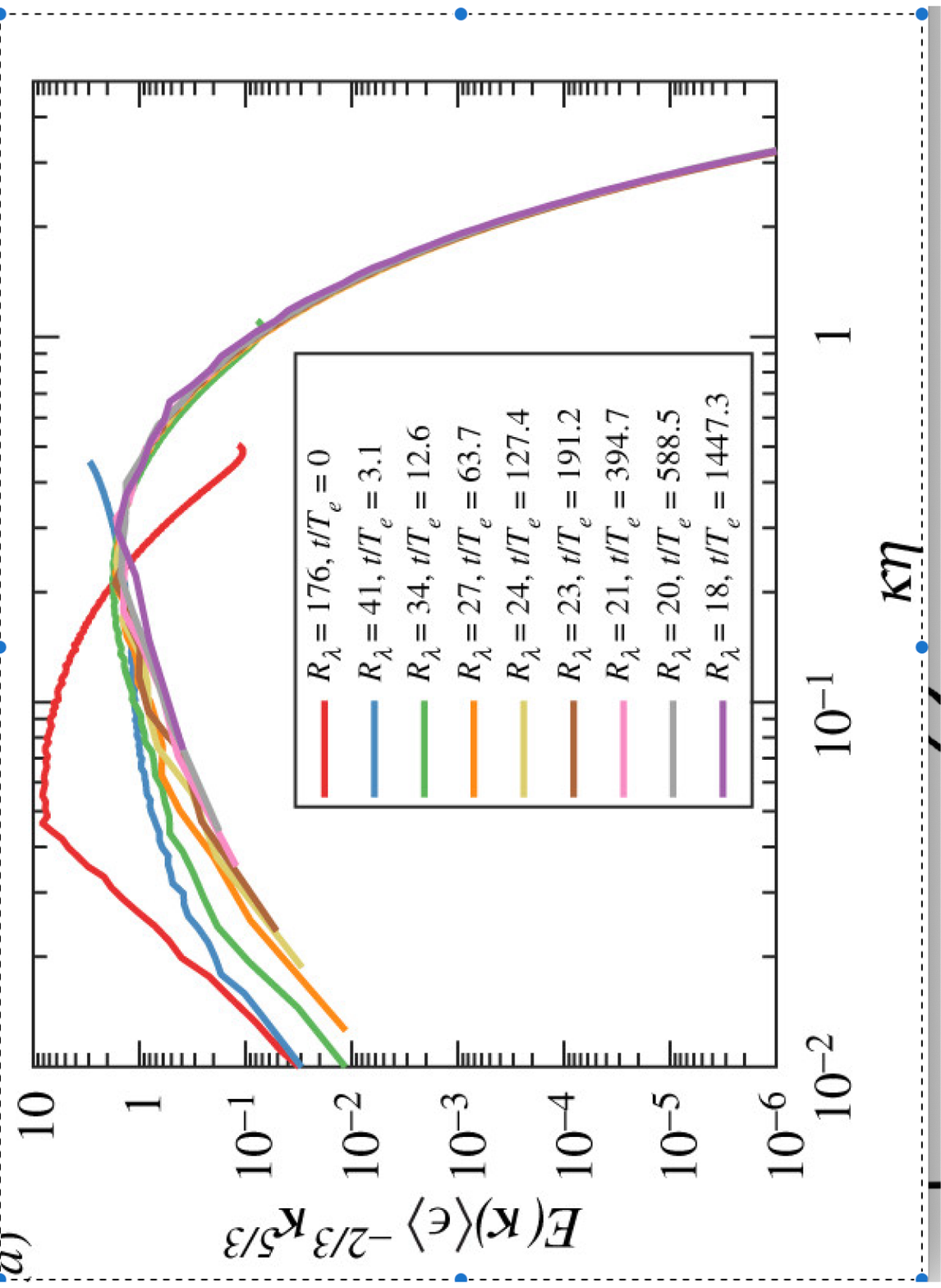}{The log-log plot of deviation from K41 energy spectrum in decaying turbulence, taken from \cite{SreeniDecaying, GDSM24}, Panickacheril John John, Donzis Diego A. and Sreenivasan Katepalli R. 2022, Laws of turbulence decay from direct numerical simulationsPhil. Trans. R. Soc. A.38020210089, licensed under a Creative Commons Attribution (CC BY) license.
The observed curve approximately matches our theoretical curve in Fig.\ref{fig::Hk53}. The K41 spectrum would correspond to a horizontal line, a total mismatch.}
\pctPDF{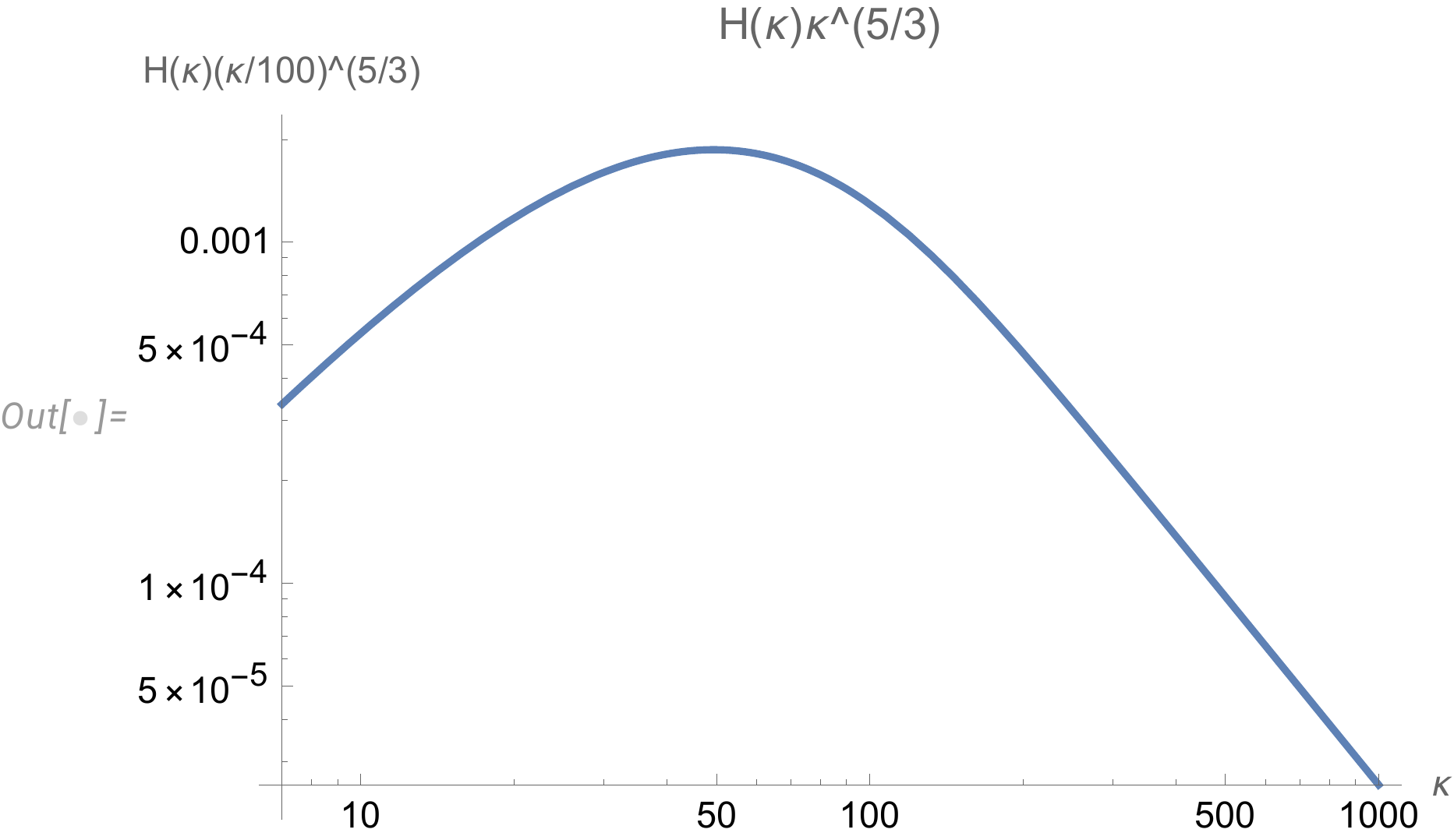}{Log-log plot of the universal function $H(\kappa) \kappa^{\sfrac{5}{3}}$. It starts growing, reaches the maximum, then turns down, asymptotically decaying as $\kappa^{-\sfrac{11}{6}}$. This qualitatively matches the DNS plot at Fig.\ref{fig::DecayingSpectrum}. A larger DNS data at higher Reynolds numbers would be required to verify our theory with more precision.}
\pctPDF{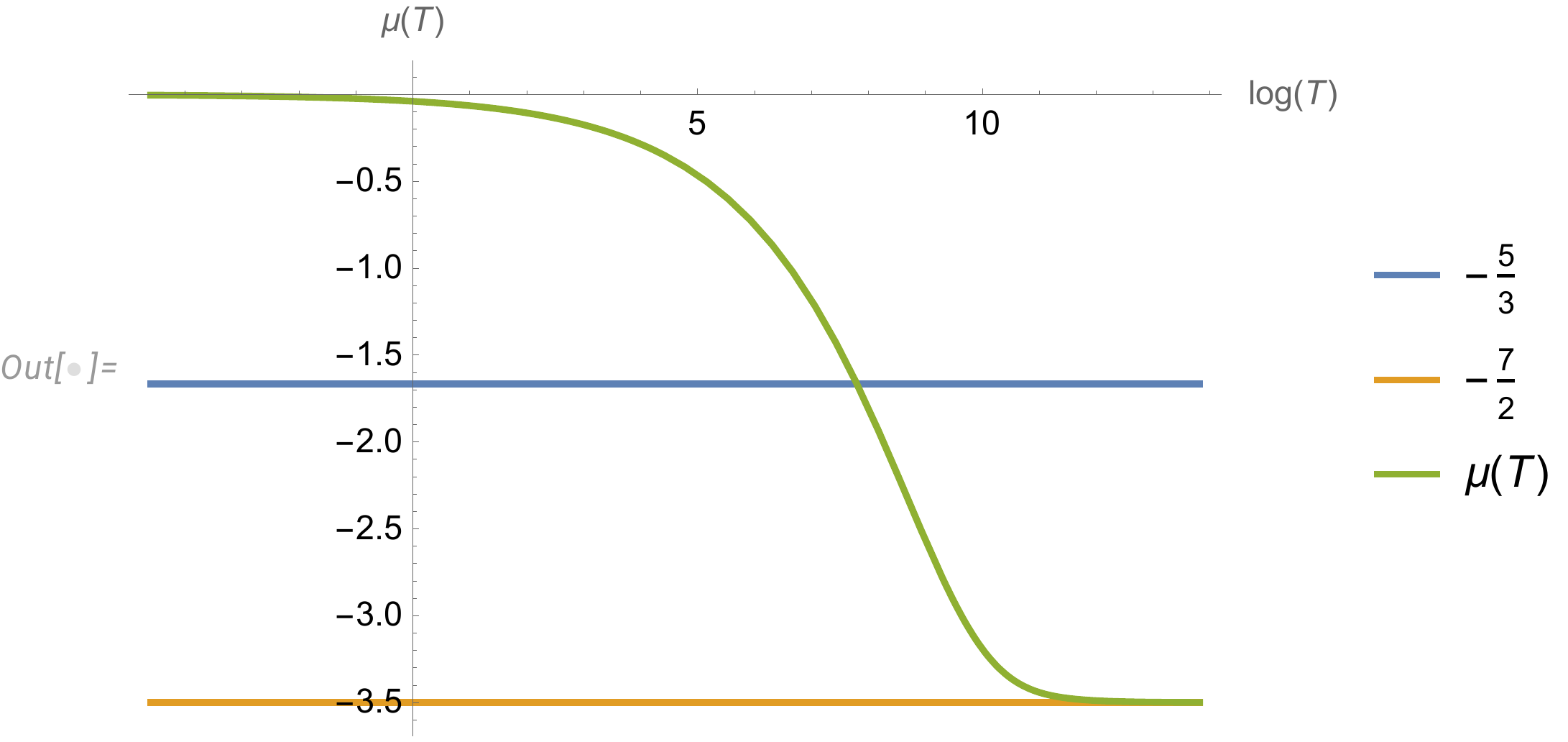}{The theoretical plot of effective spectral  index $ \mu(T) = k\partial_k \log  E(k, t) $ as a function of $T = k^2 \tilde{\nu} t $. It starts at $0$ and goes down to the limit $-\sfrac{7}{2}$.
The DNS plot Fig. \ref{fig::MuIndexDNS} starts higher, at $\mu = 2$, then goes down and stops at approximately $-1$ for accessible decay time $t$. The part from $-0.5$ to $-1$ qualitatively matches our curve, but clearly, more data at higher Reynolds numbers is required to verify this prediction of our theory.}

\pctPDFRot{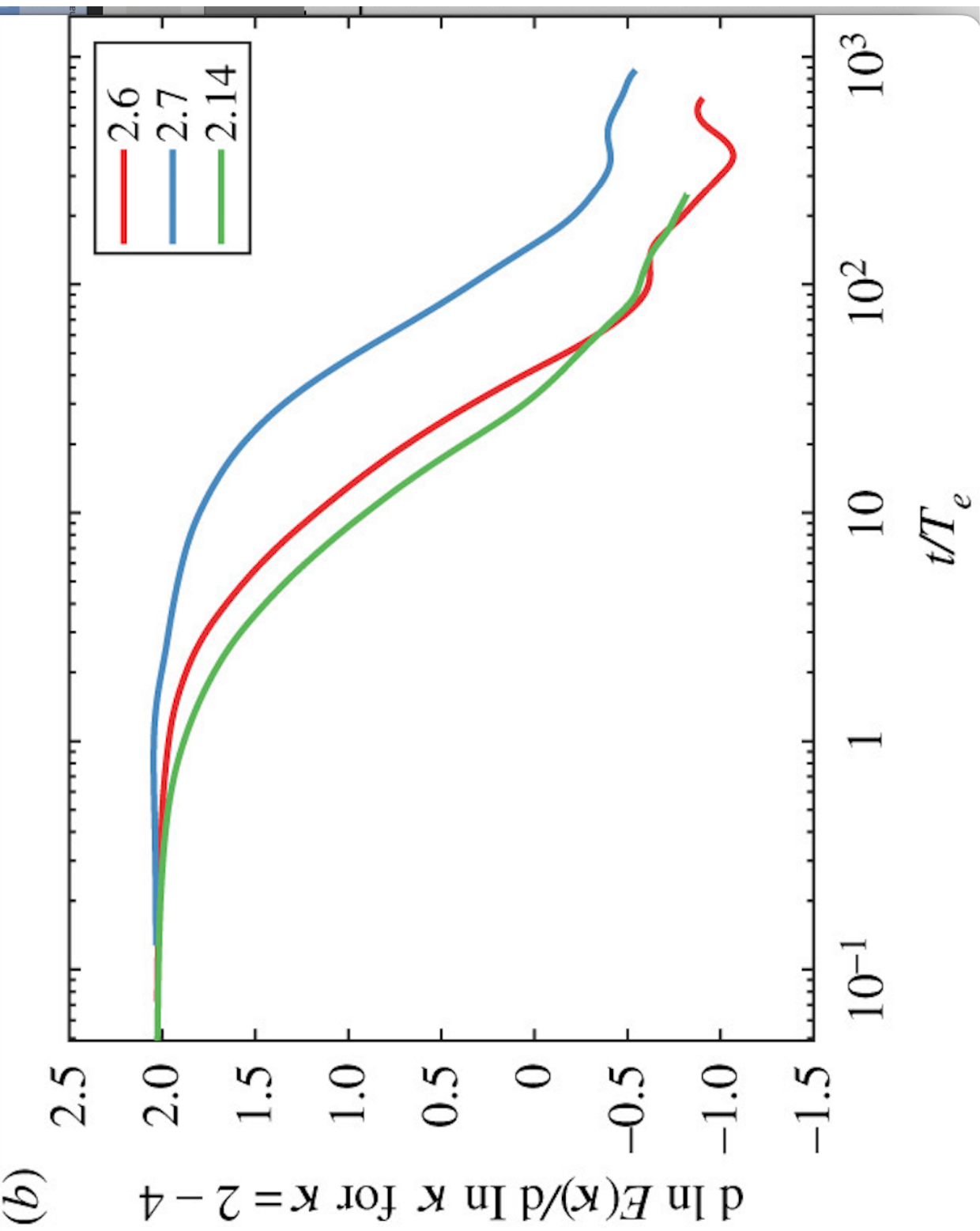}{The DNS plot of spectral index $ \mu(t) = k\partial_k \log  E(k, t) $ as a function of $\log t$  at fixed values of $k$, taken from \cite{SreeniDecaying}, Panickacheril John John, Donzis Diego A. and Sreenivasan Katepalli R. 2022, Laws of turbulence decay from direct numerical simulations. Phil. Trans. R. Soc. A.38020210089, licensed under a Creative Commons Attribution (CC BY) license. The curve goes from $ \mu = 2$ to $\mu \approx -1$, which partly overlaps with our theoretical curve in Fig. \ref{fig::MuIndex}. Larger dataset with higher Reynolds numbers is needed for a quantitative match with our theory.}
\newpage

\section{Introduction}

The original version of this paper was overloaded with formulas, contradicting the preferred 21st-century style. Modern researchers prefer to follow the flow of ideas, with heavy computations hidden under the hood. Eventually, it will become a job for AI to verify computations using \Mathematica{} and digest the results for busy readers.

Considering this, we moved all computations into a series of appendices, leaving only the general ideas, concepts, and results in the main body of the paper. This transformation provides distinct pathways for understanding and applying our study's results based on the reader's background and interests.

Portions of the upcoming discussion are borrowed from our first paper \cite{migdal2023exact} in this series. We included them to clarify the big picture of the theory of turbulence but significantly modified these sections to reflect our deepened understanding of this theory.

\begin{enumerate}
    \item \textbf{Introduction for Physicists.} The physics introduction discusses the potential correspondence between our theoretical developments and decaying turbulence as observed in real-world or numerical experiments. For physicists, this theory offers a solution to the Hopf functional equation for the statistical distribution of the velocity field in the unforced \NS{} equation. This distribution represents a much-needed analog of the Gibbs distribution for decaying turbulence. There are strong indications that our theory is relevant to one of the two universality classes observed in these experiments. 
    \item \textbf{Introduction for Mathematicians.}
    This section summarizes the mathematical framework behind the loop equation \cite{M93} and its solution \cite{migdal2023exact} in terms of the Euler ensemble. Addressed to mathematicians, this introduction allows those focusing on rigorous mathematical theory to bypass the more physically oriented discussions and delve directly into the Euler ensemble as a novel Number Theory set with conjectured connections to decaying turbulence.
    Pure mathematicians may want to prove, refine, or disprove the open mathematical problems and unproven conjectures left in this paper.
    \item \textbf{Guidance for Applied Mathematicians and Engineers.}
    Applied mathematicians and engineers, primarily interested in practical applications rather than abstract theoretical constructs, are directed to this document's section \ref{finitesystem}. Here, they will find final formulas \eqref{SbyH}, \eqref{EbyH}, and accompanying \Mathematica{}code \cite{MB43, MB44} that facilitate the computation of both the energy spectrum and dissipation rates. These formulas are compared with real and numerical experiments in section \ref{DNS}.
    \item \textbf{Notes for the Curious and Skeptical.}
    The fourth category of readers—those curious yet skeptical about applying quantum mechanics to solve complex problems in classical physics—might still harbor doubts after reading the main text of this paper. For them, we have dedicated the last section \ref{discussion}, which addresses some of their lingering questions and perhaps reassures their skepticism. These readers may want to dive into the Appendixes to learn the details of our theory after this discussion, hopefully eliminating their doubts.
\end{enumerate}

In the following sections, we explore these themes in depth, aiming to provide clarity and actionable insights for all readers, regardless of their expertise or interest.

\subsection{Physical Introduction: The Energy Flow and Random Vorticity Structures}

Decaying turbulence is an old topic, traditionally examined within a weak turbulence framework—utilizing a truncated perturbative expansion in inverse powers of viscosity $\nu$ in the forced Navier-Stokes equation. Various phenomenological models have also been aligned with experimental observations, as discussed in the recent review \cite{SreeniDecaying}.

However, the comprehensive turbulence theory requires solving the \NS{} equations in the strong coupling limit $\nu \to 0$, the direct opposite of weak coupling. The universality of strong turbulence, with or without random forcing, poses the initial question in constructing such a theory.

Direct Numerical Simulation (DNS) data for energy decay in turbulent flows, detailed in \cite{SreeniDecaying}, suggest a decay of the dissipation rate $E \sim t^{-n}$ with $n \approx 1.25$ or $1.60$ depending on the initial conditions (finite total momentum or zero total momentum but finite total angular momentum, see \cite{SreeniDecaying}). Thus, two universality classes of decaying turbulence have been identified.

It remains unclear which, if any, of the data in \cite{SreeniDecaying, GDSM24} reached the homogeneous isotropic turbulence limit corresponding to our regime. Moreover, stochastic forces added to the \NS{} equation in simulations might contaminate the natural decay of turbulence. These forces are intended to initiate and enhance the spontaneous stochasticity of turbulent flow. However, in our theory \cite{migdal2023exact}, this inherent stochasticity is related to a dual quantum system and is \textbf{discrete}.

The Gaussian forcing can distort these quantum stochastic phenomena by stirring the flow ubiquitously and constantly. When the forcing is switched off, allowing the turbulence to decay towards a universal stage, energy dissipation should occur inside vorticity structures deeply embedded in the flow by the pure turbulent dynamics we study.

The forthcoming calculation supports the above relationship between singular vorticity distribution and energy flow. In the pure \NS{} scenario, the energy balance reduces to energy dissipation by enstrophy within the bulk, counterbalanced by energy input from boundary forces (e.g., a large sphere encompassing the flow).

The general identity derived from the \NS{} equations by multiplying both sides by $\vec{v}$ and averaging over an ensemble of stochastic solutions is:
    \begin{eqnarray}
      &&\partial_t \VEV{E} +\int_V d^3 r \VEV{ \nu \vec \omega^2} =\nonumber\\
      &&-\int_V d^3 r \dbe \VEV{\vbe \left(p + \oh \val^2\right) + \nu \val ( \dbe \val - \dal \vbe)}
    \end{eqnarray}

Applying the Stokes theorem, the right side reduces to the flow through the boundary $\partial V$ of the integration region $V$. The left side represents the dissipation within this volume:
\begin{eqnarray}
    \partial_t \VEV{E} +\mathcal E_V = -\int_{\partial V} d \vec \sigma \cdot \VEV{\vec v  \left(p + \oh \val^2\right) +\nu \vec \omega \times \vec v }
\end{eqnarray}

This identity is valid for any volume, with the left side indicating viscous dissipation inside $V$ and the right side representing the energy flow through the boundary $\partial V$.

Should a finite collection of vortex structures exist within the bulk, expanding this volume to infinite sphere results in the $\vec \omega \times \vec v$ term disappearing, as no vorticity persists at infinity.

Additionally, the velocity dictated by the Biot-Savart law diminishes as $|\vec{r}|^{-3}$ at infinity, so only the $\vec v p$ term remains significant:
\begin{eqnarray}
    \partial_t \VEV{E} +\VEV{\mathcal E_V} \to -\int_{\partial V} d \vec \sigma \cdot \VEV{\vec v  p}
\end{eqnarray}

This representation of energy flow will remain finite even as the sphere expands if the pressure scales as $p \to - \vec{F} \cdot \vec{r}$, where $\vec{F}$ is the local force at any given point on a large sphere:
\begin{eqnarray}\label{boundaryPump}
    \partial_t \VEV{E} +\VEV{\mathcal E} = f_\alpha\lim_{R \to \infty} R^3\int_{S_2} n_\alpha n_\beta \VEV{ \vbe(R\vec n)}
\end{eqnarray}

What about the Kolmogorov energy flow? It persists within any finite volume surrounding the set of vortexes:
\begin{eqnarray}
   &&\partial_t \VEV{E} +\VEV{\mathcal E} = -\int_V d^3 r \VEV{\vbe \dbe p + \val \vbe\dbe \val} = \nonumber\\
   && -\int_V d^3 r \VEV{\val \vbe\dbe \val}
  -\int_{\partial V} d \vec \sigma \cdot \VEV{\vec v p}
\end{eqnarray}
The triple velocity term in the last equation describes the \KO{} energy flow inside the volume $V$, and the second term represents the energy flow through the boundary.

Without a finite force $\vec{F}$ acting on the boundary, such as with periodic boundary conditions, the boundary integral would be absent, and the \KO{} relation would be fully applicable:
\begin{eqnarray}
    \partial_t \VEV{E} +\VEV{\val \vbe\dbe \val} = -\VEV{\mathcal E}/V ;
\end{eqnarray}
This relation, alongside spatial symmetry properties in $\mathbb{R}_d$, leads to the \KO{} three-point correlation in a steady state $\partial_t \VEV{E} =0$:
    \begin{eqnarray}
    \label{K41vvv}
    &&\VEV{\val(\vec{r}_0) \vbe(\vec{r}_0) \vga(\vec{r} + \vec{r}_0)} =\nonumber\\
    &&\frac{\mathcal E }{(d-1)(d+2) V}
	\left(
	 \delta_{\alpha \gamma} r_{\beta} +
	 \delta_{\beta \gamma} r_{\alpha} -
	 \frac{2}{d}\delta_{\alpha \beta} r_{\gamma}
	\right);
\end{eqnarray}

In the conventional approach to the turbulence problem, periodic Gaussian random forces $\vec{F}(\vec{r},t)$ are added to the \NS{} equations in the conventional approach, based on time averaging:
    \begin{eqnarray}
   &&\VEV{\mathcal E_V} = -\int_V d^3 r \VEV{\vbe \dbe p  - \vbe f_\beta(\vec{r},t) + \val \vbe\dbe \val} = \nonumber\\
   &&\int_V d^3 r \VEV{\vbe f_\beta(\vec{r},t)}
\end{eqnarray}

As the force becomes uniformly distributed across space, we derive another definition with $\mathcal E = \vec{F} \cdot \vec{P}$, where $\vec{P} = \int_V d^3 r \vec v$ is the total momentum.

The phenomenon of turbulence we study exhibits a universal spontaneous stochasticity that does not depend on boundary conditions.

As long as energy flows from the boundaries, confined turbulence in the middle will dissipate this energy through singular vortex tubes. This spontaneous stochasticity results from the random distribution of these singular tubes within the volume of the velocity flow \cite{M22}. The dual picture from our recent theory \cite{migdal2023exact} represents these by random gaps in the momentum curve $\vec{P}(\theta)$, as we shall discuss in the following sections.

The relation between the energy pumping at the boundary and the distribution of vortex blobs in the bulk follows from the Biot-Savart integral:
\begin{eqnarray}\label{BS}
    \vec v(r)= -\vec \nabla \times \int d^3 r' \frac{\vec \omega(r')}{4\pi | r - r'|}
\end{eqnarray}
Generally, a gradient of harmonic potential is added to the Biot-Savart integral, dependent on the boundary conditions. We consider the velocity decaying at infinity, thus not adding such a term.

The net linear momentum $\VEV{\int d^3 r \vec v(0) \cdot \vec v(\vec{r})}$, generally, is not zero in our theory, as we impose no such restriction. This nonvanishing linear moment places our theory in the most general $k^2$ (or Saffman) class.

On a large sphere $\partial V$ with radius $R \to \infty$,
\begin{eqnarray}
   &&\lim_{R \to \infty} R^3 \vec v(R \vec n) \propto  \nonumber\\
   &&\frac{1}{4\pi}
   \sum\limits_{\text{blobs}}\;\int\limits_{\text{blob}} d^3 r' \vec \omega(r') \times 
   \left(\vec{r}'-<\vec{r}>_{\text{blob}}\right);
\end{eqnarray}
Here $<\vec{r}>_{\text{blob}}$ is the geometric center of each blob.
Substituting this into the identity \eqref{boundaryPump}, we directly relate the energy pumped by the forces at the boundary with the blob's dipole moments of vorticity.

No forcing inside the flow is needed for this energy pumping; the energy flow starts at the boundary and propagates to numerous singular vorticity blobs, where it is finally dissipated. The distribution of these vorticity blobs is all we need for the turbulence theory. The forcing is required only as a boundary condition at infinity.

Another critical comment: with the velocity correlations \textbf{growing} with distance by the approximate K41 law, even the forcing at the remote boundary would influence the potential part of velocity in bulk. This boundary influence makes the energy cascade picture non-universal; it may depend upon the statistics of the random forcing.

Two asymptotic regimes manifesting this non-universality were observed for the energy spectrum $E(k, t)$: one for initial spectrum $E(k,0) \propto k^2$  and another for $E(k,0) \propto k^4 $. The potential velocity part differs for these regimes, as the second one adds a constant velocity to the Biot-Savart integral to cancel the total momentum $\int d^3 r \vec v$. 
In the general case, it will be a harmonic potential flow with certain boundary conditions at infinity, with explicit \textbf{continuous} dependence of the boundary forces. The most general case is the $k^2$ class, which does not require any restrictions.

Only the statistics of the rotational part of velocity, i.e., vorticity, could reach some universal regime independent of the boundary conditions at infinity. Certain discrete universality classes could exist as it is common in critical phenomena.

Unlike the potential part of velocity, vorticity is localized in singular regions—tubes and sheets, filling the space, as observed in numerical simulations.
The potential part of velocity drops in the loop equations, and the remaining stochastic motion of the velocity circulation is equivalent to the vorticity statistics.
Therefore, our solutions \cite{migdal2023exact} of the loop equations \cite{M93, M23PR} describe the internal stochastization of the decaying turbulence by a dual discrete system.

\subsection{Mathematical Introduction. The loop equation and its solution}\label{mathintro}

We derived a functional equation for the so-called loop average or Wilson loop in turbulence in the early nineties. All the references to our previous works can be found in a recent review paper \cite{M23PR}.

The path to an exact solution by a dimensional reduction in this equation was proposed in the 1993 paper but has just been explored (see Fig.\ref{fig::Dual-Turbulence}).
At the time, we could not compare a theory with anything but crude measurements in physical and numerical experiments at modest Reynolds numbers.
All these experiments agreed with the K41 scaling, so the exotic equation based on unjustified methods of quantum field theory was premature.
The specific prediction of the Loop equation, namely the Area law, could not be verified in DNS at the time with existing computer power.

The situation has changed over the last decades. No alternative microscopic theory based on the \NS{} equation emerged, but our understanding of the strong turbulence phenomena grew significantly. 
On the other hand, the loop equations technology in the gauge theory also advanced over the last decades. The correspondence between the loop space functionals and the original vector fields was better understood, and various solutions to the gauge loop equations were found. 
In particular, the momentum loop equation was developed, similar to our momentum loop used below \cite{MLDMig86, SeqQuanM95}.
Recently, some numerical methods were found to solve loop equations beyond perturbation theory \cite{LoopEqBootstrap, KazakovLooqBootstrap, kazakovZheng2024}.
The loop dynamics was extended to quantum gravity, where it was used to study nonperturbative phenomena \cite{AshtekarVar, RovelliSmolin}.

All these old and new developments made loop equations a major nonperturbative approach to gauge field theory.
So, it is time to revive the hibernating theory of the loop equations in turbulence, where these equations are much simpler.
The latest DNS \cite{S19, S21, VortexGasCirculation20, QuantumCirculation21} with Reynolds numbers of tens of thousands revealed and quantified violations of the K41 scaling laws. These numerical experiments are in agreement with so-called multifractal scaling laws \cite{FP85}. 

Theoretically, we studied the loop equation in the confinement region (large circulation over large loop $C$) and justified the Area law suggested in '93 on heuristic arguments.
This law says that the tails of velocity circulation PDF in the confinement region are functions of the minimal area inside this loop.
 It was verified in DNS a few years ago \cite{S19}, which triggered the further development of the geometric theory of turbulence\cite{ VortexGasCirculation20, QuantumCirculation21, M23PR}.
 In particular, the Area law was justified for flat and quadratic minimal surfaces, and an exact scaling law in confinement region $\Gamma \propto \sqrt{Area}$ was derived  \cite{M23PR}. The area law was verified with better precision in \cite{S21}.

In the previous paper, \cite{migdal2023exact}, we have found a family of exact solutions of the loop equation for decaying turbulence \cite{M93, M23PR}. This family describes a \textbf{fixed trajectory} of solutions with the universal time decay factor.
The solutions are formulated in terms of the Wilson loop or loop average
\begin{eqnarray}\label{WilsonLoop}
    &&\Psi(\gamma,C) = \VEV{\exp{ \frac{\I \gamma}{\nu}\oint d \vec{C}(\theta) \cdot \vec v\left(\vec{C}(\theta)\right) }}_{init};\\
    \label{LoopSol}
    && \Psi(\gamma,C) \Rightarrow  \VEV{\exp{
	   \frac{\I \gamma}{\nu}\oint d \vec{C}(\theta) \cdot \vec{P}(\theta)}}_{\mathbb E};
\end{eqnarray}
In the first equation (the definition), the averaging $\VEV{\dots}$ goes over initial data for the solutions of the \NS{} equation for velocity field $\vec v(\vec{r})$. In the second one (the solution), the averaging goes over the space of solutions $\vec{P}(\theta)$ of the loop equation \cite{migdal2023exact}.
We choose in this paper the parametrization of the loop with $\xi = \frac{\theta}{2 \pi}$ to match with the fermionic coordinates below (the parametrization is arbitrary, in virtue of parametric invariance of the loop dynamics).

The loop equation for the momentum loop $\vec{P}(\theta)$ follows from the \NS{} equation for $\vec v$
\begin{eqnarray}
    &&\partial_t \val = \nu \partial_{\beta} \omega_{\beta \alpha} - v_{\beta}\omega_{\beta \alpha} - \dal \left( p + \frac{\vbe^2}{2} \right);\\
    && \dal \val =0;\\
    &&\omega_{\beta \alpha} = \dbe \val - \dal \vbe
\end{eqnarray}
After  some transformations, replacing velocity and vorticity with the functional derivatives of the loop functional, we found the following momentum loop equation in \cite{M23PR,migdal2023exact}
    \begin{eqnarray}\label{PloopEq}
  && \nu\partial_t \vec{P} =  - \gamma^2(\Delta \vec{P})^2 \vec{P}  + \nonumber\\
  &&\Delta \vec{P} \left(\gamma^2 \vec{P} \cdot \Delta \vec{P} +\I \gamma \left( \frac{(\vec{P} \cdot \Delta \vec{P})^2}{\Delta \vec{P}^2}- \vec{P}^2\right)\right);\\
  &&  \vec{P}(\theta) \equiv \frac{ \vec{P}(\theta^+) + \vec{P}(\theta^-)}{2};\\
  && \Delta \vec{P}(\theta) \equiv  \vec{P}(\theta^+) - \vec{P}(\theta^-);
\end{eqnarray}

The momentum loop has a discontinuity $\Delta \vec{P}(\theta)$ at every parameter $0<\theta \le 1$, making it a fractal curve in complex space $\mathbb{C}_d$. The details can be found in \cite{M23PR, migdal2023exact}. We will skip the arguments $t, \theta$ in these loop equations, as there is no explicit dependence of these equations on either of these variables.
This Anzatz \eqref{LoopSol} represents a plane wave in loop space, solving the loop equation for the Wilson loop due to the lack of direct dependence of the loop operator on the shape of the loop.

The superposition of these plane wave solutions would solve the \textbf{Cauchy problem in loop space}: find the stochastic function $\vec{P}(\theta)$ at $t = 0$, providing the initial velocity field distribution.
Formally, the initial distribution $W_0[P]$ of the momentum field $\vec{P}(\theta)$ is given by inverse functional Fourier transform.
\begin{eqnarray}
   && W_0[P] =\int D C \delta^3(\vec{C}[0]) \Psi(\gamma,C)_{t=0} \nonumber\\
   &&\exp{-\frac{\imath \gamma}{\nu}\int d \vec{C}(\theta) \cdot \vec{P}(\theta)}
\end{eqnarray}

In Appendix\ref{GlobalRot}, we solve the Cauchy problem for an exact stationary solution of the \NS{} equation corresponding to the global rotation with the Gaussian random rotation matrix. The stochastic function $\vec{P}(\theta)$ in \eqref{Pexp}, \eqref{Pcorr} has a nontrivial Gaussian distribution with discontinuity $\Delta P(\theta)$ related to slow $1/n$ decay of its Fourier expansion on the parametric circle. This is the simplest example of the fractal curve we study below in a decaying solution of the loop equation.

Rather than solving the Cauchy problem, we are looking for an attractor: the fixed trajectory for $\vec{P}(\theta, t)$ with some universal probability distribution related to the decaying turbulence statistics.

The following transformation reveals the hidden scaling invariance of decaying turbulence
\begin{eqnarray}\label{PFT}
    \vec{P} = \sqrt{\frac{\nu}{2(t+ t_0)}} \frac{ \vec{F}}{\gamma}
\end{eqnarray}
The new vector function $\vec{F}$ satisfies an equation
    \begin{eqnarray}\label{Fequation}
   &&2\partial_\tau \vec{F} = \left(1- (\Delta \vec{F})^2\right) \vec{F}  +\nonumber\\
   &&\Delta \vec{F} \left(\gamma^2 \vec{F} \cdot \Delta \vec{F} +\I \gamma \left( \frac{(\vec{F} \cdot \Delta \vec{F})^2}{\Delta \vec{F}^2}- \vec{F}^2\right)\right);\\
   && \tau = \log (t + t_0)
\end{eqnarray}

This equation is invariant under translations of the new variable $\tau = \log (t + t_0)$, corresponding to the rescaling/translation of the original time.
\begin{eqnarray}
    t \Rightarrow \lambda t + (\lambda-1) t_0
\end{eqnarray}
There are two consequences of this invariance.
\begin{itemize}
    \item There is a fixed point for $\vec{F}$.
    \item The approach to this fixed point is exponential in $\tau$, which is power-like in original time.
\end{itemize}
Both of these properties were used in \cite{migdal2023exact}: the first one was used to find a fixed point, and the second one was used to derive the spectral equation for the anomalous dimensions $\lambda_i$ of decay $t^{-\lambda_i}$ of the small deviations from these fixed points.
In this paper, we only consider the fixed point, leaving the exciting problem of the spectrum of anomalous dimensions for future research.

\subsection{The big and small Euler ensembles}
Let us remember the basic properties of the fixed point for $\vec{F}$ in \cite{migdal2023exact}.
It is defined as a limit $N \to \infty$ of the polygon $\vec{F}_0\dots \vec{F}_N= \vec{F}_0$ with the following vertices
\begin{eqnarray}\label{Fsol}
    && \vec{F}_k = \frac{\left\{\cos (\alpha_k), \sin (\alpha_k), i \cos \left(\frac{\beta }{2}\right)\right\}}{2 \sin \left(\frac{\beta }{2}\right)} ;\\
    && \theta_k = \frac{k}{N}; \; \beta = \frac{2 \pi p}{q};\; N \to \infty;\\
    &&\alpha_{k+1} = \alpha_k + \sigma_k \beta;\; \sigma_k =\pm 1,\;\beta \sum\sigma_k = 2 \pi p r;
\end{eqnarray}
The parameters $ \hat{\Omega},p,q,r,\sigma_0\dots\sigma_{N}= \sigma_0$ are random, making this solution for $\vec{F}(\theta)$ a fixed manifold rather than a fixed point.
We suggested calling this manifold the big Euler ensemble of just the Euler ensemble.

It is a fixed point of \eqref{Fequation} with the discrete version of discontinuity and principal value:
\begin{eqnarray}
    &&\Delta \vec{F} \equiv \vec{F}_{k+1} - \vec{F}_k;\\
    && \vec{F} \equiv \frac{ \vec{F}_{k+1} + \vec{F}_{k}}{2}
\end{eqnarray}
Both terms of the right side \eqref{Fequation} vanish; the term proportional to $\Delta \vec{F}$ and the term proportional to $\vec{F}$. Otherwise, we would have $\vec{F} \parallel \Delta \vec{F}$, leading to zero vorticity \cite{migdal2023exact}.
The ensemble of all the different solutions is called the big Euler ensemble. The integer numbers $\sigma_k= \pm1$ came as the solution of the loop equation, and the requirement of the rational $\frac{p}{q}$ came from the periodicity requirement.

We can use integration (summation) by parts to write the circulation as follows (in virtue of periodicity):
\begin{eqnarray}
\label{CircByParts}
    && \oint d \vec{C}(\theta) \cdot \vec{P}(\theta) = -\oint d \vec{P}(\theta) \cdot \vec{C}(\theta) ;\\
    && \sum_k \Delta C_k \vec{P}_k = -\oint \Delta \vec{P}_k \cdot \vec{C}_k;
\end{eqnarray}

A remarkable property of this solution $\vec{P}(\theta,t)$ of the loop equation is that even though it satisfies the complex equation and has an imaginary part, the resulting circulation \eqref{CircByParts} is real!
The imaginary part of the $\vec{P}(\theta,t)$ does not depend on $\theta$ and thus drops from the integral $\oint d\vec{P}(\theta,t)\cdot \vec{C}(\theta)$.

There is, in general, a larger manifold of periodic solutions to the discrete loop equation, which has all three components of $\vec{F}_k$ complex and varying along the polygon.

We could not find a global parametrization of such a solution\footnote{Nikita Nekrasov (private communication) suggested to me an algorithm of generating this solution as a set of adjacent triangles in complex 3-space and pointed out an invariant measure in phase space, made of lengths of shared sides and angles between them. Unfortunately, this beautiful construction does not guarantee real circulation, requiring further study.}. Instead, we generated it numerically by taking a planar closed polygon and evolving its vertices $\vec{F}_k$ by a stochastic process in the local tangent plane to the surface of the equations in multi-dimensional complex space.

We could not submit such a solution to an extra restriction $\Im \vec{F}_k = \const{}$ needed to make circulation real. 
We cannot prove that such a general solution does not exist but rather take the Euler ensemble as a working hypothesis and investigate its properties.

This ensemble can be solved analytically in the statistical limit and has nice physical properties, matching the expected behavior of the decaying turbulence solution.

We assign equal weights to all elements of this set; we call this conjecture the ergodic hypothesis.
This prescription is similar to assigning equal weights to each triangulation of curved space with the same topology in dynamically triangulated quantum gravity \cite{AgMig92}. Mathematically, this is the most symmetric weight assignment, and there are general expectations that various discrete theories converge into the same symmetry classes of continuum theories in the statistical limit.
This method works remarkably well in two dimensions \cite{gross1990nonperturbative,brezin1990exactly,douglas1990strings}, providing the same correlation functions as continuum gravity (Liouville theory \cite{KPZ}).

The fractions $\frac{p}{q}$ with fixed denominator are counted by Euler totient function $\varphi(q)$ \cite{HardyWright} 
\begin{eqnarray}\label{discretePQ}
   && \varphi(q) = \sum_{\substack{p=1 \\ (p,q)}}^{q-1} 1 = q \prod_{p|q}\left(1 - \frac{1}{p}\right) ;
\end{eqnarray}
For example $\varphi(16) = 8$ and $\varphi(17) = 16$.

In some cases, one can analytically average over spins $\sigma$ in the big Euler ensemble, reducing the problem to computations of averages over the small Euler ensemble $\mathcal{E}(N): N,p,q,r$ with the measure induced by averaging over the spins in the big Euler ensemble.

\section{The Fermi ring and its continuum limit}
In this paper, we perform this averaging over $\sigma$ analytically, without any approximations, reducing it to a partition function of a certain quantum mechanical system with Fermi particles. 
The Quantum Trace Theorem, establishing this connection \eqref{trace formula} is proven in Appendix \ref{MarkovFermi}.
This partition function is calculable using a WKB approximation in the statistical limit $N \to \infty$.
As we shall shortly see, in the continuum limit $N \to \infty$, the accumulated numbers of Fermi particles $\nu_k =1$ and Dirac holes $ \nu_k =0$ tend to some classical function $ \sum_{l < k}(2 \nu_l -1) \to \alpha(\xi)$ of "position" $\xi=\frac{k}{N}$, leading to the exact solution.

Specifically, as we prove in Appendix \ref{PathIntegral}, the loop functional in continuum limit $N \to \infty$ reduces to a quantum mechanical path integral \eqref{pathIntegralSol} over the Fermion density $\alpha(\xi)$. The effective Action for this path integral is given by circulation $\imath\Gamma$ expressed in terms of this density plus a quadratic functional corresponding to Brownian measure \eqref{gaussMeasure} for $\alpha(\xi)$.

Thus, there are three sources of fluctuations: There is a phase factor related to the circulation, there is a Brownian positive distribution of the trajectory $\alpha(\xi)$(Gaussian measure), and finally, the circulation depends on the random fraction $\frac{p}{q}$ distributed according to the small Euler ensemble.
The continuum limit of the latter distribution is derived in Appendix \ref{smallEuler}, using new cotangent sums derived in the previous paper \cite{migdal2023exact}.

This is a new kind of quantum mechanical system with a complex Action, reflecting the irreversibility of turbulence.
The square root of viscosity enters the denominator of the effective Action, like a coupling constant. As we argued in the previous paper, the turbulent limit of our theory corresponds to
\begin{eqnarray}
    \nu \to \frac{\tilde{\nu}}{N^2} \to 0
\end{eqnarray}
The parameter $\tilde{\nu}$ remains a free parameter of our theory, playing the role of turbulent viscosity.  In particular, there is an anomalous energy dissipation in this limit
\begin{eqnarray}
   && \mathcal E(t) = \frac{\tilde{\nu}}{t^2} F(k_0^2 \tilde{\nu} t);\\
   && F(0) = \frac{\pi ^2 }{864 \zeta (3)};
\end{eqnarray}
Here $k_0$ is the lower cutoff of the energy spectrum (to be discussed below). The decaying spectrum at small wave vectors $k < k_0$ is related to the energy pumping at the initial moment $t = t_0$ and is time-independent.

The dimensionless parameter $N \to \infty$ plays the role of the Reynolds number. We derive the turbulent limit $N = 2 m \to \infty$ without any $1/N$ corrections. This solution will then apply to the inertial range of the physical turbulence in a system with a finite but large Reynolds number.
\section{Instanton in the path integral}
This classical equation for our path integral reads (with $\Omega \in O(3)$ being a random rotation matrix):
\begin{eqnarray}
\label{classEq}
   && \alpha'' =  - \imath \kappa \left(\mathcal C'_\Omega \exp{\imath \alpha} + (\mathcal C'_\Omega )^\star \exp{-\imath \alpha}\right) ;\\
   && \kappa = \frac{1}{2 \pi y \sqrt{X} \sqrt{2 \nu t }};\\
   && \mathcal C_\Omega(\theta) =  \vec{C}(\theta) \cdot \hat{\Omega} \cdot\{\imath,1,0\};
\end{eqnarray}
The parameter $\kappa$ is distributed according to the distributions \eqref{WprimeX}, \eqref{WprimeY} of the variables $X,y$ in a small Euler ensemble in the statistical limit.

This complex equation leads to a complex classical solution (instanton).
It simplifies for $ z = \exp{\imath \alpha}$:
\begin{eqnarray}
    &&z'' = \frac{(z')^2}{z} + \kappa \left(\mathcal C'_\Omega z^2 + (\mathcal C'_\Omega )^\star \right);\\
    && z(0) = z(1) =1
\end{eqnarray}
This equation cannot be analytically solved for arbitrary periodic function $C'_\Omega(\xi)$.

The weak and strong coupling expansions by $\kappa$ are straightforward.
At small $\kappa$
    \begin{eqnarray}
    &&z(\xi) \to 1 + 2 \kappa \left(- A \xi +\int_0^\xi \Re \mathcal C_\Omega(\xi') d \xi' \right) \nonumber\\
    &&+ O(\kappa^2);\\
    && A = \int_0^1 \Re \mathcal C_\Omega(\xi') d \xi'
\end{eqnarray}
At large $\kappa$
\begin{eqnarray}
    &&z(\xi) \to \imath \exp{ -\imath \arg \mathcal C'_\Omega(\xi)} = \imath \frac{\abs{\mathcal C'_\Omega(\xi)}}{\mathcal C'_\Omega(\xi)}
\end{eqnarray}
This solution is valid at intermediate $\xi$, not too close to the boundaries $\xi = (0,1)$.
In the region near the boundaries $\xi(1-\xi) \ll \frac{1}{\sqrt{\kappa}}$, the following asymptotic agrees with the classical  equation 
\begin{eqnarray}
\label{endptZ}
    &&z \to  1 \pm \imath\xi \sqrt{2 \kappa\Re C'_\Omega(0)} + O(\xi^2 );\\
     &&z \to  1 \pm \imath (1-\xi) \sqrt{ 2 \kappa\Re C'_\Omega(1)} + O((1-\xi)^2);
\end{eqnarray}
One can expand in small or large values of $\kappa$ and use the above distributions for $X, y$ term by term.

As it was noticed above, the viscosity $\nu = \frac{\tilde{\nu}}{N^2} \to 0$ in our theory. This limit makes $ \kappa \sim N \to \infty$, justifying the strong coupling limit for the Wilson loop solution.

The classical limit of the circulation in exponential of \eqref{pathIntegralSol}
\begin{eqnarray}
 \int_0^1 d \xi\Im\left( \mathcal C'_\Omega(\xi) \exp{\imath\alpha(\xi)} \right)\to \int_0^1 d \xi \abs{\mathcal C'_\Omega(\xi) }
\end{eqnarray}
becomes a positive definite function of the rotation matrix $\Omega$. At large $\kappa$ the leading contribution will come from the rotation matrix minimizing this functonal.

Let us think about the physical meaning of this finding. We have just found the density of our Fermi particles on a parametric circle
\begin{eqnarray}
     \alpha(\xi) = \frac{\pi}{2} - \arg \mathcal C'_\Omega(\xi)
\end{eqnarray}
This density does not fluctuate in a turbulent limit, except near the endpoints $\xi \to 0, \xi \to 1$. In the vicinity of the endpoints, there is a different asymptotic solution \eqref{endptZ} for $ \alpha \to (z-1)/\imath$.

Compute the Wilson loop for a specific loop, say, the circle, is an interesting problem, but there is a simpler quantity.
In the next section we are considering an important calculable case of the vorticity correlation function, where the full solution in quadratures is available. 
This function has been directly observed in grid turbulence experiments \cite{GridTurbulence_1966, AgMig92} more than half a century ago and is being studied in modern large-scale real and numerical experiments\cite{SreeniDecaying, GDSM24, GregXi2}.

This is the vorticity correlation function \cite{M23PR}, corresponding to the loop $C$ backtracking between two points in space $\vec{r}_1=0, \vec{r}_2=\vec{r}$, with the vorticity operators are inserted at these two points (see \cite{migdal2023exact} for details and the justification). The Fourier transform of this function describes the decaying energy spectrum, also measured experimentally.

In Appendixes \ref{VelCorr}, \ref{TurbVisc}, \ref{ClasTraj}, we express this correlation function as a particular case of the second variation of the loop functional. Then, in Appendixes \ref{FuncDet}, \ref{FlucTerm}, we compute the path integral in the leading WKB approximation. This is a one-dimensional version of the instanton computations, familiar to the gauge theory experts.
In the turbulent limit $N\to \infty, \tilde{\nu} = \const{}$, there are no higher order corrections to this WKB approximation of the path integral.

\section{The decaying energy in finite system}\label{finitesystem}

The vorticity correlation in Fourier space doubles as an energy spectrum
\begin{eqnarray}
\label{spectrum}
    E(k,t) = 4\pi\vec{k}^2 \VEV{ \vec v \cdot \vec v_k} = 4\pi\VEV{ \vec \omega \cdot \vec \omega_k}
\end{eqnarray}

The energy spectrum in a finite system with size $L$ is bounded from below. At low $|\vec{k}| \le \pi/L$, the spectrum is no longer related to the turbulence but is given by the energy pumping by external forces at the boundaries.

This energy pumping \cite{SreeniDecaying, GDSM24} takes place at $ t < t_0$, after which the pumping stops. At this moment, the energy spectrum is growing with wavevector by one of two possible laws (with $P$ being the Birkhoff-Saffman invariant of the fluid and $M$ being the Loitzansky invariant)
\begin{eqnarray}
    \begin{cases}
        &E(k,t)   \propto  P k^2\\
        &E(k,t)   \propto  M k^4\\
    \end{cases}
\end{eqnarray}
The small $k$ limit of the spectrum is time-independent as both $P$ and $M$ do not depend of time.

At $t> t_0$, without the forcing, the pumped energy dissipates at large $k$ corresponding to smaller spatial scales of the hierarchy of vortex structures of all scales, ending with dissipative scales, or wavevectors $k \gg \pi/L$. After sufficient time, the universal regime kicks in, corresponding to the decaying turbulence. It is implied that a large amount of energy was pumped in, so it takes a long time to reach this decaying regime, corresponding to some fixed trajectory.

Our solution does not impose any restrictions on the SB invariant $P$ and thus applies to the most general, first regime with $k^2$ spectrum at small $k$ and some universal decay at large $k$, reflecting these distributed vortex structures. 

The decaying energy $E_d(t)$, given by the part of the spectrum $k > k_0 \sim 1/L$, has the following form
\begin{eqnarray}
\label{Edecay}
    &&E_d(t)=  \int_{k_0}^\infty d k E(k,t) \propto\frac{4\pi \tilde{\nu}}{t} \int_{k_0 \sqrt{\tilde{\nu} t}}^\infty H(\kappa) d \kappa
\end{eqnarray}

On top of the trivial decrease $\frac{\tilde{\nu}}{t}$, as prescribed by dimensional counting in an infinite system, there is some extra decrease related to the increase of the lower limit.
\subsection{Computation of the energy dissipation}
The energy in our theory does not have a finite statistical limit because of the contribution from the unknown potential part of velocity. This contribution could be infinite in the infinite system.
Thus, we compute the energy by integrating its time derivative (i.e. the dissipation rate $-\mathcal E(t)$ ) given the zero energy left at infinite time.
\begin{eqnarray}
\label{Ediss}
    E(t) = \int_t^\infty \mathcal E(t') d t'
\end{eqnarray}

This energy dissipation rate $\mathcal E(t')$ is calculable
    \begin{eqnarray}
\label{dissipation}
    &&\mathcal E(t) \propto 4 \pi \nu \int d k \,k^2 \VEV{\vec\omega(\vec 0) \cdot \vec \omega(\vec{k})} = \nonumber\\
    &&4 \pi\frac{\tilde{\nu} }{t^2}\int_{k_0 \sqrt{\tilde{\nu} t}}^\infty \kappa^2 H(\kappa) d \kappa
\end{eqnarray}

In our theory, this integral has a finite limit in an infinite system ($k_0=0$). 

This limit was computed in \cite{migdal2023exact} in a slightly different grand canonical ensemble, where $N$ was fluctuating with the weight $\exp{-\mu N}, \mu \to 0$.

With our current ensemble of fixed even $N\to \infty$ the results of \cite{migdal2023exact} read:
\begin{eqnarray}
\label{exactEnstrophy}
   && \mathcal E_{\infty}(t) = \frac{\tilde{\nu}}{t^2} \frac{\pi ^2 }{864 \zeta (3)};
\end{eqnarray}
In our present theory, the same quantity is given by the above integral at $k_0 =0$
\begin{eqnarray}
    \mathcal E_{\infty}(t) = 4 \pi \frac{\tilde{\nu} }{t^2}\int_{0}^\infty \kappa^2 H(\kappa) d \kappa
\end{eqnarray}
Comparing these two expressions, we get the normalization  of $H(\kappa)$
\begin{eqnarray}
    4 \pi \int_{0}^\infty \kappa^2 H(\kappa) d \kappa = \frac{\pi ^2 }{864 \zeta (3)} 
\end{eqnarray}
The integral on the left can be further reduced \cite{MB41} to the following normalization condition:
    \begin{eqnarray}
  && \mathcal Z =\frac{276480 \zeta (3) \zeta \left(\frac{15}{2}\right)  }{119 \zeta \left(\frac{17}{2}\right)}\int_{\Delta_1}^{\Delta_2} d \Delta  \frac{P(\Delta)}{Q(\Delta)};\\
  && P(\Delta) = (1- \Delta)\sqrt{r_0(\Delta )} S(\Delta ) Q_{\alpha }(\Delta ,1) \nonumber\\
  &&\left(\pi  J(\Delta ) \left(r_0(\Delta )+12\right)+L(\Delta ) \left(r_0(\Delta )-6\right)\right);\\
  && Q(\Delta) = L(\Delta )^2 \left(r_0(\Delta )+12\right);
\end{eqnarray}
This normalization constant $\mathcal Z$ can be used in equation \eqref{Edecay} for the energy decay in a \textbf{finite} system. All the functions of $\Delta$ were defined above.

As for the energy spectrum, this is not an independent function in our theory. Comparing the two expressions \eqref{spectrum} and \eqref{dissipation}, we arrive at the following relation
\begin{subequations}
\label{EbyF}
\begin{eqnarray}
\label{EbyH}
   && t^2 \mathcal E(t) = 4\pi \tilde{\nu} F\left(k_0 \sqrt{\tilde{\nu} t}\right);\\
   && F(\kappa) = \int\displaylimits_{\kappa}^\infty H(x) x^2 \, d x,\\
   \label{SbyH}
   && \sqrt{t} E(k,t) \propto 4 \pi \tilde{\nu} \sqrt{\tilde{\nu}}H\left(k \sqrt{\tilde{\nu} t}\right);
\end{eqnarray}
\end{subequations}
Both the energy dissipation and the energy spectrum are related to the same function $H(\kappa)$, but the energy spectrum related to this function at large argument $ \kappa = k \sqrt{\tilde{\nu} t}$, whereas the energy dissipation is related to integral of this function from small argument $ \kappa = k_0 \sqrt{\tilde{\nu} t}$ to infinity.

Our theory does not have the Saffmann part of the spectra; it only applies to an infinite system described by the region $k > k_0$. At the boundary of this region, we have the value of the spectrum (assuming $k_0 \sqrt{\tilde{\nu} t} \ll 1$)
\begin{eqnarray}
    E(k_0, t) \propto \frac{4 \pi}{\sqrt{t}} \tilde{\nu} \sqrt{\tilde{\nu}}H\left(k_0 \sqrt{\tilde{\nu} t}\right) \sim \frac{4 \pi H(0)\tilde{\nu} \sqrt{\tilde{\nu}}}{\sqrt{t}}
\end{eqnarray}
This will match the Saffman spectrum $ P k^2$ at small $k$
\begin{eqnarray}
    k\sim \tilde{\nu}^{\sfrac{3}{4}}t^{-\oq} P^{-\oh}
\end{eqnarray}
This boundary value decays as $t^{-\oq}$; it is below the value $k_0 \sim 1/L$, which is time-independent.  We are not considering extremely large times such that $k_0 \sqrt{\tilde{\nu} t} \gg 1$. At these times, the decay is over: there is not enough energy left for turbulence, and the whole pumped energy has dissipated.

We computed the universal function $H(\kappa)$ in the Appendix \ref{RiemanZeros} and numerically integrated it to obtain $F(\kappa), \kappa = k \sqrt{\tilde{\nu} t}$. Asymptotically, at large $k \sqrt{\tilde{\nu} t}$
\begin{eqnarray}
   && E(k, t) \propto \frac{\tilde{\nu} \sqrt{\tilde{\nu}}}{\sqrt{t}} \left(k \sqrt{\tilde{\nu} t}\right)^{- \sfrac{7}{2}};\\
   && \mathcal E(t) \propto t^{-\sfrac{9}{4}};\\
   && E(t) = \int_t^\infty \mathcal E(t') \, d t' \propto t^{-\sfrac{5}{4}}
\end{eqnarray}

We computed the effective critical indexes
\begin{eqnarray}
    &&n( t) = -t\partial_t \log E  = -\frac{ t \mathcal E(t)}{E(t)};\\
    &&s(\kappa) = -\kappa\partial_\kappa \log H(\kappa)
\end{eqnarray}
numerically in the Appendix \ref{RiemanZeros}, using \Mathematica{}.
The accuracy is just 4-5 digits, but it can be easily improved by taking more CPU time once experimental data gets more precise. 
These curves are universal, and they change the regime before approaching their limits from below $n(\infty) = \frac{5}{4}, s(\infty) = \frac{7}{2}$.
These regime changes are due to the quantum effects (complex zeros of the zeta function contributing to the energy spectrum's Mellin transform, as shown in Appendix \ref{RiemanZeros}).

The experimental data \cite{GridTurbulence_1966, Comte_Bellot_Corrsin_1971} yields $n(\infty) \approx 1.25\pm 0.02$, which agrees with our theoretical prediction in Fig.\ref{fig::NPlot}.
\pctPDF{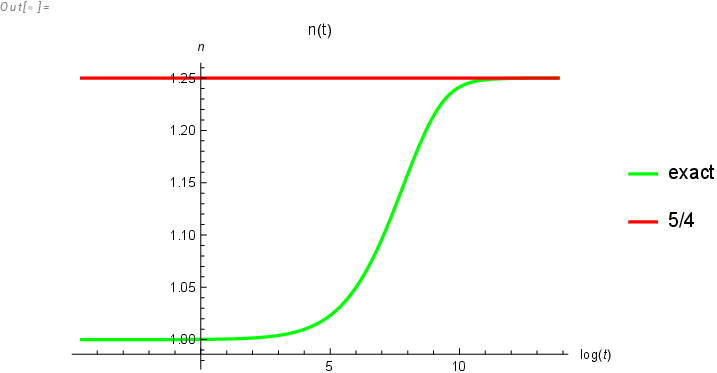}{The effective index $n(t) = -\frac{t \mathcal E(t)}{E(t)}$ compared with asymptotic value $n(\infty) = \sfrac{5}{4}$.}
Our universal curves for $n(t), s(\kappa)$ were computed directly from the analytic solution of the loop equation in the turbulent limit without any fitting parameters. 
It will be very interesting to compare these curves with more precise experiments (real or numerical).
\subsection{The energy normalization problem}\label{Enorm}
The above formulas do not specify the energy spectrum's normalization, just the energy dissipation's normalization. When one tries to recover the normalization of the energy spectrum, the following problem arises.

The normalization of the decaying energy \eqref{Edecay}  seems incompatible with its time derivative \eqref{dissipation}. In conventional turbulence models, the integral for dissipation diverges at large wavelengths, reflecting the singular vortex structures such as the Burgers vortex filaments \cite{BurgersVortex} with viscous thickness $w \sim \sqrt{\nu}$. Integrating the square of vorticity in the Burgers vortex in the transverse plane, we get a large factor $ 1/\nu$; this compensation leads to anomalous dissipation $\mathcal E = \nu \VEV{\vec \omega^2}$. independent of $\nu$.

In our dual theory, this factor of $\nu$ is compensated by a different mechanism: the vorticity is represented as a discontinuity at the curve $\vec{P}(\theta)$ in our solution: $\vec \omega \propto \vec{P} \times \Delta \vec{P}$. Summing over a large number $N \to \infty$ of these discontinuities leads to the factor $N^2$ compensating small viscosity $\nu$ factor in front of the enstrophy. The enstrophy integral $\int d k k^2 \VEV{\vec \omega \cdot \vec \omega_k}$ converges at large $k$.

As a consequence, the vorticity correlation $\VEV{\vec \omega \cdot \vec \omega_k} \sim 1/\nu$ is \textbf{large} at all $k$, not just at large $k$. Our limit $\nu \to 0, N \to \infty$ applies to the computation of integrated quantities such as total decaying energy but not to the energy spectrum as a function of wavelength.

The problem boils down to the following. The turbulent limit differs from the inviscid limit of the \NS{} equation. In the turbulent limit, the average circulation $\Gamma$ is much larger than viscosity, but the dimensional scales, determined by viscosity, stay finite.

To be more specific, the enstrophy in our system has the structure
\begin{eqnarray}
    \VEV{\vec \omega^2} = \frac{\mathcal E(t, |\Gamma|/\nu)}{\nu}
\end{eqnarray}
The dissipation $\mathcal E(t, |\Gamma|/\nu)$ depends on time and the effective Reynolds number $\RE =|\Gamma|/\nu $ where $\Gamma$ is a typical velocity circulation in our problem.

The turbulent limit corresponds to $|\Gamma| \gg \nu $, rather than $\nu \to 0$.
Mathematically, we are looking for a residue of the enstrophy at zero viscosity
\begin{eqnarray}
    \VEV{\vec \omega^2} \to \frac{\mathcal E(t, \infty)}{\nu}
\end{eqnarray}
However, in the real physical world (or in the DNS), we shall use this formula for finite $\nu$ corresponding to the actual viscosity of water. We only use our dual theory to compute this residue $ \mathop{\mathrm{Res}}_{\nu=0} \VEV{\vec \omega^2} = \mathcal E(t, \infty)$.
\begin{eqnarray}\label{Residue}
    \VEV{\vec \omega^2} \to \frac{\mathcal E(t, \infty)_{dual}}{\nu}
\end{eqnarray}
In particular, in the infinite system ($k_0=0$), we have the value \eqref{exactEnstrophy} computed in the previous paper.

This comparison restores the normalization of the vorticity correlation and the energy spectrum
\begin{eqnarray}\label{EnergyResidue}
   && \VEV{\vec\omega(\vec 0) \cdot \vec \omega(\vec{k})} = \frac{\tilde{\nu}^{\sfrac{5}{2}}H\left(k \sqrt{\tilde{\nu} t}\right)}{\nu \sqrt{t}};\\
   && \sqrt{t} E(k,t) = \frac{4 \pi \tilde{\nu}^{\sfrac{5}{2}}H\left(k \sqrt{\tilde{\nu} t}\right)}{\nu}
\end{eqnarray}
 where $\nu$ is the physical viscosity of the fluid under consideration (say, water) and $\tilde{\nu}$ is the auxiliary scale parameter that comes from the dual theory in the turbulent limit.  

 Our theory has two unknown parameters: $\tilde{\nu} $ and $t_0$.
 The energy spectrum decreased faster than $k^{-3}$ so that the enstrophy integral $\int H(\kappa)\kappa^2 d \kappa$ converges at large $\kappa$ and so does the energy integral $\int H(\kappa) d \kappa$.

As we shall see in the next section, this fast decay of the energy spectrum also happens in DNS, in qualitative agreement with our decay $\kappa^{-\sfrac{7}{2}}$.

\section{Comparing our theory with the DNS}\label{DNS}
As the first such test of our theory, we took the raw DNS data from \cite{SreeniDecaying}, provided to us by the authors. This data is now available online \cite{GDSM24} and can be downloaded without permission. We only compared the data corresponding to our $k^2$ case and restricted ourselves to four samples with the largest grid $1024^3$, labeled as sample $13,14,15,16$.

First, we verify the decay of the Reynolds number for each sample, as our theory only applies to the large Reynolds numbers. These plots are shown in Fig. \ref{fig::ReynoldsPlot}. As we see from these plots, all four Reynolds numbers are modest, but the sample $14$ stands out as the closest to the strong turbulence we seek.

The next test is the effective length scale, which we define as
\begin{eqnarray}
\label{Ldef}
    &&L(t) = \frac{\int E(k,t) k d k}{\int E(k,t) k^2 d k};
\end{eqnarray}
The effective length $L(t)$ as a function of time is shown in Fig.\ref{fig::EffectiveLength}. 
\pctPDF{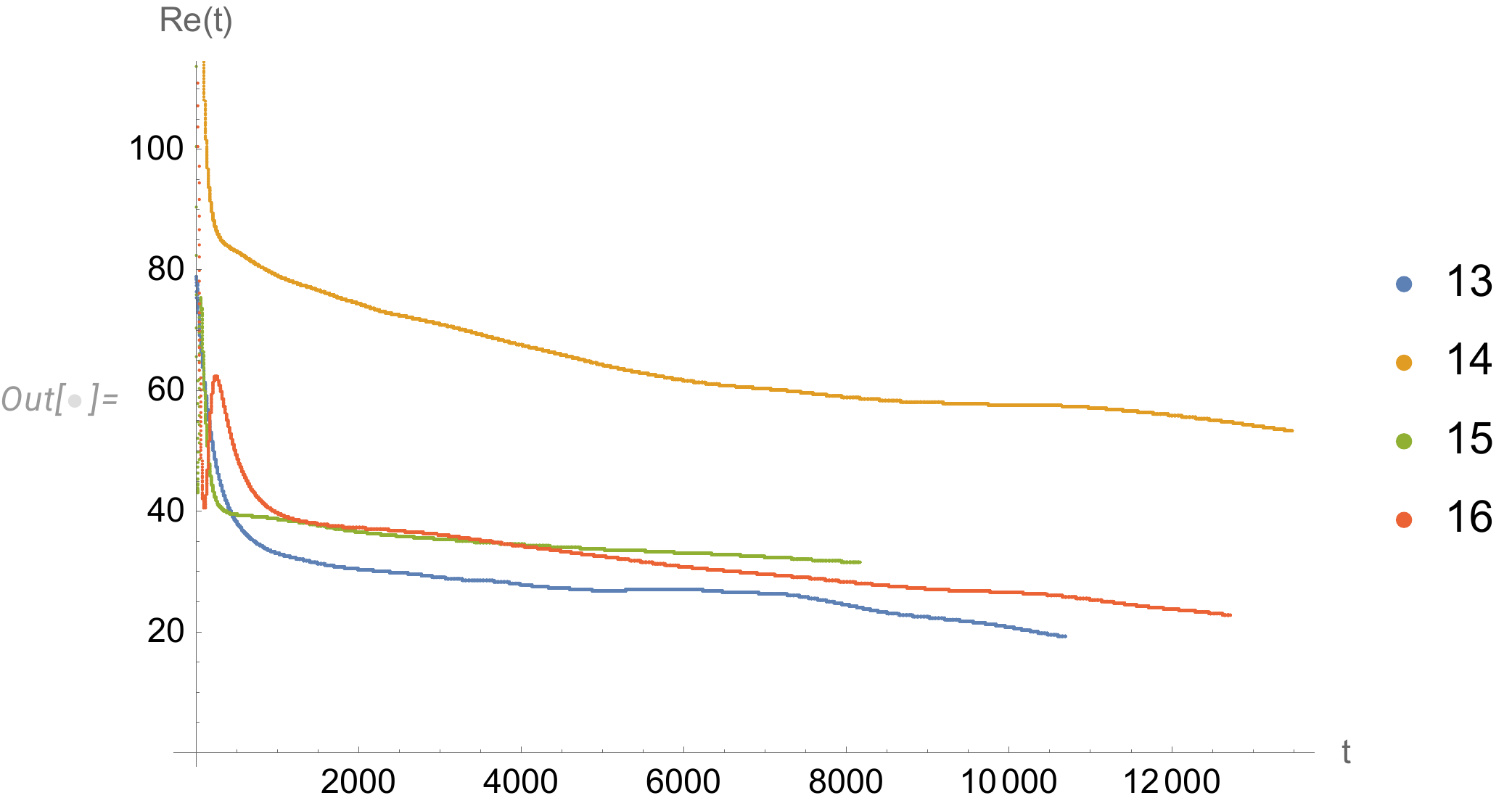}{The time decay of Reynolds numbers for each of the four samples from \cite{SreeniDecaying, GDSM24}}
\pctPDF{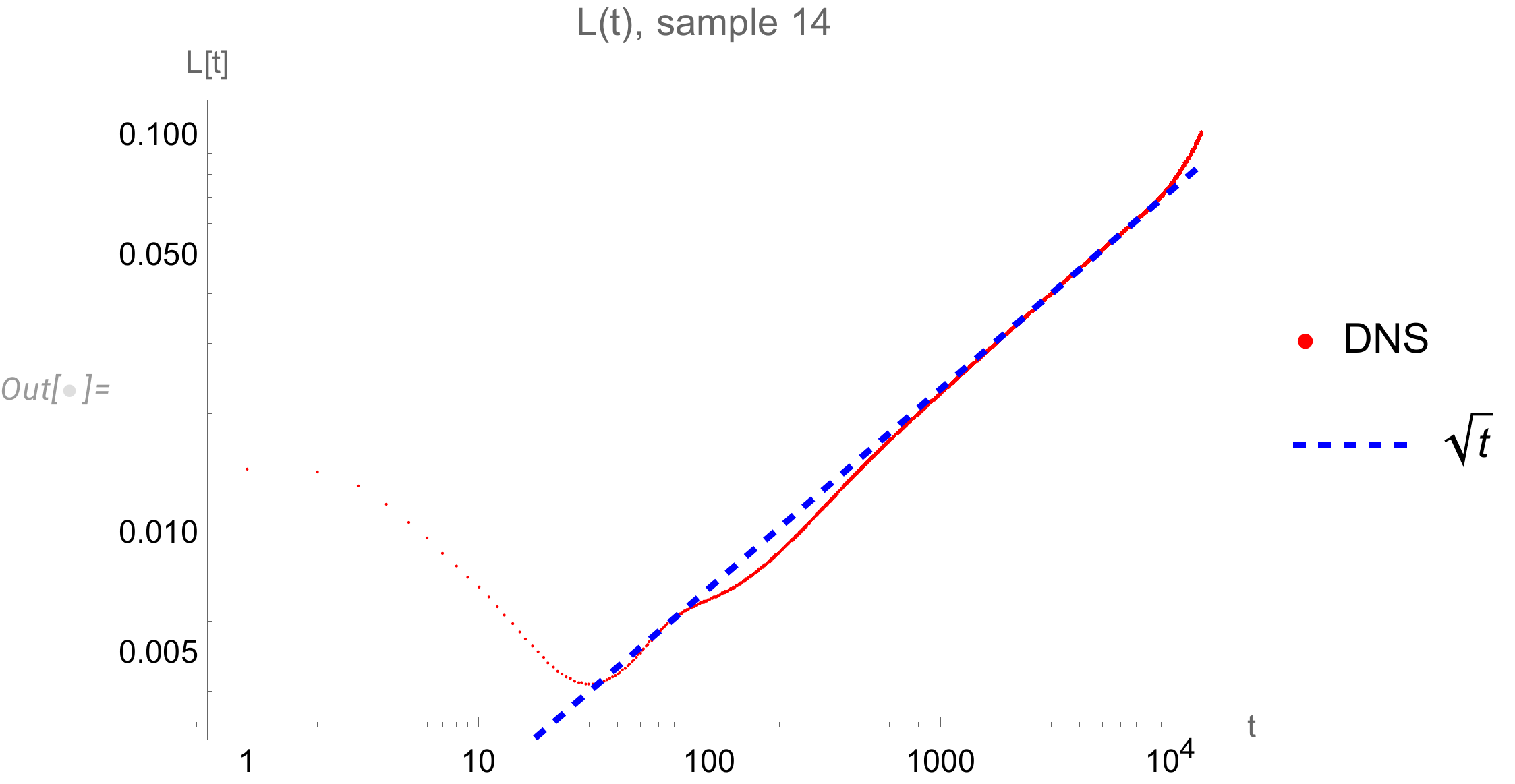}{The log-log plot of the effective length $L(t)$ for the sample 14 from \cite{SreeniDecaying, GDSM24}. The turbulent part, $1000< t < 8000$, closely fits our theory $ L(t) \propto \sqrt{t}$.}
The statistical equilibrium was not yet reached at $ t < 1000$, so we discarded this period.
The late stages of decay where $L(t) \propto t$ correspond to low Reynolds number and do not agree with our theory: we interpret it as the non-turbulent stage of decay when the remaining energy is insufficient for the strong turbulence phase.

The next test is the energy decay curve $E(t)$ for the turbulent region of sample 14. 
\begin{eqnarray}
\label{Edef}
    && E(t) = \int E(k,t)  d k;\\
    &&\log E(t) \approx a + f(b +\log L(t))
\end{eqnarray}
Our theory has two free parameters $t_0, \tilde{\nu}$. The fitting of $t_0$ is not trivial as we do not know which time range corresponds to the universal regime of the decay but still contains enough energy left for the strong turbulence.
Following the suggestion of \cite{GregDecayTurb}, we avoid fitting $t_0$ by investigating the decaying energy as a function of the decay length $L(t)$, which in our case scales as $\sqrt{\tilde{\nu} (t+ t_0)}$. The precise definition was given in \eqref{Ldef}.

The parameters $a,b$ were fitted by nonlinear regression using "NonlinearModelFit" in \Mathematica{}.
The resulting log-log plot is shown in(Fig. \ref{fig::EVsLPlot}). 
The fit is perfect, with less than one percent of the standard deviation. This relation is approximately linear with the slope $ - \frac{5}{2}$.
Note also that the asymptotic index $-\frac{5}{2}$ comes as a ratio of the energy decay index $-\frac{5}{4}$ to the index $m = \oh$ for $L(t)$ which we already tested.
 \pctPDF{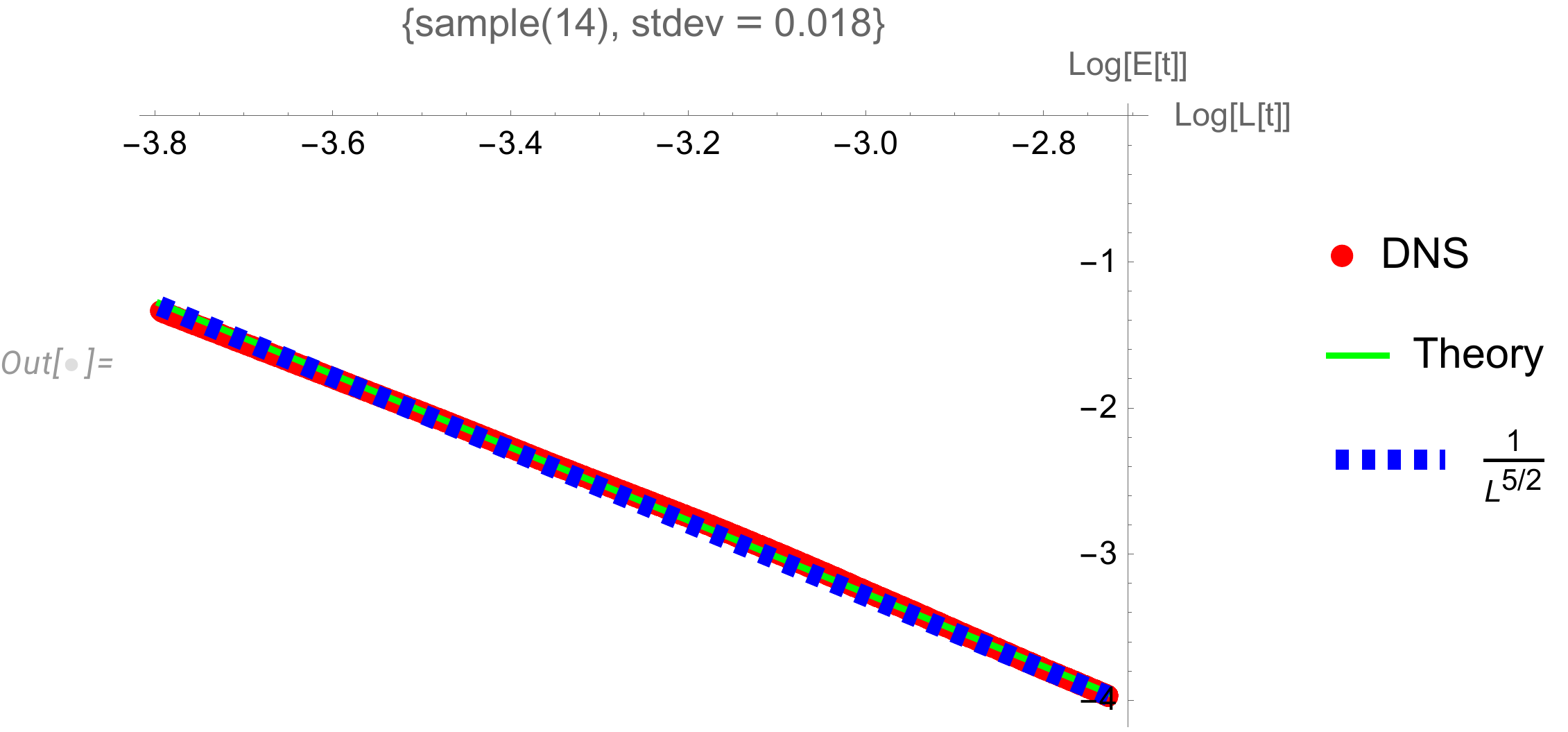}{The log-log plot of the decaying energy as a function of decaying length scale $L(t)$. Red dots are the DNS data, the green curve is an exact theoretical curve, and the dashed blue line is its asymptotic limit $E(t) \propto L(t)^{-\sfrac{5}{2}}$ corresponding to $E \propto t^{-\sfrac{5}{4}}$.}
Let us now turn to the energy spectrum. The data is not as good here as the energy decay data for two reasons. 
First, each wave vector component has only $L=512$ independent values on a $1024$ grid.

The energy spectrum is a function of the length of the wavevector, which is taking $O(L(L+1)(L+2)/6)$ different values between $0$ and $L\frac{\sqrt{3}}{2}$.
Unfortunately, the available data \cite{SreeniDecaying, GDSM24} aggregates the statistics at $512$ equidistant bins in $|\vec{k}|$, reducing statistics. The large number of rank bins (by an equal number of data points in each bin) would give us much more information about the spectrum.

We have a scaling law $E(k,t) = \frac{H(k\sqrt{t+ t_0})}{\sqrt{t+ t_0}}$, which means that the two-dimensional array of the data for $E(k,t)$ must collapse at one-dimensional subset.

We already saw the consequence of that collapse in the scaling law for $L(t)$. However, the low $k$ part of the spectrum is discrete and corresponds to lattice artifacts. 

We found the following method to avoid choosing the range of discrete wavelengths or fitting any scale parameters. 

We consider the second moment of velocity, related to the energy spectrum by Fourier transform
\begin{eqnarray}
    &&\VEV{\Delta \vec v^2}(r) =\int \frac{d ^3 k}{4 \pi^3} \frac{\left(1-e^{\I \vec{k} \cdot \vec{r}} \right)E(k,t)}{4 \pi \vec{k}^2};\\
    && \VEV{\Delta \vec v^2}(r)  = \VEV{\left(\vec v(\vec{r}, t) - \vec v(0,t)\right)^2}
\end{eqnarray}

There is a sharper test, namely the effective index $\xi_2(r,t)$ defined as a log-log derivative of this second moment
\begin{subequations}\label{xi2}
    \begin{eqnarray}
    &&\xi_2(r,t) = r \partial_r \log \VEV{\Delta v^2}(r)= \nonumber\\
    &&\frac{\int_0^\infty d k \left(-\cos(k r) + \frac{\sin(k r)}{k r}\right) E(k,t)}{\int_0^\infty d k \left(1 - \frac{\sin(k r)}{k r}\right) E(k,t)};\\
    && \VEV{\xi_2(x L(t),t)}_t  = f(x);
\end{eqnarray}
\end{subequations}

This universal function $ f(x)$ is numerically well defined as the integration over the spectrum suppresses the noise related to lattice discretization unless the dimensionless coordinate $x$ is too large.

In the K41 theory, this index is $\frac{2}{3}$, so one would expect at least a plateau of nearly K41 values in the inertial range of space scales. This index would have a higher constant value in multifractal models, around $0.7$.
Fig.\ref{fig::PlotIndex14} shows what we found instead for the DNS data \cite{SreeniDecaying, GDSM24}.
\pctPDF{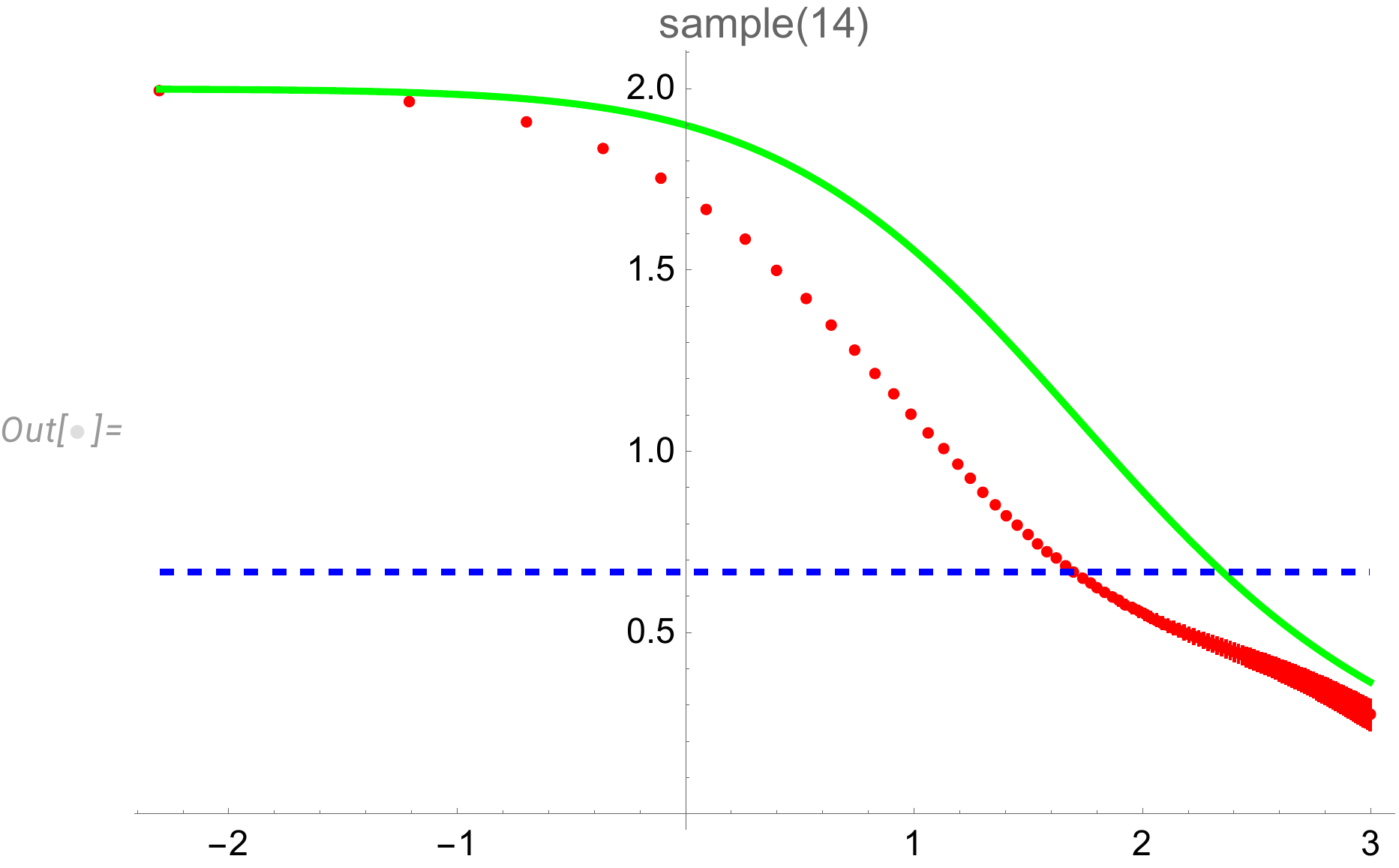}{The plot  of $f(x)$ in \eqref{xi2} as a function of $\log x$ (red dots with error bars). The green curve is the theoretical index without adjustment of the coordinate scale, and the dashed blue line is the K41 constant value $\frac{2}{3}$.}

This time, the theoretical curve (green line) goes outside the error bars, which means a lack of a global fit. Both curves are far from any constant value, so Kolmogorov and multifractal are out of the competition. However, the left parts of the theoretical and DNS curves, with $f > 0.5$, are almost parallel in $\log x$ scale. The larger values of x do not match so closely, and the error bars are bigger there, likely because of lattice artifacts in DNS. These left parts, however, can be moved on top of each other by properly choosing the length scale in our theory.

We tested this hypothesis by selecting the left side of this plot and adjusting the length scale in the theoretical line to get on top of the DNS data (shifting the green curve horizontally by a mean distance to the red curve).
Here is the result of this selection/shifting (see Fig.\ref{fig::MoveIndex14}).
\pctWPDF{0.45}{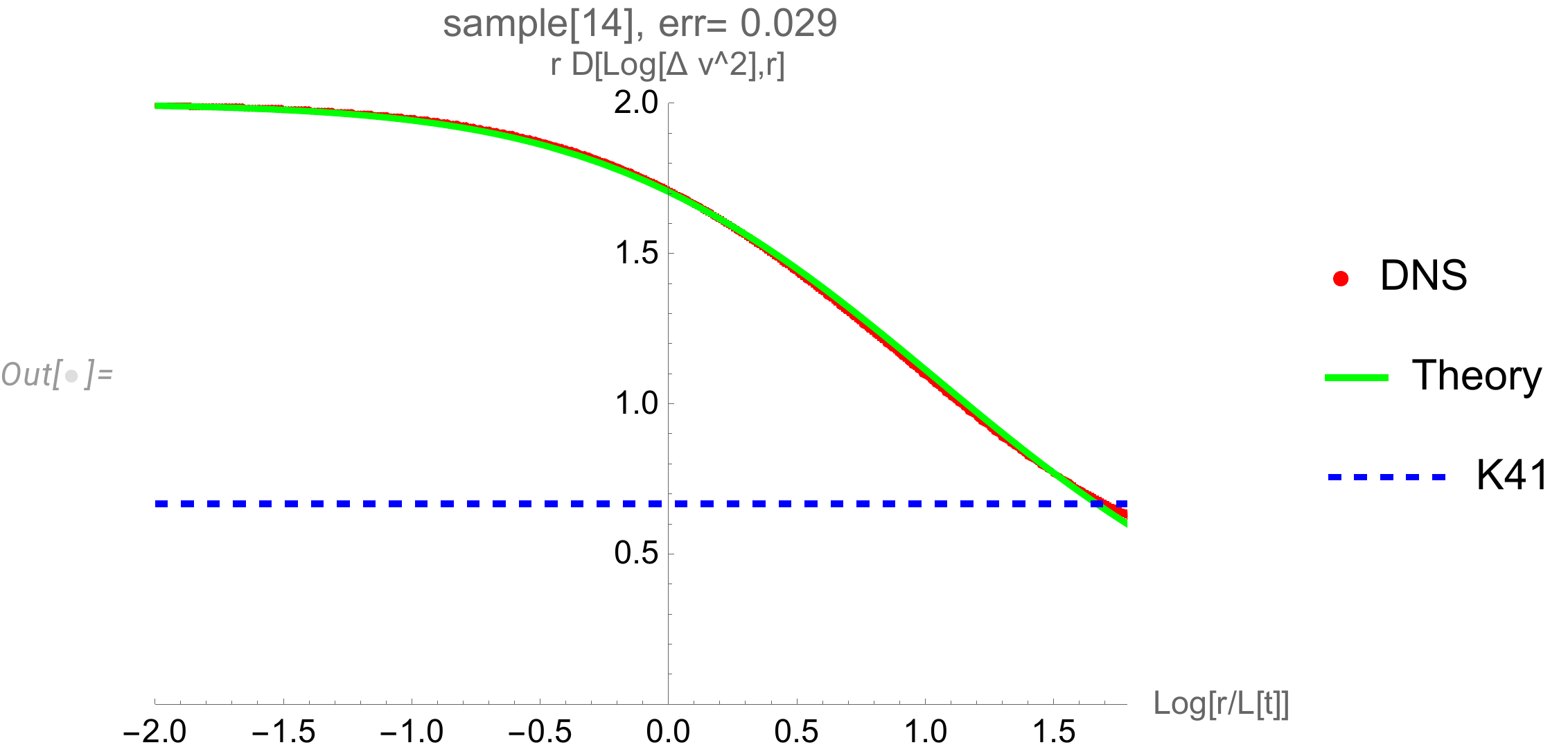}{The plot of the effective index $f(x)$ in \eqref{xi2} as a function of $\log x $ (red dots with error bars) in the turbulent range. 
The K41 $\frac{2}{3}$ law (blue dashed line) is totally off the charts.}

The theoretical curve (green line) is shifted left by some amount $\delta \log x =s$ to minimize the mean square of the horizontal distance in the turbulent range. This shift corresponds to the adjustment of the length scale in our theory. 
The curves now match up to the standard deviation of the DNS. The time limit $ T_0 < t < T_1$  of decaying turbulence was chosen to ensure the diffusion law $L(t) \propto \sqrt{t}$. 
These DNS completely rule out the K41 scaling law $\frac{2}{3}$. 

We recently encountered real experiments for compressible decaying turbulence in wind tunnels and atmosphere \cite{GregXi2}. Our theory assumes incompressibility, so it doesn't apply to these data in the air turbulence. Also, the magnitude of the errors in \cite{GregXi2} is unclear.

Nevertheless, we still compared our prediction for the second moment with the air data from \cite{GregXi2}. The shape of the experimental curve for $\VEV{\Delta v^2}(r)$ (Fig. 1 in \cite{GregXi2} ) is similar to ours. It significantly deviates from the K41 prediction: instead of a straight line in the log-log scale, there is a curved line with the slope varying from $2$ to $0$ as $r$ varies from zero to infinity, the same as our curves. There is no plateau at $\frac{2}{3}$ slope (Fig.1 in \cite{GregXi2}; the slope linearly decreases with $\log r$, similar to our decreasing slope in Fig. \ref{fig::PlotIndex14}.

However, the numerical values for the slope and curvature in \cite{GregXi2} are quite different from those we derived above from the incompressible DNS \cite{SreeniDecaying, GDSM24}, which fits our theory within the error bars. The origin of such a big discrepancy is unclear to us. 
Compressibility alone seems unlikely to explain such a large deviation from incompressible DNS. 

It is difficult to estimate the slope of the measured data numerically \cite{GregXi2} without increasing errors, which are already large enough in the data in \cite{SreeniDecaying, GDSM24}. Unlike numerical differentiation of the wind tunnel data, the Fourier integration of the DNS data suppresses the noise, which makes the DNS data for effective index $\xi_2(r)$ much more accurate.

\textbf{We conclude that our theory passed its first test with flying colors, but a more detailed comparison with new large DNS or experiments is desirable}.
\section{The spectrum of scaling dimensions}
\label{Analytic}
In addition to numerically computing the energy spectrum and plotting effective critical indexes, we can go one step further in the mathematical analysis of the concept of the scaling laws in turbulence.

In a scale-invariant theory, the Mellin transform of the correlation function in coordinate or momentum space is a meromorphic function.
The \textbf{spectrum of critical indexes} is given by the positions of poles of this function in a complex plane.
The whole spectrum of critical indexes is real in the theory of critical phenomena, described by a Conformal Field Theory (CFT).

As we shall prove below, our theory is scale-invariant by this definition, but it is not a CFT. In particular, some critical indexes come in complex conjugate pairs, reflecting the dissipative nature of our turbulence theory.

\subsection{The energy spectrum}
The spectral density $H(\kappa)$ is in \eqref{answer}. The Mellin transform is the following integral, assuming our function decreases at infinity.
\begin{eqnarray}
    &&H(\kappa) = \int_{-\eps - \I \infty}^{-\eps + \I \infty} \frac{d p}{2 \pi \I} h(p) \kappa^{p} ;\\
    && h(p) = \int_0^\infty \frac{d \kappa}{\kappa} \kappa^{-p} H(\kappa)
\end{eqnarray}
We change the sign of $p$ in conventional definition to better describe our functions.
\pctWPDF{0.5}{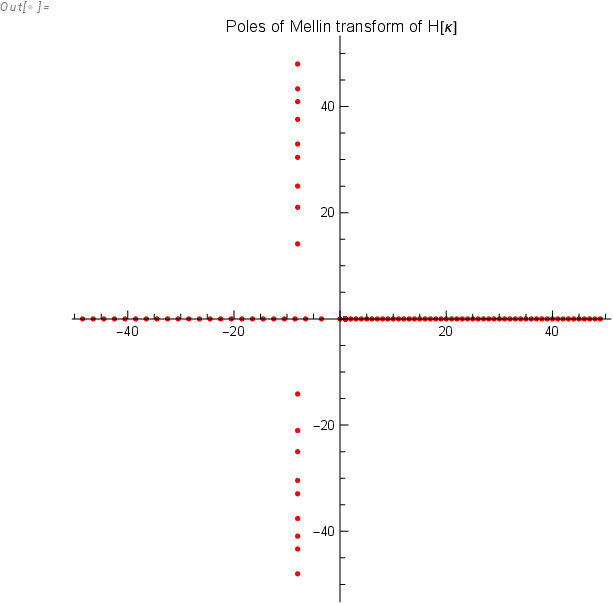}{Complex poles of the Mellin transform $h(p)$ of $H(\kappa)$ in \eqref{Mellin} defining the critical indexes of the energy spectrum as a function of $\kappa = |\vec{k}| \sqrt{ \tilde{\nu} t }$.}
\pctWPDF{0.5}{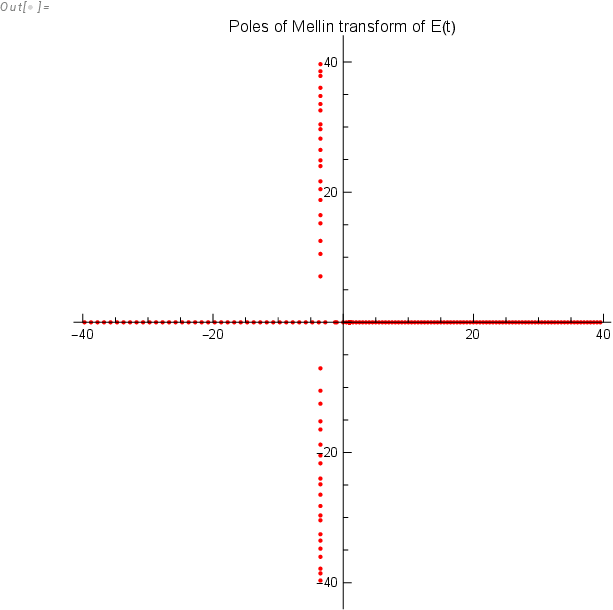}{Complex poles of the Mellin transform $e(q)$ of $E(t)$ in \eqref{EnergyMellin} defining the spectrum of critical indexes of remaining energy as a function of time.}
\pctWPDF{0.5}{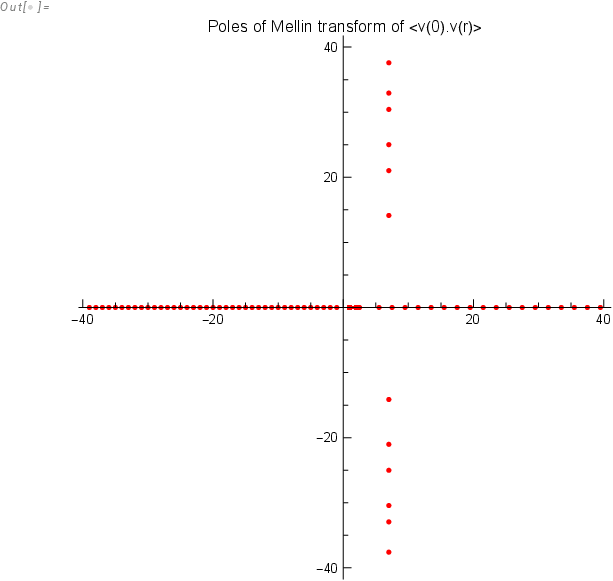}{The spectrum of (complex) indexes of the power expansion of the velocity correlation function. The poles in the right semiplane determine the small-distance expansion of the correlation function $ \vec v(0) \cdot \vec v(r) \sim \sum_n O_n |\vec{r}|^{p_n}$.}

The pure power law of decay would correspond to the Mellin transform $h(p)$ having a single pole in the left semi-plane. The position $p = -a$ of this pole becomes an index of the power law.
\begin{eqnarray}
    && H(\kappa) \propto \kappa^{-a}
\end{eqnarray}
The next level of complexity would be a function, depending on an extra parameter $n$, such that the Mellin transform has a simple pole, moving with this parameter $ p = a(n)$. In application to the moments of velocity difference, this parameter $n$ is the degree of the moment $M_n =\VEV{\Delta v^n}$. This pole will produce multifractal scaling laws
\begin{eqnarray}
   M_n \sim r^{a(n)}
\end{eqnarray}
Our energy spectrum has a more complex singularity structure in its Mellin transform
\begin{eqnarray}
\label{Mellin}
   && h(p) = \frac{f(p) \zeta \left(p+\frac{15}{2}\right) \Gamma (-p)}{(2 p+7) (2 p+17) \zeta \left(p+\frac{17}{2}\right)};\\
   && f(p) = 20 \int_{\Delta_1}^{\Delta_2} d \Delta(1-\Delta)  C^{p-1} (A C-B p)
\end{eqnarray}
where $A,B,C$ are some smooth positive functions of $\Delta$ varying in finite limits (see Appendix \ref{RiemanZeros}).
Given these properties, it is simple to prove that the Taylor series of $f(p)$ at the origin converges as the expansion coefficients decrease as a factorial of the expansion order.

This convergence makes this function an entire function without any finite singularities. The values of $C(\Delta)$ are bounded by two positive limits
\begin{eqnarray}
    0.0541984 < C < 0.0630755,  \forall{ \Delta_1 < \Delta < \Delta_2}
\end{eqnarray}
Therefore, the entire function $f(p)$ decreases as $e^{-2.76342 p}$ in the right semiplane and grows as $e^{-2.9151 p}$ in the left semiplane. It oscillates along an imaginary axis, which is our integration path. 

Theoretically, $f(p)$ could vanish at one or more positions of the poles of the remaining meromorphic function, eliminating these poles.
We computed $f(p)$ at the lowest poles and ensured it was far from zero. However, canceling some higher poles by the root of $f(p)$ remains an open problem.

The singularities of the Mellin transform $h(p)$ for the energy spectrum $H(\kappa)$ are  given by the following table of \textbf{simple poles}
    \begin{eqnarray}\label{SpecPoles}
    \left|
\begin{array}{c}
\text{energy spectrum indexes}  \\
\hline
  -\sfrac{7}{2}  \\
 -\sfrac{13}{2} \\
  \fbox{$-8  \pm \I t_{n}\text{ if }n\in \mathbb{Z}$} \\
 \fbox{$-\sfrac{17}{2}-2 n\text{ if }n\in \mathbb{Z}\land n\geq 0$}  \\
 \hline
 \fbox{$n\text{ if }n\in \mathbb{Z}\land n\geq 0$}  \\
\end{array}
\right|
\end{eqnarray}
Here $\pm t_n$ are imaginary parts of zeros of the $\zeta$ function on the critical line $z = \oh + \I t$. About $10^{13}$ zeros are already known, though the Riemann hypothesis (no other complex zeros) still needs to be proven.
These poles are shown in Fig. \ref{fig::MellinPoles}.

Only poles with negative real parts contribute to the power expansion at $\kappa \to \infty$.
The real negative poles yield decaying power terms, but the infinite series of complex poles at Riemann zeros adds the oscillations in log scale
\begin{eqnarray}
   |W_n|\kappa^{-8} \cos\left( t_n \log \kappa + \arg W_n\right)
\end{eqnarray}
These slow oscillations are visible as regime change in the log-log plots of effective indexes in Fig.\ref{fig::NPlot}, Fig. \ref{fig::MuIndex}.

As we already discussed above, the energy spectrum decays as $ t^{-\sfrac{9}{4}} k ^{-\sfrac{7}{2}}$. The wavelength decay at a fixed time is faster than K41 $k^{-\sfrac{5}{3}}$.  There is no theoretical reason (even at the level of heuristic) for K41 in decaying turbulence, as the dissipation $\mathcal E(t)$ is not a constant, so it cannot be used as a single scaling parameter.
\subsection{The energy decay}

The energy $E(t)$ is related to the same function
    \begin{eqnarray}
    &&E(t) = \int_t^\infty \, d t  \mathcal E(t) = \nonumber\\
    &&4 \pi \tilde{\nu} \int_t^\infty \, \frac{d T}{T^2} \int_{k_0 \sqrt{\tilde{\nu} T}}^\infty H(x) x^2 d x
\end{eqnarray}

Substituting the Mellin transform for $H(x)$ and integrating twice, we get the Mellin transform for the energy
\begin{subequations}
\label{EnergyMellin}
    \begin{eqnarray}
   && E(t) = \int_{-1- \eps - \I \infty}^{-1-\eps + \I \infty} \frac{d q}{2 \pi \I} e(q) \left(k_0^2 \tilde{\nu} t\right)^q ;\\
   && e(q) = 2 \pi  k_0^2 \tilde{\nu} ^2 \frac{h(2q-1)}{q(q+1)}
\end{eqnarray}
\end{subequations}

The table of complex poles of this function is
\begin{eqnarray}\label{Epoles}
    \left|
\begin{array}{c}
\text{energy decay indexes} \\
\hline
 -\sfrac{5}{4}\\
 -\sfrac{11}{4} \\
 \fbox{$-\sfrac{7}{2}  \pm \sfrac{\I}{2} t_{n} \text{ if }n\in \mathbb{Z}$} \\
 \fbox{$-\sfrac{15}{4} -n\text{ if }n\in \mathbb{Z}\land n\geq 0$} \\
 \hline
  \fbox{$\sfrac{n}{2}  \text{ if }n\in \mathbb{Z}\land n\geq 0$} \\
\end{array}
\right|
\end{eqnarray}
These poles are shown in Fig.\ref{fig::EnergyPoles}. The leading pole is at $p= -\sfrac{5}{4}$, corresponding to the asymptotic decay we compared to the grid turbulence decay data. Only poles with negative real parts contribute to the large $t$ expansion.

\subsection{The velocity correlation function in coordinate space}\label{VVCF}

Let us transform the velocity correlation back to coordinate space from Fourier space
    \begin{eqnarray}
    &&\VEV{\Delta \vec v^2}(r) =\nonumber\\
    &&2\int \frac{d ^3 k}{ (2 \pi)^3} \left(1-\exp{\I \vec{k} \cdot (\vec{r})} \right)\frac{E(k,t)}{4 \pi \vec{k}^2}= \nonumber\\
    &&\frac{\tilde{\nu}^2}{\pi ^2 \nu t} \int_{\eps - \I \infty}^{\eps + \I \infty} \frac{d p}{2 \pi \I} \left(\frac{|\vec{r} - \vec{r}'|  }{\sqrt{ \tilde{\nu}  t}}\right)^p \nonumber\\
    &&\cos \left(\frac{\pi  p}{2}\right) \Gamma (-p-1)h(-1-p) ;\\
    && \VEV{\Delta \vec v^2}(r)  = \VEV{\left(\vec v(\vec{r}, t) - \vec v(0,t)\right)^2}
\end{eqnarray}

The Mellin transform simplifies to
\begin{eqnarray}
\label{VVcorr}
    &&\VEV{\Delta \vec v^2}(r)  =\nonumber\\
    &&\frac{\tilde{\nu}^2}{ \nu t} \int_{\eps - \I \infty}^{\eps + \I \infty} \frac{d p}{2 \pi \I} V(p) \left(\frac{|\vec{r} - \vec{r} '|  }{\sqrt{ \tilde{\nu}  t}}\right)^p;\\
    && V(p) = \nonumber\\
    &&-\frac{f(-1-p) \zeta \left(\frac{13}{2}-p\right) \csc \left(\frac{\pi  p}{2}\right)}{16 \pi ^2 (p+1) (2 p-15) (2 p-5) \zeta \left(\frac{15}{2}-p\right)}
\end{eqnarray}

The poles of this function in the right semiplane represent the indexes of the power singularities of the velocity correlation function at coinciding points $(0, \vec{r})$. Should our theory be a CFT (which it is not), this spectrum would be related to the spectrum of anomalous dimensions $ \Delta_n$ in the OPE:
$$ 
\text{CFT: }\vec v(0) \cdot \vec v(r) \sim \sum_n O_n |r|^{p_n};\; p_n = 2 \Delta_v - \Delta_n;
$$
We have such an expansion with $\frac{|\vec{r}|  }{\sqrt{ \tilde{\nu}  t}}$ in place of $|r|$ and a factor $\frac{\tilde{\nu}^2}{\nu t} $ in front.
The spectrum of these scaling indexes $p_n$ (unrelated to a dilatation operator as far as we know)  is given in the following table:
\begin{eqnarray}
\label{VVSpectrum}
    \left|
\begin{array}{c}
\text{indexes of velocity correlation} \\
\hline
 -1 \\
 0\\
 \hline
  \fbox{$2 n\text{ if }n\in \mathbb{Z}\land n\geq 1$}\\
 \sfrac{5}{2}\\
 \sfrac{11}{2}\\
 \sfrac{15}{2}\\
 \fbox{$\sfrac{1}{2} (15+4 n)\text{ if }n\in \mathbb{Z}\land n\geq 1$} \\
 \fbox{$7 \pm\imath t_{n}\text{ if }n\in \mathbb{Z}$} \\
\end{array}
\right|
\end{eqnarray}
\pctPDF{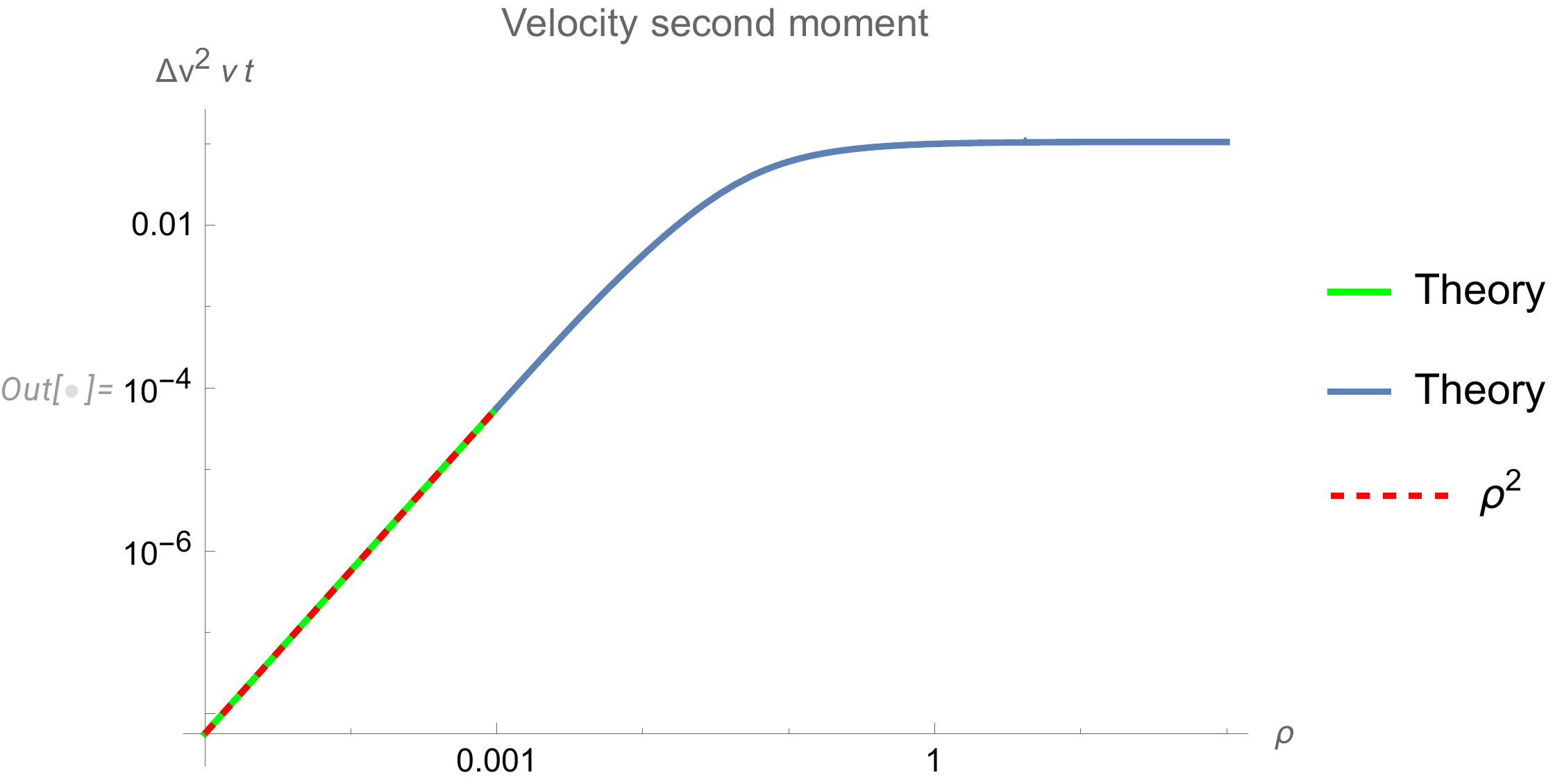}{The universal function \eqref{VVcorr} as a function of $\rho =|\vec{r}|/\sqrt{\tilde{\nu} t}$. The turnover is caused by subleading terms in the power expansion starting with $\rho^{2}$. The next terms involve quantum oscillations manifesting as a turnover from power growth to power decay. Asymptotic at large $\rho$ is $\const{}$ as it follows from the table \eqref{VVSpectrum} of the decay indexes. }
\pctWPDF{0.9}{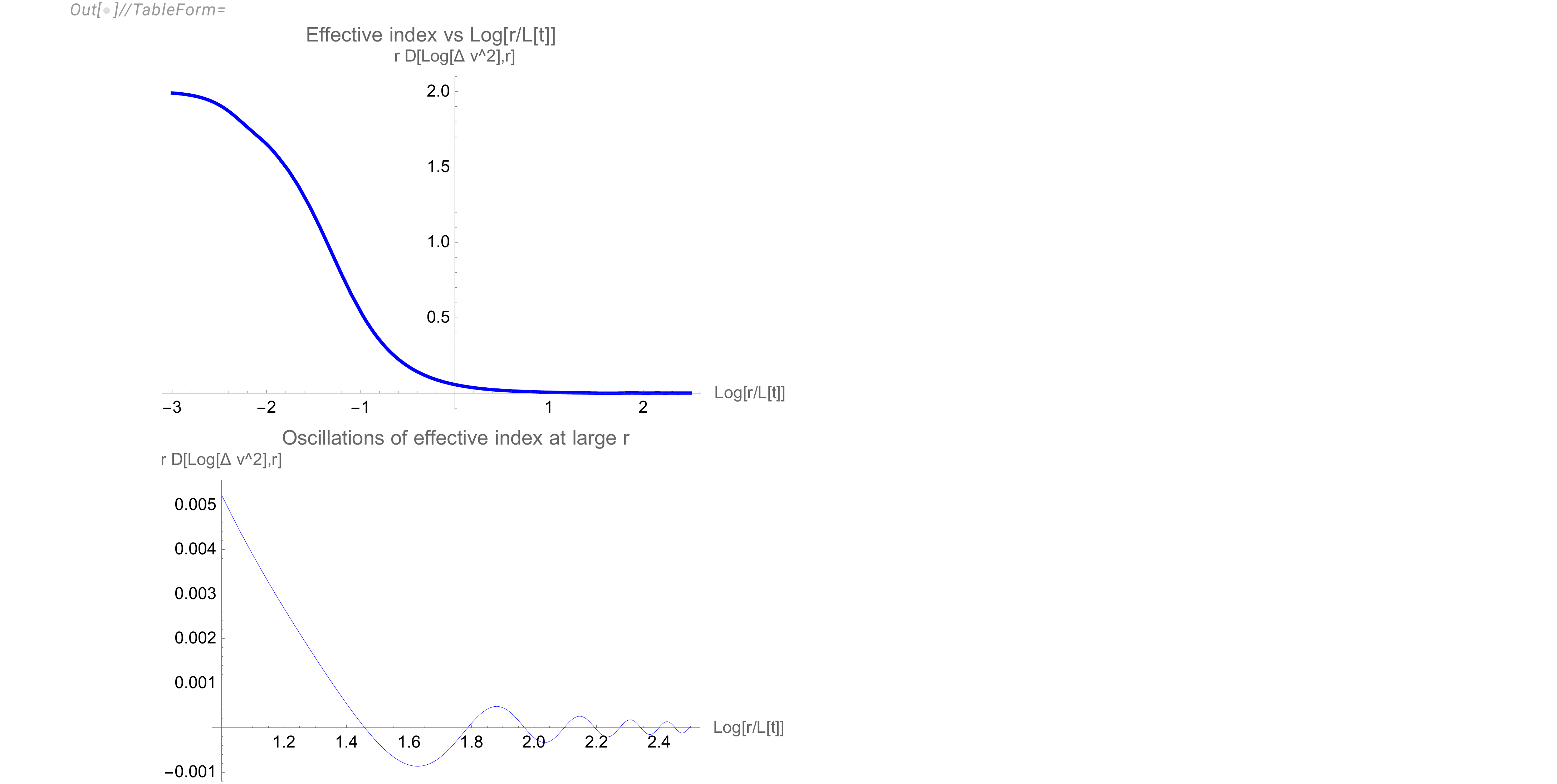}{Oscillations of the effective index $\xi_2(r)$ at large $\log_{10}r$. This is a theoretical curve corresponding to the zoom into a region of large separations, currently inaccessible by DNS with required accuracy.}

We have no CFT but a calculable spectrum of scaling dimensions. Unlike the CFT in three dimensions, this spectrum is complex.

Only poles with positive/(non-positive) real parts contribute to the power expansion at $|\vec{r}| \to 0/\infty$.
The leading term at $r \to 0$ is $r^2$, which is calculable in general form from its definition after expanding the exponential and averaging over directions of $\vec{k}$
\begin{eqnarray}
   && \VEV{\Delta \vec v^2}(r) \to \frac{  r^2 \mathcal E_\infty(t)}{ 24 \pi^3 \nu} =\frac{ r^2 \tilde{\nu}}{20736  \pi\zeta (3)\nu t^2} 
\end{eqnarray}
The next term is $r^{\sfrac{5}{2}}$ as it follows from the table in \eqref{VVSpectrum}.
As mentioned about the energy spectrum, there is no K41 scaling index $ p = \sfrac{2}{3}$. This omission is not a contradiction, as K41 does not apply to decaying turbulence.
Instead of pure scaling laws with single decay indexes, we found an infinite spectrum of scaling indexes, some of which come as complex conjugate pairs, which leads to quantum oscillations: see Fig.\ref{fig::Xi2Oscillations} for the oscillation of the index $\xi_2$ at the latest stage. This region is inaccessible with modern scale of the DNS.

The theoretical curve in Fig.\ref{fig::VelocityCorrelation} agrees with the Fourier-transformed data of \cite{SreeniDecaying, GDSM24} but deviates from \cite{GregXi2}. The probable reason is the compressibility of the air in real experiments in \cite{GregXi2}.

\textbf{The imaginary parts of these complex scaling dimensions coincide with those of the famous Riemann zeta  zeros, establishing an intriguing relation between Turbulence and Number Theory.}

\section{Discussion}\label{discussion}

This section will try to reconcile traditional perspectives on turbulence phenomena, including enduring beliefs and myths, with our new theory.

\subsection{Myth and Reality of Turbulent Scaling Laws}

More than eighty years ago, Kolmogorov and Obukhov made a breakthrough in turbulence theory by establishing the relation \eqref{K41vvv} for the three-point correlation function of the velocity field in turbulent flow. This formula only selects the \textbf{potential part} of the triple velocity correlation function by taking two coincident points. When taking the curl, we get zero: 
$$ \VEV{\val(\vec{r}_0) \vbe(\vec{r}_0) \oga(\vec{r} + \vec{r}_0)} = 0 $$
This relation indicates no constraints on the rotational part concerning triple vorticity correlations and sheds no light on the scale invariance of turbulence theory.

Moreover, the linear term in coordinates does not have any support in the Fourier spectrum: this is an example of the harmonic terms added to the \BS{} integral for the velocity field.

The K41 scaling law was introduced as a phenomenological model, not intended to replace the missing microscopic theory. It was based on the assumption that the local dissipation density does not fluctuate—a limitation its creators were aware of, prompting them to propose a log-normal distribution for this variable later on. However, even this modified model lacked a microscopic justification and failed to fully correspond with empirical observations.

Subsequent experiments and DNS \cite{YZ93, YS04, SY21} have invalidated the K41 scaling laws (including the log-normal model) over the past thirty years. Regarding decaying turbulence, the experimental data \cite{SreeniDecaying} have diverged even further from Kolmogorov scaling laws despite all attempts to stretch this data or discard the non-fitting region as "erosion." We highlight significant deviations—six orders of magnitude—from the $k^{-\sfrac{5}{3}}$ scaling in Fig. \ref{fig::DecayingSpectrum}. There are also recent measurements \cite{GregXi2} with significant deviations of the log-log derivative of the second velocity moment $\VEV{\Delta v^2}(r)$ from the $\frac{2}{3}$ predicted by K41.

A broader assumption posited that \textbf{power laws with anomalous dimensions} might exist in the inertial range. The assumed analogy to critical phenomena led to the proposal of multifractal scaling laws \cite{FP85}, which, as a phenomenological model, successfully described observed deviations from the K41 laws in forced turbulence\cite{YS04, SY21}.

However, there are no theoretical grounds for conformal symmetry in turbulence; the 'current conservation' conditions $\dal \val = 0, \dga \oga = 0$ in the CFT would prescribe both velocity and vorticity dimensions of $d-1 = 2$, contradicting the fact that vorticity is a curl of velocity.

Moreover, the anomalous dimension would not explain decaying turbulence, as the log-log plots would remain straight lines, though the slopes would become irrational numbers. We must allow nonlinear correlation functions on a log-log scale, as indicated by the data in Fig. 2 (top) in \cite{GregXi2}. Here, energy spectra for various parameters converge into a universal curve on the decreasing part of the spectrum, which is curved on the log-log scale, indicating that a simple power law cannot describe it. Instead, it is a nontrivial universal function of $\log k$, spanning several decades.

Both the DNS and experimental papers \cite{SreeniDecaying, GregXi2} noted significant deviations from scaling laws (whether K41 or multifractal). The conclusion of \cite{SreeniDecaying} was cautiously negative: "it is somewhat disappointing that the results are not more closely aligned with theoretical arguments." The most recent paper \cite{GregXi2} made a stronger negative claim: "Our results point to a Reynolds number-independent logarithmic correction to the classical power law for decaying turbulence that calls for theoretical understanding."

Our recent paper \cite{M22} presented the theoretical argument for the breaking of scaling laws due to logarithmic divergences in a dilute gas of vortex filaments. In this approximation, there were logarithmic terms in the effective energy for the filament, leading to violations of scaling laws akin to asymptotic freedom in QCD. This approximation does not apply to decaying turbulence with a large density of vortex structures, but at least it identifies a dynamical mechanism for the deviations from the scaling laws.

In the present paper, we used raw data from the DNS \cite{SreeniDecaying, GDSM24} to compute the effective index of the velocity correlation by numerical Fourier transform of their energy spectrum (see Section \ref{DNS}). Our effective index is plotted in Fig. \ref{fig::MoveIndex14}. The K41 scaling law $\Delta v^2 \sim r^{\sfrac{2}{3}}$ is very far from reality, as it is clear from these plots. Our theory is much closer, and by fitting our arbitrary length scale, we obtained a very good fit in the turbulent range within experimental errors.

The microscopic theory developed here is not conformally invariant but retains a critical aspect of CFT. The Mellin transform of the vorticity field's correlation function in coordinate space is a meromorphic function of the Mellin parameter $p$. This characteristic implies some underlying scale invariance with an infinite discrete spectrum of complex anomalous dimensions \eqref{VVSpectrum}.

In this way, \textbf{our theory extends the multifractal scaling laws by accommodating an infinite discrete spectrum of scaling dimensions}. According to recent DNS results \cite{DecayingTalk}, our predictions align with observed slopes, unlike conventional scaling models such as those proposed by Kolmogorov and Saffman.

This result suggests that experiments and DNS in decaying turbulence should be conducted at larger scales and higher Reynolds numbers, fitting the data for logarithms of the spectrum and energy dissipation decay as nonlinear functions of the logarithm of the product of wavevector and the square root of time, as we (successfully) did in Section \ref{DNS}.

As part of this reevaluation, one should magnify and study the decaying part of the spectrum way beyond its middle part, roughly described by $-\frac{5}{3}$ law with logarithmic corrections. This decaying part, the "dissipative subrange," was discarded as an unfitting puzzle piece in conventional data fitting, but it fits well in our theory.

Heisenberg \cite{heisenberg1971bedeutung} and Chandrasekhar \cite{chandrasekhar1949heisenberg} proposed in the middle of the last century for the "dissipative subrange," the spectrum decay $k^{-7}$, based on a model equation by Heisenberg. At that time, there were no mathematical tools to solve the turbulence problem exactly, so the model equations like that one passed as theories. The fame of two Nobel laureates involved added weight to this model assumption, so it stays alive to this day.

K.R. Sreenivasan dispelled this die-hard myth in his paper \cite{SreeniChandra}.
\begin{quote}
    Chandra's initial enthusiasm for Heisenberg's work was moderated when he learned from J. von Neumann, in a colloquium that Chandra gave at Princeton in the spring of 1949, that the $k^{-7}$ power law in the far-dissipation range did not have experimental support.
\end{quote}
Our theory also contradicts the $k^{-7}$ law: no pole exists between $-13/2$ and $-8 \pm \imath t_n$ in our Mellin transform spectrum \eqref{SpecPoles}. Instead, we have nontrivial dynamics at this "dissipative subrange," not just in the "inertial range" between energy pumping and dissipation. The full plot of the effective index for the spectrum is shown in Fig. \ref{fig::MuIndex}, asymptotically approaching $-\frac{7}{2}$. As long as enough energy is left for a turbulent flow, our theory has a universal decaying spectrum spanning several decades and a strongly curved second moment $\VEV{\Delta v^2}(r)$ with effective index (log derivative) shown in Fig. \ref{fig::PlotIndex14}. We call the corresponding range of scales "turbulent range", combining old inertial and dissipation ranges. Our theory perfectly matches DNS/experiment in the whole turbulent range without any dimensionless fitting parameters.

In conclusion, single-power scaling laws cannot describe the observed critical phenomena in decaying turbulence. Instead, compare these phenomena with the microscopic theory, which goes beyond empirical laws, replacing them with universal nonlinear functions for the energy spectrum, energy decay, and velocity correlation.

\subsection{Stochastic solution of the \NS{} equation and ergodic hypothesis}

Richard Feynman wrote about turbulence in his Lectures in Physics \cite{Feynman}: 
\begin{quote}
    Nobody in physics has really been able to analyze it mathematically satisfactorily in spite of its importance to the sister sciences. It is the analysis of circulating or turbulent fluids.... What we really cannot do is deal with actual, wet water running through a pipe. That is the central problem which we ought to solve some day, and we have not.
\end{quote}
These words were written over sixty years ago, but the problem remains unsolved.

We address this problem by seeking a stochastic solution to the unforced Navier-Stokes equation, covering a universal manifold over an infinite time. Our solution reveals power-like singularities in correlation functions, which emerge after averaging across this manifold in the statistical limit as its dimension approaches infinity.

We already have a partial answer to the question posed by Feynman about water flowing through a pipe: its local kinetic energy density decays with time (and also with distance from the grid at the entrance) as $t^{-\sfrac{5}{4}}$. A more detailed answer for the pressure as a function of the total amount of water pushed through the pipe would require some future investigation of our solution.

These singularities originate from the infinite time required to cover this manifold uniformly.

We identify this manifold (the Euler ensemble) by solving the loop equation—a subset of the Hopf functional equation for the generating functional of velocity field probabilities. Notably, none of the solutions within this manifold experiences finite-time blow-ups. Instead, we encounter singularities from the fixed trajectory of the loop equation, not from its finite-time solutions.

We adopt the most natural invariant measure from the perspective of number theory: each element of the Euler ensemble is weighted equally, an assumption we term the quantum ergodic hypothesis.

With this invariant measure, the Euler ensemble stands out because the loop functional is equal to the trace of the evolution operator in a quantum system—the Fermi particles on a ring interacting with a quantum field made of fractions of $\pi$. Every distinct state, including every distinct fraction, contributes equally to the quantum trace in this discrete system. For our purposes, this means treating every element in the Euler ensemble with equal weight.

Our quantum ergodic hypothesis thus stipulates an exact equivalence between the loop functional and the quantum trace of an evolution operator for the one-dimensional ring of Fermi particles. This quantum analogy has paved the way for an analytical solution in the turbulent limit. This limit corresponds to the quasiclassical limit of this Fermi system, where viscosity acts like Planck's constant.

The quantum ergodic hypothesis results from a more general relation between fluid dynamics and quantum mechanics. The loop equation, in the general case, with finite viscosity and external stochastic forces in the \NS{} equation, represents the \Schr{} equation in loop space \cite{M93, M23PR}.

The time evolution of the wave function, which is the loop functional, is given by the sum over "classical" histories, corresponding to this loop space Hamiltonian. Dirac and Feynman established that the weight of each history $\exp{\imath S/\hbar}$ for any quantum system with the Action $S$.

Comparing the Dirac-Feynman rule with the definition of the loop functional, we conclude that the velocity circulation $\Gamma_C[v] = \oint_C \vec v \cdot d \vec{r}$ plays the role of the Action in the loop quantum mechanics, and viscosity plays the role of Planck's constant. The sum goes over the classical solutions of the \NS{} equation with various initial data, with equal weight for each solution.

We cannot describe all these solutions for the velocity field, but surprisingly, we can compute the weighted sum of all these solutions, i.e., the loop functional. Let us stress that this is not an asymptotic solution of the loop equation, with some terms neglected at large times. Our solution \eqref{LoopSol}, \eqref{PFT}, \eqref{Fsol} exactly satisfies the \NS{} equation at a finite time for the loop functional in the turbulent limit $N \to \infty, \nu \to 0$.

The wave functional is not localized in the weak turbulence phase (small circulations compared to viscosity), so states are not quantized. This quantization occurs only in the strong turbulent phase (large circulations); the Euler ensemble or the Fermi ring describes it.

In the same way, as one-dimensional quantum mechanical motion in external potential becomes finite and quantized when potential well becomes deep enough, our loop functional at large time transforms from continuous distribution in loop space to the quantized finite motion characterized by the momentum loop $\vec{P}(\theta)$. The continuous quantum mechanical integral over phase space with equal weight $1/(2 \pi)$ per DOF becomes the discrete sum over all distinct quantum levels with unit weight.

We hope our quantum ergodic hypothesis can be proven from the \NS{} equation, starting with the quantum representation of the loop equation as a \Schr{} equation in loop space. If confirmed, this hidden quantum mechanics of classical turbulence may become a law of Nature rather than a computational method.

Like the classical ergodic hypothesis, this may take another hundred years. Theoretical physics does not wait for rigorous proof but rather explores the consequences of the conjectured theory and compares them with physical and numerical experiments.

\subsection{The physical meaning of the loop equation and dimensional reduction}

The long-term evolution of Newton's dynamical system with many particles eventually covers the energy surface (microcanonical ensemble). The ergodic hypothesis, accepted in Physics but still not proven mathematically, states that this energy surface is covered uniformly. Turbulence theory aims to find a replacement for the microcanonical ensemble for the \NS{} equation. This surface would also participate in the decay in the pure \NS{} equation without artificial forcing. 

In both cases, Newton and \NS{}, the probability distribution must satisfy the Hopf equation, which follows from the dynamics without specifying the mechanism of the stochastization. Indeed, the Gibbs and microcanonical distributions in Newton's dynamics satisfy the Hopf equation in a rather trivial way: it reduces to the conservation of the probability measure (Liouville theorem), which suggests the energy surface as the only additive integral of motion to use in the exponent of the fixed point distribution.

The loop technology has been thoroughly discussed in the last few decades in gauge theories, including QCD \cite{MLDMig86, Mig98Hidden, LoopEqBootstrap, KazakovLooqBootstrap, kazakovZheng2024}, where the loop equations were first derived \cite{MMEq79, Mig83}.

In the case of decaying turbulence, the loop equations represent a closed subset of the Hopf equations, which is still sufficient to generate the statistics of vorticity. In this case, the exact solution we have found for the loop functional also follows from the integrals of motion, this time, the conservation laws in the loop space.

The loop space Hamiltonian we derived from the unforced \NS{} equation does not have any potential terms (those with explicit dependence upon the shape of the loop). The Schrödinger equation with only kinetic energy in the Hamiltonian conserves the momentum. The corresponding wave function is a superposition of plane waves $\exp{\I \vec{p} \cdot \vec{x}}$. This superposition is the solution we have found, except the dot product $\vec{p} \cdot \vec{x}$ becomes a symplectic form $\oint \vec{P}(\theta) \cdot d \vec{C}(\theta)$ in the loop space.

Our momentum $\vec{P}(\theta, t)$ is not an integral of motion, but simple scaling properties of the pure \NS{} equation lead to the solution with $\vec{P}(\theta, t) \propto \vec{F}(\theta)/\sqrt{t}$, with $\vec{F}(\theta)$ being the integral of motion at large time (i.e., a fixed point). The rest is a purely technical task: substituting this scaling solution into the \NS{} equation and solving the resulting universal equation for a fixed point $\vec{F}(\theta)$.

This equation led us to the Fermi ring in the quasiclassical limit. The solution of the Fermi ring in this limit resulted in the energy spectrum and dissipation in a finite system found in Section \ref{finitesystem}.

\subsection{Classical Flow and Quantum Mechanics}

Our computations rely significantly on number theory, particularly Jordan's multitotients, $\varphi_l(q)$, which extend the Euler totient function \cite{multitotients}. What could number theory share with turbulent flow? The quantization of parameters in the fixed manifold of decaying turbulence originates from the quantum correspondence identified in the nineties \cite{M93}. The statistical distribution of a nonlinear classical Navier-Stokes (NS) PDE is related to the wave functional of a linear Schrödinger equation in loop space, as detailed in the previous section.

Is quantum mechanics at work in a water faucet with a grid filter? Yes and no.

This relationship is indirect: the loop functional, a Fourier transform of the classical probability distribution for circulation, equals the complex quantum amplitude of the loop space theory represented by a Fermi ring. Probability is real and positive, while its Fourier transform is complex, reflecting the irreversibility of the \NS{} dynamics.

A probability distribution for circulation satisfies another loop equation \cite{M93, M23PR}, with all coefficients being real. The time evolution of this probability spans alternative histories, as it typically does in statistical mechanics, but the weights of each history remain real and positive.

On the other hand, the complex loop functional adheres to the quantum mechanical evolution equation, resulting in quantum interference of alternative histories. The quantum interference is quite significant here, with the dominant complex trajectory—the instanton—describing notable quantum effects such as exponential cancellation of contributions from alternative histories and penetration into classically forbidden regions within loop space.

The quantization mechanism of the parameters in the plane wave solution mirrors that of conventional quantum mechanics: the solution's periodicity $\vec{P}(\xi + 1) = \vec{P}(\xi)$.

From the conventional perspective, the fractal curve in complex momentum space, $\vec{P}(\theta) \in \mathbb{C}_3$, or Fermi particles on a circle, may seem unrelated to turbulent flow. One could ask how turbulence can be addressed without directly studying the velocity field.

The well-established duality phenomenon, known as ADS/CFT duality, equates the strong coupling phase of a conformal field theory to the weak coupling phase of quantum geometry in another dimension. This relationship is more than just a method for calculating correlation functions of a strongly fluctuating vector field; it reveals a \textbf{second identity} of the original theory.

The quantum Fermi ring's particle density fluctuations disappear in the turbulent limit. In contrast, the original theory's fluctuations are so intense that the vorticity field ceases to exist. This Fermi ring arguably reveals the \textbf{true identity} of decaying turbulence, as a classical function describes a smooth Fermion density—the instanton solution we identified.

Coming back to decaying turbulence in a water faucet, the probability distribution of velocity circulation in the water stream you wash your hands with is classical, of course; however, this classical distribution decays by a complex law based on quantum interference for its Fourier transform. This Fourier transform (loop functional) adds up from alternative histories with quantum mechanical complex weights.

Mathematical physics sometimes has dual representations for the same phenomena, such as the duality between particles and waves in quantum mechanics, or between matrix models \cite{brezin1990exactly, douglas1990strings, GMQG} and Liouville theory \cite{Pol81, KPZ} in 2D quantum gravity.

Additional complexities arise in gauge theory due to short-distance singularities involving the infinite fluctuating degrees of freedom in quantum field theory. Wilson loop functionals in coordinate space are singular in gauge field theory and cannot be multiplicatively renormalized.

There are no short-distance divergences in the Navier-Stokes equations and NS loop dynamics. The Euler equations represent a singular limit that, as argued, should be resolved through singular topological solitons regularized by the Burgers vortex.

In the dual theory of this paper, the singularities exist in the dual space $\mathbb{C}_3$. Anomalous dissipation is achieved through numerous finite discontinuities of the fractal curve $\vec{P}(\xi) \in \mathbb{C}_3$.

However, these singularities only occur in the inviscid limit, $\nu \propto 1/N^2 \to 0$, representing Euler singularities like line vortices that are regularized by finite viscosity, just like our singularities.

\subsection{Renormalizability of the Inviscid Limit of the Loop Equation}

Let us examine the relationship $\nu N^2 = \const{}$ between diminishing viscosity and the increasing number $N$ of discontinuities on the momentum loop $\vec{P}(\theta)$.

The Navier-Stokes (\NS{}) equation is essentially an idealization of molecular dynamics, approximating nonlocal theory by a truncated expansion in powers of gradients.

In the case of laminar flow, this truncated expansion poses no issues. However, in our solution for the loop equation, the velocity field becomes singular in the local limit.

Mathematically, velocity and vorticity are not ordinary functions in $\mathbb{R}_3$ but stochastic variables, with $\Delta \vec v \sim  (\Delta \vec{r})^\alpha$.

Fractal Calculus \cite{FractalCalculus} was introduced to generally describe such fields. Yet, this alone does not account for turbulence, particularly since the more general power laws with multifractal dimensions $\VEV{(\Delta \vec v)^n} \sim  (\Delta \vec{r})^{\zeta_n}$ cannot be explained in this manner.

In our theory, there is a universal function \eqref{VVcorr} for the velocity correlation function as a function of the separation $\vec{r}$. However, the $\vec{r}$ dependence is not defined by a single power: the Mellin transform of this function reveals infinitely many poles \eqref{VVSpectrum} in the complex plane. The absence of branch cuts supports the idea of scale invariance, but the analogy with CFT ends here.

Our solution characterizes a fractal vorticity field, yet this field does not conform to known fractal types. The fractal curve $\vec{P}(\theta)$ is not merely a random walk on a circle: dynamic restrictions, such as periodicity, lead to nontrivial critical behavior not describable by any finite set of fractal power laws.

We define the turbulent limit of the velocity field distribution by discretizing the loop equation (replacing the continuum loop with a polygon with an increasing number of vertices).

The relation between vanishing viscosity and an increasing number of degrees of freedom parallels the renormalization group (RG) relation between the bare coupling constant $g_0$ and the lattice spacing $a$ in QCD: $a \propto g_0^\alpha \exp{- \beta/g_0^2}$.

The naive local limit $a \to 0, g_0 = \const{}$ does not exist in QCD, but the RG limit effectively describes the strong interaction of hadrons.

Similarly, in our approach, renormalizability is present. The dissipation rate remains finite as $\tilde{\nu} = \nu N^2 \to \const$. Furthermore, the energy spectrum, expressed as a function of $k \sqrt{\tilde{\nu} t}$, remains finite in the local limit.

Thus, akin to renormalizable quantum field theory (QFT), a "dimensional transmutation" occurs where infinities are absorbed into the dimensional parameter $\tilde{\nu}$, defining the time scale.

The energy scale is related to the fluid's physical viscosity $\nu$ in the denominator. Formally setting $\nu = 0$ would tend the physical energy scale to infinity. The fictitious limit $\nu \to 0$ was only taken to compute the turbulent limit of the energy spectrum as a function of Reynolds number $\Gamma/\nu$.

In the real world, this number is large because of the large circulation $\Gamma$ compared to finite viscosity $\nu$. Still, we can compute these dimensionless functions by taking the inviscid limit, as described by equation \eqref{EnergyResidue} in Section \ref{Enorm}.

In particular, the energy spectrum is proportional to $1/\nu$ times the function of the effective Reynolds number, time, and wave vector. We tend the Reynolds number to infinity in the turbulent limit but keep a finite physical viscosity factor in the energy spectrum.

\subsection{Relation of Our Solution to Weak Turbulence}

The solution of the loop equation with finite area derivative, satisfying the Bianchi constraint, belongs to the category of Stokes-type functionals \cite{MMEq79}, similar to the Wilson loop for gauge theory and fluid dynamics.

The Navier-Stokes (NS) Wilson loop represents a case of the Abelian loop functional, characterized by commuting components of the vector field $\vec v$. As extensively discussed in \cite{MMEq79, Mig83, M23PR}, any Stokes-type functional $\Psi(\gamma, C)$ that satisfies the boundary condition at a shrunk loop $\Psi[\gamma, 0]=1$ and solves the loop equation can be iterated in the nonlinear term in the NS equations. This iteration is applicable in conditions of high viscosity.

The resulting expansion in inverse powers of viscosity (representing weak turbulence) coincides with the standard perturbation expansion of the NS equations for the velocity field, averaged over the initial data distribution.

We have shown in \cite{M93, M23PR} (and also here, in Section \ref{mathintro}) how the velocity distribution for random uniform vorticity in the fluid was reproduced by a singular momentum loop $\vec{P}(\theta)$. 

The solution for $\vec{P}(\theta)$ at this specific fixed point of the loop equation was complex and had slowly decreasing Fourier coefficients, leading to a discontinuity $\sign(\theta-\theta')$ in a pair correlation function \eqref{Pcorr}. The corresponding Wilson loop is equated to the Stokes-type functional \eqref{InitPsi}.

Using this example, we demonstrated how a discontinuous random momentum loop describes the vorticity distribution in the stochastic NS flow. Here, the vorticity acts as a global random variable corresponding to a random uniform fluid rotation—a well-known exact solution of the NS equation.

This example corresponds to a special fixed point for the loop equation. Although not general enough to describe turbulent flow, it is an ideal mathematical model for loop technology. It illustrates how the momentum loop solution aggregates all terms of the $1/\nu$ expansion in the NS equation.

In our general solution, with the Euler ensemble, the summation of a divergent perturbation expansion occurs at an extreme level, leading to a universal fixed point in the limit of vanishing viscosity.

Given initial conditions, after a finite time, the solution will still depend on viscosity and the initial conditions. We expect no singularities for a smooth initial velocity field.

\section{Remaining problems}
\begin{itemize}
    \item 
    The loop functional for the circular loop is the simplest object in this theory. It can be computed using the methods developed in this paper, yielding even simpler results. In this case, the classical equation is trivial, so the computations reduce to the functional determinant and the resolvent. On the other hand, this is an observable quantity and could be measured in DNS. It would be an interesting problem to solve and compare with the DNS.
    \item 
    The higher moments of circulation or velocity differences are calculable from this general WKB approximation for the path integral at $\nu \to 0, N \to \infty$. The $n$th moment of $\Delta \vec v = \int d^3 k \frac{\I \vec{k} \times \vec\omega_k}{\vec{k}^2} \left(\exp{\I \vec{k} \cdot \vec{r}} - \exp{\I \vec{k} \cdot \vec{r}'}\right)$ reduces to the loop functional for the same backtracking "hairpin" traversed $n$ times, with vorticity inserted $n$ times at the ends. This computation will produce analytical results for the multifractal scaling laws for velocity moments.
    \item 
    The spectrum of indexes for deviations from our fixed trajectory \cite{migdal2023exact} can be evaluated to compute vorticity correlation functions in the \NS{} equation with an infinitesimal random force.
\end{itemize}

\section{Conclusion}

We have established an exact duality between decaying classical turbulence in 3+1 dimensions and a solvable one-dimensional quantum theory of Fermi particles on a ring. In this framework, Fermi particles are confined within the vorticity field, where strong vorticity fluctuations correspond to weak fluctuations in the Fermion density. Elaborating on this theory, we present an analytical solution for decaying turbulence in quadrature.

Our theory replaces the old scaling laws (including multifractal versions) with universal functions, nonlinear in a log-log scale. This microscopic theory challenges the prevailing paradigms of the past eighty years, but it perfectly fits the available experimental data \cite{GridTurbulence_1966, Comte_Bellot_Corrsin_1971} and recent DNS data \cite{SreeniDecaying, GDSM24, GregDecayTurb}. It offers a far more intriguing perspective, revealing unexpected connections between nonlinear classical physics and quantum mechanics.

\section*{Acknowledgments}
Maxim Bulatov helped me with mathematical suggestions and discussions. He also participated in an unpublished paper \cite{BulMig} where we simulated the Euler ensemble as a Markov chain. The strong cancellations in the vorticity correlation function prevented this simulation from producing reliable numerical results.

I filled the gaps in my knowledge of Number Theory in discussions of the Euler ensemble with Peter Sarnak, Alexandru Zaharesku, Konstantin Khanin, and Debmalya Basak.

Yang-Hui and I discussed Euler's totients at the Cambridge University workshop, where this theory was first reported in November 2023. His comments helped me derive the asymptotic distribution for scaling variables.

I am also grateful to the organizers and participants of the "Field Theory and Turbulence " workshop at ICTS in Bengaluru, India, where this work was advanced in December 2023. Discussions with Katepalli Sreenivasan, Rahul Pandit, and Gregory Falkovich were especially useful. They helped me understand the physics of decaying turbulence in a finite system and match my solution with the DNS data.

This work was discussed at the "Conformal Field Theory, Integrability, and Geometry" conference in Stony Brook in
March 11-15, 2024.
I am very grateful to Nikita Nekrasov, Sasha Polyakov, Sasha Zamolodchikov, Dennis Sullivan and other participants for deep and inspiring discussions.

This work was also thoroughly discussed for a week at the Perimeter Institute for Theoretical Physics, where I spoke at a colloquium on April 3. I received many interesting questions and comments, reflected in the last section, "Discussion," of this paper.
I thank Davide Gaiotto, Luis Lehner, Sabrina Pastersky, Sergey Sibiryakov, Lee Smolin, and Andrey Shkerin for inspiring discussions.
The help from Zechuan Zheng in Mathematica computations of the energy spectrum was also very useful.

I recently presented my theory at the 126th Statistical Mechanics Conference in Rutgers on May 19. There were very useful discussions, and I also got some unpublished data from John Panickacheril John and Diego A. Donzis.

After the Oxford Physics Department colloquium, Alex Schekochihin and I discussed the data and physics of decaying turbulence. I am grateful to Alex for these deep discussions and sharp comments.

The most important help and support came from Sreeni, who, for the last year, discussed the physics of decaying turbulence with me and guided me through the maze of inconsistent DNS and physical experiment data.

This research was supported by a Simons Foundation award ID $686282$ at NYU Abu Dhabi.
\section*{Data Availability}
The \Mathematica{}notebooks used to verify equations and compute some functions are available for download in \cite{MB40, MB41, MB42, MB43, MB44, MB45, MB46}.
The raw DNS data from \cite{SreeniDecaying} can be downloaded from the shared Google Drive \cite{GDSM24}.

\bibliographystyle{plain}
\bibliography{bibliography}
\setcounter{section}{0} 
\appendix
\section{Global Random rotation and Momentum Loop Space}\label{GlobalRot}
This path integral was computed in \cite{M23PR, migdal2023exact} for a special stochastic solution of the Navier-Stokes equation: the global rotation with Gaussian random rotation matrix.
The initial velocity distribution is Gaussian, with a slowly varying correlation function.
The corresponding loop field reads (we set $\gamma =1$ for simplicity in this section)

    \begin{eqnarray}
  &&\Psi_0[C] \equiv \Psi(1,C)_{t=0}  =\nonumber\\
    &&\exp{
	 -\frac{1}{2 \nu^2}\IINT{C} d \vec{C}(\theta) \cdot d \vec{C}(\theta') f\left(\vec{C}(\theta)-\vec{C}(\theta')\right)
	}
\end{eqnarray}

where $ f(\vec{r}) $ is the velocity correlation function
\begin{equation}
  \left \langle v_{\alpha}(r) v_{\beta}(r') \right \rangle =
\left(\delta_{\alpha \beta}- \partial_{\alpha} \partial_{\beta}
\partial_{\mu}^{-2} \right) f(r-r')
\end{equation}

The potential part drops out in the closed loop integral.
The correlation function varies at the macroscopic scale, which means that one could expand it in the Taylor series
\begin{equation}
  f(r-r') \rightarrow f_0 - f_1 (r-r')^2 + \dots \label{Taylor}
\end{equation}

The first term $ f_0 $ is proportional to initial energy density,
\begin{equation}
  \frac{1}{2} \left \langle v_{\alpha}^2 \right \rangle =\frac{d-1}{2}
f_0
\end{equation}
and the second one is proportional to initial energy dissipation
rate ${\mathcal E}_{0}$
\begin{equation}
f_1 = \frac{\mathcal E_{0}}{2 d(d-1) \nu}
\end{equation}
where $ d=3 $ is the dimension of space.
The constant term in (\ref{Taylor}) as well as $ r^2 + r'^2 $ terms
drop from the closed
loop integral, so we are left with the cross-term $ r r' $, which reduces to a full square
\begin{eqnarray}\label{InitPsi}
  &&\Psi_0[C] \to \exp{- \frac{f_1}{\nu^2}\left(\oint dC_{\alpha}(\theta) C_{\beta}(\theta)\right)^2}
\end{eqnarray}

This distribution is almost Gaussian: it reduces to Gaussian one by
extra integration

    \begin{eqnarray}
  &&\Psi_0[C] \rightarrow \const{}\int (d \phi) \exp{ -\phi_{\alpha \beta}^2}\nonumber\\
    &&\exp{
	 2\imath \frac{\sqrt{f_1}}{\nu}
	 \phi_{\mu\nu} \oint dC_{\mu}(\theta) C_{\nu}(\theta)}
\end{eqnarray}

The integration here involves all  $ \frac{d(d-1)}{2} =3 $ independent $ \alpha < \beta $
components of the antisymmetric tensor $ \phi_{\alpha \beta} $.
Note that this is ordinary integration, not the
functional one. 

This distribution can be translated into the momentum loop space.
Here is the resulting stochastic function $\vec{P}(\theta)$, defined by the Fourier expansion on the circle
\begin{eqnarray}
\label{Pexp}
  &&  P_\alpha(\theta)= \sum_{\text{odd }n=1}^\infty P_{\alpha,n} e^{\I n \theta} + \bar{P}_{\alpha,n} e^{-\I n \theta};\\
  && P_{\alpha,n} = \mathcal N(0,1) ;\\
  && \bar{P}_{\alpha,n} =\frac{4 \sqrt{f_1}}{n \nu} \phi_{\alpha \beta}P_{\beta,n} ;\\
  && \phi_{ \alpha\beta} = - \phi_{\beta\alpha};\\
  && \phi_{\alpha\beta} = \mathcal N(0,1) \forall \alpha < \beta;
\end{eqnarray}
At fixed tensor $\phi$ the correlations are
\begin{eqnarray}
  &&\VEV{P_{\alpha,n} P_{\beta,m}}_{t=0}
= \frac{4 \sqrt{f_1}}{m \nu} \delta_{n m} \phi_{\alpha \beta};\\
\label{Pcorr}
&&\VEV{ P_{\alpha}(\theta) P_{\beta}(\theta')}_{t=0} =
2\imath \frac{\sqrt{f_1}}{\nu}\phi_{\alpha \beta} \sign(\theta'-\theta);\\
&&\Psi_0[C]=  \VEV{\exp{
	   \frac{\I }{\nu}\oint d \vec{C}(\theta) \vec{P}(\theta)}}_{P, \phi}
\end{eqnarray}
Though this special solution does not describe isotropic turbulence, it helps understand the mathematical properties of the loop technology. In particular, it shows the significance of the discontinuities of the momentum loop $\vec{P}(\theta)$.
\section{The Markov chain and its Fermionic representation}\label{MarkovFermi}
Here is a new representation of the Euler ensemble, leading us to the exact analytic solution below.

We start by replacing independent random variables $\sigma$ with fixed sum by a Markov process, as suggested in \cite{migdal2023exact}. 
We start with $n$ random values of $\sigma_i = 1$ and remaining $N-n$ values of $\sigma_i = -1$.
Instead of averaging over all of these values simultaneously, we follow a Markov process of picking $\sigma_N,\dots \sigma_1$ one after another. At each step, there will be $M = N,\dots 0$ remaining $\sigma$. We get a transition $ n \Rightarrow n-1$ with probability $\frac{n}{M}$ and $ n \Rightarrow n$ with complementary probability.

Multiplying these probabilities and summing all histories of the Markov process is equivalent to the computation of the product of the Markov matrices
\begin{eqnarray}
    &&\prod_{M=1}^{N} Q(M);\\
    &&  Q(M) \lvert n\rangle = \frac{M-n}{M} \lvert n\rangle + \frac{n}{M}\lvert n-1\rangle;
\end{eqnarray}
This Markov process will be random until $n=0$. After that, all remaining $\sigma_k$ will have negative signs and be taken with probability 1, keeping $n=0$.

The expectation value of some function $\hat{X}(\left\{\sigma\right\})$ reduces to the matrix product

   \begin{eqnarray}
   &&\mathbb{P}[\hat{X}] = \sum_{n=0}^{N_+} \nonumber\\
    &&\left\langle n \right\rvert \left(\prod_{M=1}^{N} \hat{Q}(M)\right)\cdot \hat{X} \cdot \left \lvert N_+\right\rangle;\\
   && \hat{Q}(M)\cdot \hat{X}\lvert n\rangle = \frac{n}{M} \hat{X}\left(\sigma_{M}\to 1\right)\lvert n-1\rangle + \nonumber\\
    &&\frac{M-n}{M} \hat{X}\left(\sigma_{M}\to -1\right)\lvert n\rangle
\end{eqnarray} 

Here $N_+ = (N + \sum \sigma_l)/2 = (N + q r)/2$ is the number of positive sigmas.
The operator $\hat{Q}(M)$ sets in $\hat{X}\lvert n\rangle$ the variable $\sigma_{M}$ to $1$ with probability $\frac{n}{M}$ and to $-1$ with complementary probability.  The generalization of the Markov matrix $Q(M)$ to the operator $\hat{Q}(M)$ will be presented shortly.

Once the whole product is applied to $ \hat{X}$, all the sigma variables in all terms will be specified so that the result will be a number.

This Markov process is implemented as a computer code in \cite{BulMig}, leading to a fast simulation with $O(N^0)$  memory requirement.

Now, we observe that quantum Fermi statistics can represent the Markov chain of Ising variables.
Let us construct the operator $\hat{Q}(M)$ with Fermionic creation and annihilation operators, with occupation numbers $\nu_k = (1 + \sigma_k)/2 = (0,1)$. These operators obey (anti)commutation relations, and they create/annihilate $\sigma=1$ state as follows (with Kronecker delta $\delta[n] \equiv \delta_{n,0}$):

    \begin{eqnarray}
    &&\left[a_i, a^\dagger_j\right]_+ = \delta_{i j};\\
    &&\left[a_i, a_j\right]_+ =\left[a^\dagger_i, a^\dagger_j\right]_+ =0;\\
    &&a^\dagger_n\lvert\sigma_1,\dots,\sigma_N\rangle = \nonumber\\
    &&\delta[\sigma_n+1]\lvert\sigma_1,\dots, \sigma_n \to 1,\dots\sigma_N\rangle;\\
    &&a_n\lvert\sigma_1,\dots,\sigma_N\rangle =\nonumber\\
    && \delta[\sigma_n-1]\lvert\sigma_1,\dots, \sigma_n \to -1,\dots\sigma_N\rangle;\\
    && \hat{\nu}_n  = a^\dagger_n a_n;\\
    && \hat{\nu}_n\lvert\sigma_1,\dots,\sigma_N\rangle =\delta[\sigma_n-1]\lvert\sigma_1,\dots,\sigma_N\rangle 
\end{eqnarray}

The number $n(M)$ of positive sigmas $\sum_{l=1}^M\delta[\sigma_l -1]$ coincides with the occupation number of these Fermi particles. 
\begin{eqnarray}
    && \hat n(M) = \sum_{l=1}^M \hat{\nu}_l;
\end{eqnarray}
This relation leads to the representation
\begin{eqnarray}
    && \hat{Q}(M) = \hat{\nu}_M \frac{\hat n(M)}{M}  + (1 - \hat{\nu}_M)\frac{M-\hat n(M)}{M} ;
\end{eqnarray}
The variables $\sigma_l$ can also be expressed in terms of this operator algebra by using 
\begin{eqnarray}
   \hat\sigma_l = 2 \hat{\nu}_l -1.
\end{eqnarray}
The Wilson loop in \eqref{WilsonLoop} can now be represented as an average over the small Euler ensemble $\mathcal{E}(N)$ of a quantum trace expression

    \begin{eqnarray}
\label{trace formula}
&&\Psi(\gamma, C)= \frac{\VEV{\hat{W}[C]}_{\hat{\Omega} ,\mathcal{E}(N)}}{\VEV{\hat{W}[0]}_{\hat{\Omega} ,\mathcal{E}(N)}};\\
&&\hat{W}[C] =\tr{\left(\hat Z(q r)\exp{\frac{i \gamma \hat \Gamma}{\nu} } \prod_{M=1}^{N} \hat{Q}(M)\right)},\\
&& \hat \Gamma = \sum_l\Delta\vec{C}_l \cdot \hat{\Omega} \cdot  \vec{\mathcal P}_l(t);\\
&&  \hat Z(s) =\oint \frac{d \omega}{2 \pi} \exp{\imath \omega\left(\sum_l \hat\sigma_l - s\right)};\\
 && \Delta\vec{C}_l = \vec{C}\left(\frac{l+1}{N}\right) -\vec{C}\left(\frac{l}{N}\right),\\
&& \vec{\mathcal P}_l(t)=\sqrt{\frac{\nu}{2\left(t+t_0\right)}} \frac{\vec{\mathcal F}_l}{\gamma} , \quad \hat{\Omega} \in O(3) ,\\
&& \vec{\mathcal F}_l=\frac{\left\{\cos \left(\hat \alpha_l\right), \sin \left(\hat\alpha_l\right), 0\right\}}{2 \sin \left(\frac{\beta}{2}\right)} , \\
&& \mathcal{E}(N):\quad p,q,r \in \mathbb{Z} \quad -N\leq qr\leq N ,\nonumber\\
&&\textrm{with} \quad 0<p<q<N, \textbf{ gcd}(p,q) =1,  \\
&& \hat\alpha_{l}= \beta\sum_{k=1}^{l-1} (2 \hat{\nu}_k -1);
\end{eqnarray}

The last component of the vector $\vec{\mathcal F}_l$ is set to $0$ as this component does not depend on $l$ and yields zero in the sum over the loop $\sum_l \Delta \vec{C}_l =0$.

The proof of equivalence to the combinatorial formula with an average over $\sigma_l = \pm 1$ can be given using the following Lemma (obvious for a physicist).

\begin{lemma}  
The operators $\hat{\nu}_l $ all commute with each other.
\end{lemma}
\begin{proof}
    Using commutation relations, we can write 
    \begin{eqnarray}
        &&\hat{\nu}_l \hat{\nu}_n = a^\dagger_l (\delta_{l n} - a^\dagger_n a_l) a_n =\nonumber\\
        &&a^\dagger_l a_n \delta_{l n} - a^\dagger_l  a^\dagger_n  a_l  a_n 
    \end{eqnarray}
    Interchanging indexes $l, n$ in this relation, we see that the first term does not change due to Kronecker delta, and the second term does not change because $a^\dagger_l, a^\dagger_n$ anti-commute, as well as $a_l, a_n $, so the second term is symmetric as well.
    Therefore, $\hat{\nu}_l \hat{\nu}_n = \hat{\nu}_n \hat{\nu}_l $
\end{proof}
\begin{theorem}
    The trace formula \eqref{trace formula} equals the expectation value of the momentum loop ansatz \eqref{LoopSol}, \eqref{PFT}, \eqref{Fsol} in the big Euler ensemble.
\end{theorem}
\begin{proof}
As all the operators $\hat{\nu}_l$ commute with each other, the operators $\hat{Q}(M)$ can be applied in arbitrary order to the states $\Sigma =\lvert\sigma_1,\dots \sigma_N\rangle$ involved in the trace. The same is true about individual terms in the circulation in the exponential of the Wilson loop. These terms $\vec{\mathcal F}_l$ involve the operators $\hat \alpha_l$, which commute with each other and with each $\hat{Q}(M)$. Thus, we can use the ordered product of the operators $ \hat G_l =\hat{Q}(l) \exp{\imath \omega \hat\sigma_l + \frac{i \gamma}{\nu} \Delta \vec{C}_l \cdot \vec{\mathcal P}_l(t)}$. 
    Each of the operators $\hat G_l $ acting in turn on arbitrary state $\Sigma$ will create two terms with $\delta[\sigma_l\pm 1]$. 
    The exponential in $\hat G_l$ will involve $\hat \sigma_k , k \le l$.
    As a result of the application of the operator $\hat Z_l =\prod_{k=1}^l \hat G_k $ to the state vector $\Sigma$ we get $2^l$ terms with $\Sigma\prod_{k=1}^l \delta[\sigma_k- \eta_k], \eta_k = \pm 1$. The factors $\hat Z_l$ will involve only $\hat \sigma_k, k \le l$, which are all reduced to $\eta_k, k \le l$ in virtue of the product of the Kronecker deltas.
    Multiplying all operators $\hat G_M$ will lead to superposition $\hat \Pi_N$ of $2^N$ terms, each with product $ \prod_{M=1}^N \delta[\sigma_M-\eta_M]$ with various choices of the signs $\eta_i=\pm 1$ for each $i$.
    Furthermore, the product of Kronecker deltas will project the total sum of $2^N$ combinations of the states $\Sigma$ in the trace $\tr{\dots}$ to a single term corresponding to a particular history $\eta_1,\dots\eta_N$ of the Markov process. 
    The product of Kronecker deltas in each history will be multiplied by the same state vector $\Sigma$, by the product of Markov transition probabilities, and by the exponential $\exp{ \frac{i \gamma}{\nu} \sum_l\Delta \vec{C}_l \cdot \hat{\Omega} \cdot \vec{P}_l(t)}$ with the operators $\hat \sigma$ in $\vec{\mathcal P}_k(t)$ replaced by numbers $\eta$ leading to the usual numeric $\vec{P}_l(t)$.
    The transition probabilities of the Markov process are designed to reproduce combinatorial probabilities of random sigmas, adding up to one after summation over histories \cite{Norris_2007}.
    The integration over $\omega$ will produce $\delta\left[\sum_l \hat\eta_l - s\right]$.
    This delta function will reduce the trace to the required sum over all histories of the Markov process with a fixed $\sum_l \eta_l$.
\end{proof}

We have found a third vertex of the triangle of equivalent theories: the decaying turbulence in three-dimensional space, the fractal curve in complex space, and Fermi particles on a ring. 
By degrees of freedom, this is a one-dimensional Fermi-gas in the statistical limit $N \to \infty$. However, there is no local Hamiltonian in this quantum partition function, just a trace of certain products of operators in Fock space. So, an algebraic (or quantum statistical) problem remains to find the continuum limit of this theory of the fermion ring.
\section{Path integral over Markov histories}\label{PathIntegral}
Let us represent the product $ \Pi_N$ of the transitional probabilities of the particular history of the Markov processes as follows (with $n_\pm \equiv n_\pm(l), \Delta n_\pm = -1$)
    \begin{eqnarray}
   &&  \Pi_N = \exp{ N \Lambda_N};\\
   &&  \Lambda_N = \frac{1}{N}\sum_l  G_l; \\
   && G_l = \Delta n_+ \log\left(\frac{n_+}{n_+ + n_-}\right) +\nonumber\\
   &&\Delta n_- \log\left(\frac{n_-}{n_+ + n_-}\right);\\
   && n_+ = \sum_{k \le l} \nu_k;\\
   && n_- = \sum_{k \le l} (1-\nu_k);
\end{eqnarray}
These $n_\pm$ are net numbers of $\eta = \pm1$ in terms of Ising spins or occupation numbers $\nu_k = (1,0)$ in the Fermi representation.
There is an extra constraint on the Markov process
\begin{eqnarray}\label{constraint1}
    n_+ + n_- = l ;\;\forall l
\end{eqnarray}
which follows from the above definition in terms of the occupation numbers.
We can redefine $n_\pm$ as $N$ times the piecewise constant functions. 

    \begin{eqnarray}
    &&n_\pm = N f_\pm(\xi);\\
    && f_\pm(\xi) = \sum_{k =1}^{\floor{N \xi}} \frac{\nu_k}{N};\\
    && f'_\pm(\xi) = \sum_{l=1}^N\delta\left(\xi -\frac{l}{N}\right) \sum_{k =1}^{l} \frac{\nu_k}{N} ;\\
    && 0 < \xi < 1;
\end{eqnarray}

The sums can be rewritten as Lebesgue integrals
    \begin{eqnarray}
    &&\Lambda_N = \int_1^0 \nonumber\\
    &&\left(d f_+(\xi) \log\left(\frac{f_+(\xi)}{\xi}\right) +d f_-(\xi) \log\left(\frac{f_-(\xi)}{\xi}\right)\right) 
\end{eqnarray}

The sum over histories of the Markov process will become a \textbf{path integral} over the difference $\phi =f_+(\xi)-f_-(\xi) $ 

    \begin{eqnarray}\label{func int phi}
   && \sum_{\eta_.=\pm1}\exp{ N (\Lambda_N + \imath \Lambda_N^{(1)}) } \nonumber\\
   &&\to\int D\phi \exp{ N (\Lambda_N + \imath \Lambda_N^{(1)}) }
\end{eqnarray}

This path integral will be dominated by the "classical history," maximizing the product of transitional probabilities if such a classical trajectory exists.
\pctPDF{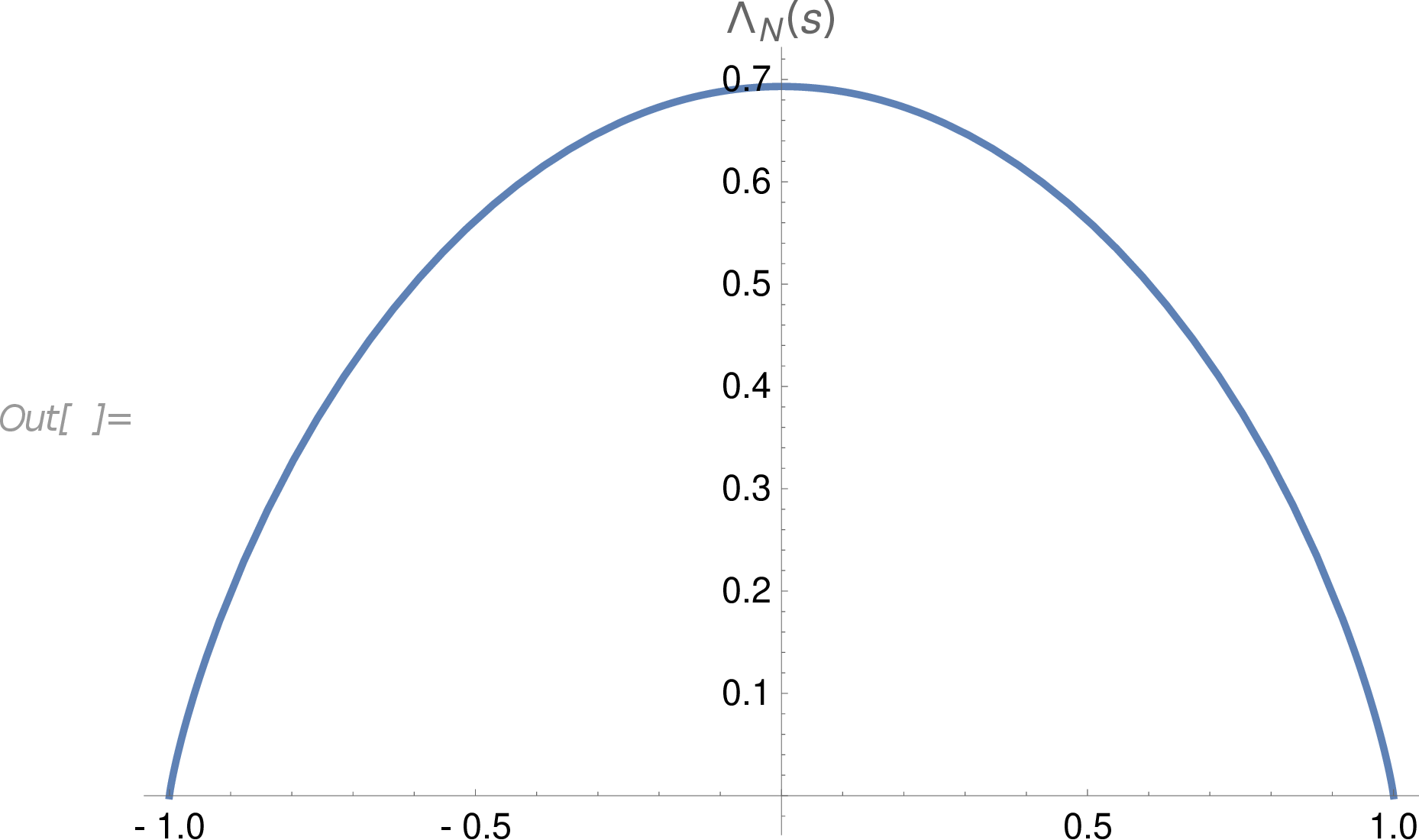}{The plot of the function $\Lambda(s)$. As required, it is positive, takes a maximal value $\log (2)$ at $s=0$, and vanishes at both ends $s = \pm 1$ of the physical region.}
The first term (without the circulation term) brings the variational problem

   \begin{eqnarray}
&&\max_{\phi}\Lambda_N[\phi];\\
 && \Lambda_N[\phi] =\int_1^0 d \xi \nonumber\\
 &&\left(\dbyd{f_+}{\xi}
\log\left(\frac{f_+}{\xi}\right) + \dbyd{f_-}{\xi}  \log\left(\frac{f_-}{\xi}\right) \right);\\
   && f_\pm(\xi) = \frac{1}{2} \left(\xi \pm \phi(\xi)\right);\\
   && f_\pm(\xi) \ge 0;
\end{eqnarray} 

This problem is, however, a degenerate one, as the functional reduces to the integral of the total derivative:
   \begin{eqnarray}
   \label{degenerate}
  && \fbyf{\Lambda_N[\phi]}{\phi(\xi)} =0;\\
  &&\Lambda_N[\phi]  =\int_1^0 d \left(f_+\log f_+ + f_-\log f_- \right)  +1 + \int_0^1 d \xi \log\xi=\nonumber\\
  &&-\frac{1}{2} (1-s) \log (1-s)-\frac{1}{2} (1+s) \log (1+s)+\log (2);
\end{eqnarray} 

It depends on the boundary condition $\phi(0) = 0, \phi(1) = s$ but not on the shape of $\phi(\xi)$.

This expression matches the Stirling formula for the logarithm of the binomial coefficient in the combinatorial solution  \cite{migdal2023exact} for the Euler ensemble 
    \begin{eqnarray}
   &&\lim_{N \to \infty}\frac{\log \Binom{N}{N(1+s)/2}}{N} =\log (2) +\nonumber\\
   &&-\frac{1}{2} (1-s) \log (1-s)-\frac{1}{2} (1+s) \log (1+s)
\end{eqnarray}

This $\Lambda(s) = \Lambda_\infty(s)$ is a smooth even function of $s$ taking positive values from $\Lambda(\pm1) =0$ to the maximal value $\Lambda(0)= \log (2)$ (see Fig.\ref{fig::ZNS}).

Now, let us add the circulation term to the exponential of the partition function \eqref{trace formula}.
This term can be directly expressed in terms of the difference between our two densities $N \phi(\xi) = N f_+(\xi) - N f_-(\xi)$:
\begin{eqnarray}\label{Lambda1}
&& \imath N\Lambda^{(1)}_N[\phi, C_\Omega] = \frac{\imath }{\sqrt{2\nu t}} \int_0^1  d\vec{C}_\Omega(\xi) \cdot \vec{F}(\xi);\\
&& \vec{F}(\xi) = \frac{\left\{\sin(\beta N\phi(\xi)),\cos(\beta N\phi(\xi)),0 \right\}}{2 \sin(\beta/2)};\\
&& \vec{C}_\Omega(\theta) = \hat{\Omega}\cdot \vec{C}(\theta);
\end{eqnarray}
We remember that the last component of the vector $\vec{F}(\xi)$ does not contribute to the circulation integral in \eqref{Lambda1} with the closed loop $\vec{C}_\Omega(\xi)$.  This is why we replaced it with zero, not because it is small but because it drops.
The key assumption is, of course, the existence of the smooth limit of the charge density $\phi(\xi)$ of these fermions when they are densely covering this loop.

We are working with $\alpha(\xi) = \beta N \phi(\xi)$ in the following.

The measure for paths $[D \alpha] $ is undetermined. The derivatives of these alphas were quantized in the original Fermi theory: each step $ \alpha'(\xi) \approx N \Delta \alpha = N\beta \sigma = \pm N\beta$. 

As we demonstrate below, in continuum theory, this discrete distribution can be replaced by a Gaussian distribution with the same mean square
\begin{eqnarray}
    \sum\displaylimits_{\alpha' = \pm N\beta} \leftrightarrow \int d  \alpha' \exp{ -  \frac{(\alpha')^2}{2 N^2 \beta^2}}
\end{eqnarray}
To demonstrate that, we consider in the critical region $\beta^2 \sim N^{-1} \to 0$ the most general term that arises in the moments of the circulation
in \eqref{Lambda1} (see \cite{BZ23} for some exact computations of these moments)

    \begin{eqnarray}
    &&2^{-N}\sum\displaylimits_{\sigma_i = \pm 1} \exp{ \imath \beta \sum_i  k_i  \sigma_i} \nonumber\\
    &&=\prod_i \cos \beta  k_i \to  \exp{ -  \beta^2/2\sum_i k_i^2}
\end{eqnarray} 

where $k_i$ are some integers. With a large number $N$ of these integers, the sum in the exponential becomes an integral, which is equivalent to a Gaussian integral
 \begin{eqnarray}
    &&\exp{ -  \beta^2/2\sum_i k_i^2} =\nonumber\\
    &&\prod_i \int_{-\infty}^\infty \frac{d \sigma_i}{\sqrt{2 \pi}} \exp{ -\sigma_i^2/2} \exp{ \imath \beta k_i  \sigma_i} \nonumber\\
    &&\to \exp{  - N \beta^2/2\int_0^1 d \xi k(\xi)^2}
\end{eqnarray}   

We arrive at the standard path integral measure
\begin{eqnarray}  
    \label{gaussMeasure}
  &&\int [D \alpha]  =  \int D \alpha(\xi) \exp{- \int_0^1 d \xi \frac{(\alpha')^2}{2 N \beta^2}};\\
  && \frac{\int [D \alpha] \exp{ \imath N \int_0^1 d \xi\alpha(\xi) K(\xi)}}{\int [D \alpha]} = \nonumber\\
  &&\exp{-\frac{N^2 \iint d\xi_1 d\xi_2 K(\xi_1)K(\xi_2)G_{1,2}}{2}};\\
  && G_{1,2} = \VEV{\alpha(\xi_1) \alpha(\xi_2)};
\end{eqnarray}

The next Appendix will compute this Green's function $G_{1,2}=G(\xi_1, \xi_2)$.

 Thus, we arrive at the following path integral in the  continuum limit
 \begin{subequations}\label{pathIntegralSol}
          \begin{eqnarray}
\label{PsiSol}
    &&\Psi[C] = \frac{\displaystyle\sum\displaylimits_{p<q; (p,q)} \int\displaylimits_{\Omega \in O(3)} d \Omega\int [D \alpha]}{\displaystyle\sum\displaylimits_{p<q; (p,q)}|O(3)|\int [D\alpha]}\nonumber\\
    &&\exp{ \imath\frac{\int_0^1 d \xi \Im\left( \mathcal C'_\Omega(\xi)
    \exp{\imath\alpha(\xi)}\right)}{2 \sin(\pi p/q) \sqrt{2\nu (t+ t_0)}}};\\
    && \mathcal C_\Omega(\theta) =  \vec{C}(\theta) \cdot \hat{\Omega} \cdot\{\imath,1,0\};
\end{eqnarray}
 \end{subequations}

We get the $U(1)$  statistical model with the boundary condition $\alpha(1) = \alpha(0) + \beta N s$. The period $ \beta N s = 2 \pi p r$ is a multiple of $2 \pi$, which is irrelevant at $N \to \infty$. The effective potential for this theory is a linear function of the loop slope $\vec{C}'(\xi)$.

This model is yet another representation of the Euler ensemble, suitable for the continuum limit.

\section{Matching path integral with combinatorial sums}
The results of the path integration over $\alpha$ must match the combinatorial calculations with $\sigma_l = \pm 1$ in the limit of large $N$. Without the interaction provided by the circulation term in \eqref{PsiSol}, this path integral is dominated by a linear trajectory
\begin{eqnarray}
    \alpha_{cl}(\xi) = \beta N \xi s;
\end{eqnarray}

We already saw the match between the classical Action $\Lambda_N[\phi(\xi) = \xi s]$ and the asymptotic value of the logarithm of the Binomial coefficient of the combinatorial solution for the sum over $\sigma$ variables.

Let us verify some examples of the expectation values over $\sigma$. The simplest is (with $ n \neq m$)
\begin{eqnarray}
    \VEV{\sigma_n \sigma_m}_{\sum \sigma = N s}
\end{eqnarray}
The direct calculation using methods of \cite{migdal2023exact, BZ23} leads to
   \begin{eqnarray}
   && \VEV{\sigma_n \sigma_m}_{\sum \sigma = N s} = -\oint \frac{d \omega}{2 \pi Z} e^{\imath \omega N s} \cos((N-2) \omega) \sin^2(\omega) = \nonumber\\
   &&\frac{ (1-N s^2)}{2^{N-3} Z N(1-s^2)} \Binom{N-2}{\frac{1}{2} (N+N s-2)};\\
   && Z = \oint \frac{d \omega}{2 \pi} e^{\imath \omega N s} \cos(N \omega) =
   2^{-N} \binom{N}{\frac{1}{2} (s N+N)}
\end{eqnarray} 
Using Gamma function properties, this ratio simplifies to
\begin{eqnarray}
    \frac{N s^2-1}{N-1} 
\end{eqnarray}
This result can be derived from symmetry without any integration. 
    \begin{eqnarray}
    &&\VEV{\sigma_n \sigma_m}_{\sum \sigma = N s} =\nonumber\\
    &&A(N,s) (1- \delta_{n m}) + \delta_{n m};\\
    && \sum_{n, m} \VEV{\sigma_n \sigma_m}_{\sum \sigma = N s} = s^2 N^2 = A(N,s) N(N-1) + N\\
    &&A(N,s) = \frac{N s^2-1}{N-1} 
\end{eqnarray}
The same limit $A(\infty,s) = s^2$ follows from the classical trajectory
\begin{eqnarray}
    \VEV{\sigma_n \sigma_m}_{\sum \sigma = N s} \to \frac{\alpha_{cl}'(\xi)}{\beta N}\frac{\alpha_{cl}'(\xi')}{\beta N} = s^2
\end{eqnarray}

Let us consider less trivial example \cite{migdal2023exact, BZ23}
\label{Uvar}
\begin{eqnarray}
    &&U_{n,m} \to  \sum_{k=n}^{m-1} \exp{\imath  \alpha_{k,n}} ;\\
    && \alpha_{k,n} = \beta\sum_{\substack{l=0\\ l \neq n}}^{k} \sigma_l;
\end{eqnarray}
We shall set $s=0$, as this is the leading contribution to the partition function.
The expectation value of $U_{n,m}$ in our continuum limit becomes
   \begin{eqnarray}
    &&\VEV{U_{n,m}} = N  \int_{\xi_1}^{\xi_2} d \xi \VEV{\exp{\I \alpha(\xi)}} =\nonumber\\
    &&N \int_{\xi_1}^{\xi_2} d \xi \exp{- \oh G(\xi,\xi)}
\end{eqnarray} 

Here $G(\xi_1, \xi_2)$ is the Green's function corresponding to a 1D particle on a line interval $ \xi \in(0,1)$, introduced in the previous section.
It satisfies the equation, which follows from our Gaussian Action
\begin{eqnarray}
    &&\partial_\xi^2 G(\xi, \xi') = -\beta^2 N \delta(\xi-\xi');\\
    &&G(0, \xi') = G(\xi, 0) =0;
\end{eqnarray}
The solution is
\begin{eqnarray}
    G(\xi, \xi') = \oh \beta^2 N \left(\xi + \xi' - |\xi-\xi'|\right)
\end{eqnarray}
Thus, we find
   \begin{eqnarray}
    &&\VEV{U_{n,m}} = N \int_{\xi_1}^{\xi_2} d \xi \exp{-\oh \beta^2 N\xi} 
   =  \nonumber\\
    &&\frac{2} {\beta^2} \left(\exp{- y/2} - \exp{- x/2}\right);\\
   && x = \beta^2 N\xi_1;\\
   && y = \beta^2 N\xi_2;
\end{eqnarray} 

in agreement with \cite{migdal2023exact, BZ23} in the critical region $N \to \infty ,\beta^2 \sim 1/N$.
Finally, the expectation value
\begin{eqnarray}
    \VEV{U_{n,m}\bar U_{n,m}} = \sum_{l=n}^{m-1} \sum_{k=n}^{m-1} \VEV{\exp{\I \alpha_{k n} - \I \alpha_{l n}}}
\end{eqnarray}
Here, the Gaussian path integration yields

    \begin{eqnarray}
    \label{Uaverage}
    &&\VEV{U_{n,m}\bar U_{n,m}} \to N^2 \int_{\xi_1}^{\xi_2} d \xi \int_{\xi_1}^{\xi_2}d \xi' \nonumber\\
    &&\exp{-\oh \left(G(\xi, \xi) + G(\xi', \xi') - 2 G(\xi, \xi')\right)} =\nonumber\\
    && N^2 \int_{\xi_1}^{\xi_2} d \xi \int_{\xi_1}^{\xi_2}d \xi'\exp{\oh \beta^2 N | \xi - \xi'|}=\nonumber\\
    && \frac{4}{\beta^4} \left(2\exp{ (y-x)/2} + x -y -2\right)
\end{eqnarray}

This result also agrees with combinatorial computations in \cite{migdal2023exact, BZ23}.
\section{Small Euler ensemble in statistical limit}\label{smallEuler}
The remaining problem is averaging over the variables $N,p,q,r$ of the small Euler ensemble.

The variable $s = \frac{q r}{N}$ is  distributed between $-1,1$ with the binomial weight \cite{migdal2023exact, BZ23} $\Binom{N}{N(1+s)/2}$ peaked at $s =0$. There is a finite term coming from $ r =0$ plus a continuum spectrum coming from large $r$
\begin{eqnarray}
    W(r) = \begin{cases}
        1   &\text{ if } r=0;\\
        \frac{\sqrt{ 2 \pi N}}{q}\exp{- \frac{(q r)^2}{2 N}} &\text{otherwise};
    \end{cases}
\end{eqnarray}
As it was conjectured in \cite{migdal2023exact} and supported by rigorous estimates in \cite{BZ23}, the $r =0$ term dominates the sums, after which the variables $y, x $ can be treated as continuous variables.

The variable $p$ at fixed $q$ has a discrete distribution
\begin{eqnarray}
    f_p(p\mid q) = \frac{\displaystyle \sum_{\substack{p=1 \\ (p,q)}}^{q-1} \delta(p-n) }{\varphi(q)}
\end{eqnarray}
As we shall see, rather than $p$, we would need an asymptotic distribution of a scaling variable
\begin{eqnarray}\label{Xvar}
    X(p,q) = \frac{1}{q^2} \cot^2\left(\frac{\pi p}{q}\right)
\end{eqnarray}
This distribution for $X(p,q)$ at fixed $q \to \infty$ can be found analytically, using newly established relations for the cotangent sums (see Appendix in \cite{migdal2023exact}, and \Mathematica{} notebook \cite{MB40}). Asymptotically, at large $q$, these relations read

    \begin{eqnarray}
   &&\VEV{X^n} \equiv \lim_{q \to \infty} \frac{\displaystyle\sum_{\substack{p=1 \\ (p,q)}}^{q-1} X(p,q)^n}{q} = \nonumber\\
    &&\delta_{n,0} + \frac{2 \pi ^{-2 n}  \zeta (2 n)}{(2 n+1) \zeta (2 n+1)}
\end{eqnarray}

This relation can be transformed further as

    \begin{eqnarray}
    &&\VEV{X^n}  = \nonumber\\
    &&
    \begin{cases}
        1 &\text{if } n=0\\
        \frac{\displaystyle 2  \sum_{k=1}^{\infty} \varphi(k)k^{-(2 n +1)}}{(2 n+1)\pi ^{2 n} } &\text{if } n >0
    \end{cases}
\end{eqnarray}

The Mellin transform of these moments leads to the following singular distribution

    \begin{eqnarray}\label{CotDist}
&& \VEV{X^n} = \int_0^\infty f_X(X)\,d X\, X^n;\\
\label{WprimeX}
   && f_X(X)= (1-\alpha)\delta(X) + \pi X\sqrt{X}\Phi\left(\floor*{\frac{1}{\pi \sqrt{X}}}\right);\\
   && \alpha = \pi \int_0^\infty X \sqrt{X} d X \Phi\left(\floor*{\frac{1}{\pi \sqrt{X}}}\right)=\nonumber\\
   && \frac{2\pi}{5} \sum_1^\infty \Phi(k) \left(\frac{1}{(\pi k)^5} - \frac{1}{(\pi (k+1))^5}\right) = \nonumber\\
    &&\frac{2}{5 \pi^4} \sum_1^\infty \frac{\varphi(k)}{k^5} = \frac{1}{225 \zeta (5)}
\end{eqnarray}

where $\Phi(n)$ is the totient summatory function
\begin{eqnarray}
    \label{PhiDef}
    \Phi(q) = \sum_{n=1}^q \varphi(n)
\end{eqnarray}

The distribution can also be rewritten as an infinite sum

   \begin{eqnarray}
    &&\int d x f_X(x) F(x) = (1-\alpha)F(0) + \nonumber\\
    &&\pi \sum_{n=1}^\infty \varphi(n) \int\displaylimits_{0}^{\frac{1}{\pi^2 n^2}} x^{\frac{3}{2}} \, d x F(x)
\end{eqnarray} 

The normalization of this distribution comes out $1$ as it should, with factor $1-\alpha$ in front of the delta function.

\pctPDF{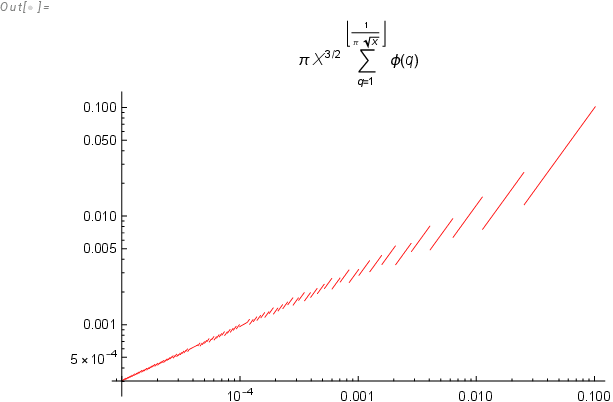}{ The log-log-plot of the distribution $f_X(X)$.  It is equals $\pi\Phi(k) X\sqrt{X}$ at each internal $\frac{1}{\pi^2 (k+1)^2} < X < \frac{1}{\pi^2 k^2}$ with positive integer $k$. Asymptotically $f_X(X) \to \frac{3 \sqrt{X}}{\pi^3 }$ at $ X \to 0$.}

The upper limit of $X$
\begin{eqnarray}
    X_{max} = X\left(q -1, q\right) \to \frac{1}{\pi^2}
\end{eqnarray}
Our distribution \eqref{CotDist} is consistent with this upper limit, as the argument $\floor*{\frac{1}{\pi \sqrt{X}}}$ becomes zero at $X \pi^2 > 1 $.
It is plotted in Fig. \ref{fig::PiPhi}.

Once we are zooming into the tails of the $p,q$ distribution, we also must recall that
\begin{eqnarray}
  &&  \mathbb{P}(q < y N) = \frac{\Phi\left( \floor*{N y}\right)}{\Phi\left( N \right)}\to y^2;\\
  \label{WprimeY}
  && f_y(y) = \frac{\displaystyle\sum_{q=2}^{N} \delta\left( y - \frac{q}{N}\right) \varphi(q)}{\Phi\left( N \right)} 
\end{eqnarray}

\section{The velocity correlation in Fourier space}\label{VelCorr}
Let us outline an analytical solution. We shift the time variable by $t+t_0\Rightarrow t$ to simplify the formulas.

The correlation function reduces to the following average over the big Euler ensemble $\mathbb{E}$ of our random curves in complex space \cite{migdal2023exact}
    \begin{eqnarray}\label{CorrFunc}
   && \VEV{\vec\omega(\vec 0) \cdot \vec \omega(\vec{r})} =  \int_{O(3)} \frac{d \Omega}{4 t^2|O(3)|}\nonumber\\
   &&\sum_{0\le n<m< N}\VEV{\vec \omega_m \cdot \vec \omega_n e^{\imath\vec \rho\cdot \hat{\Omega} \cdot \left( \vec S_{n,m}- \vec S_{m,n} \right) }}_{\mathbb E};\\
   && \vec S_{n,m} = \frac{\sum_{k=n}^{m-1}  \vec{F}_k}{m  -n \pmod{N}};\\
   &&\vec \omega_k = \left\{0,0,\frac{\imath \sigma_k}{2} \cot\left(\frac{\beta}{2}\right)\right\};\\
   && \vec \rho = \frac{\vec{r}}{2\sqrt{\nu t}};\\
   && \VEV{X[\sigma...]}_{p,q,r}\equiv \frac{\sum_{\mathbb{E}} X[\sigma...]\delta[q r-\sum \sigma]}{\sum_{\mathbb{E}} \delta[q r-\sum \sigma]};
\end{eqnarray}
Integrating the global rotation matrix $O(3)$ is part of the ensemble averaging.

Let us apply our path integral to the expectation value over spins $\sigma = \pm 1$ in the big Euler ensemble, with the distribution of $q, X$ established in the previous section.
In the continuum limit, we replace summation with integration. We arrive at the following expression for the correlation function:
    \begin{eqnarray}
    &&\VEV{\vec\omega(\vec 0) \cdot \vec \omega(\vec{r})} \propto \nonumber\\
    &&\sum_{\text{even } q<N} \sum_{p;\, (p|q)}\; \frac{\cot^2(\pi p/q)}{(p/q)^2}\int\displaylimits_{0 < \xi_1 < \xi_2 < 1} d \xi_1 d \xi_2\int_{O(3)} d \Omega\nonumber\\
    && 
    \frac{\int [D \alpha]\alpha'(\xi_1) \alpha'(\xi_2)e^{\imath \frac{\vec{r} \cdot \hat{\Omega} \cdot \Im  \vec V(\xi_1,\xi_2)}{\sqrt{\nu  t}} }}
   { t^2 \Phi(N)|O(3)| \int [D\alpha]};\\
   && \vec V(\xi_1,\xi_2) = \nonumber\\
    && q \sqrt{X}\left\{ \imath, 1,0\right\}\left(S(\xi_1, \xi_2)-   S(\xi_2, 1+\xi_1) \right);\\
   && S(a,b)  =\frac{\int_{a}^{b} d \xi e^{\imath \alpha(\xi)}}{ b-a};
\end{eqnarray}

Here and in the following, we skip all positive constant factors, including powers of N. Ultimately, we restore the correct normalization of the vorticity correlation using its value at $\vec{r} =0$ computed in previous work \cite{migdal2023exact}. 

The computations significantly simplify in Fourier space.
    \begin{eqnarray}
\label{wwK}
    &&\VEV{\vec\omega(\vec 0) \cdot \vec \omega(\vec{k})} = \int d^3 \vec{r} \VEV{\vec\omega(\vec 0) \cdot \vec \omega(\vec{r})} \exp{- \imath \vec{k} \cdot \vec{r}}\propto\nonumber\\
    &&\sum_{\text{even } q<N} \sum_{p;\, (p|q)}\; \frac{\cot^2(\pi p/q)}\int_{O(3)} d \Omega{(p/q)^2}\int\displaylimits_{0 < \xi_1 < \xi_2 < 1} d \xi_1 d \xi_2\nonumber\\
    &&
    \frac{\displaystyle\int [D \alpha]\alpha'(\xi_1) \alpha'(\xi_2)\delta\left(\frac{\vec{r} \cdot \hat{\Omega} \cdot \Im  \vec V(\xi_1,\xi_2)}{\sqrt{\nu  t}} - \vec{k}\right)}{ t^2 \Phi(N)|O(3)| \int [D\alpha]}
\end{eqnarray}

The angular integration $\int d \Omega$ yields

 \begin{eqnarray}
   && \int_{O(3)} d \Omega\delta\left(\frac{ \hat{\Omega} \cdot \Im  \vec V(\xi_1,\xi_2)}{\sqrt{\nu  t}}- \vec{k}\right) \nonumber\\
    &&\propto \frac{\sqrt{\nu  t}}{\vec{k}^2} \delta\left(  \abs{\Im  \vec V(\xi_1,\xi_2)} - |\vec{k}|\sqrt{\nu  t}\right)
\end{eqnarray}

Now, using the Lagrange multiplier $\lambda$ for this condition, 
and shifting integration over $\lambda$ to the real axis, we have to minimize effective action
  \begin{eqnarray}
   && A[\alpha,\lambda] = \frac{\pi y^2 X}{2} \int_{\xi_1}^{1+\xi_1} (\alpha')^2 + \nonumber\\
    &&\lambda y \sqrt{X} \abs{\frac{\displaystyle\int_{\xi_1}^{\xi_2} d \xi\; e^{\imath \alpha}}{\xi_2-\xi_1} - 
   \frac{\displaystyle\int_{\xi_2}^{1+\xi_1} d \xi\; e^{\imath \alpha}}{1+\xi_1-\xi_2}};\\
   \label{J12Def}
   && \partial_\lambda A[\alpha,\lambda] = |\vec{k}|\sqrt{\nu  t};
\end{eqnarray}  
 
This variational problem reduces to two pendulum equations
    \begin{eqnarray}
   && \alpha_1'' + \frac{r}{\xi_2-\xi_1} \sin\alpha_1 =0 ; \forall{\xi_1 < \xi < \xi_2}\\
   &&\alpha_2'' + \frac{r}{\xi_2-\xi_1-1}\sin\alpha_2 =0 ; \forall{\xi_2 < \xi < 1+\xi_1}\\
   \label{lambdaR}
   &&r = \frac{\lambda}{\pi y \sqrt{X}  I(r)};\\
   && I(r) = \abs{\frac{\displaystyle\int_{\xi_1}^{\xi_2} d \xi\; e^{\imath \alpha_1}}{\xi_2-\xi_1} - 
   \frac{\displaystyle\int_{\xi_2}^{1+\xi_1} d \xi\; e^{\imath \alpha_2}}{1+\xi_1-\xi_2}};
\end{eqnarray}

The well-known solution is Jacobi amplitude $\text{am}(x\mid u)$, 

    \begin{subequations}
    \begin{eqnarray}
    &&\alpha_1(\xi)= \nonumber\\
    &&2 \text{am}\left(\frac{\xi-\alpha_2}{2} a_1 \mid\frac{r}{a_1^2 (\xi_2-\xi_1) }\right);\\
    &&\alpha_2(\xi)= \nonumber\\
    &&2 \text{am}\left(\frac{\xi-\alpha_2}{2} a_2 \mid-\frac{r}{a_2^2 (1-\xi_2+\xi_1 )}\right);
\end{eqnarray}
\end{subequations}

The free parameters $a_1, a_2, \alpha_1, \alpha_2$ satisfy four equations 
    \begin{subequations}
    \begin{eqnarray}
&&\alpha_1(\xi_2) = \alpha_2(\xi_2);\\
&&\alpha'_1(\xi_2) = \alpha'_2(\xi_2);\\
&&\alpha_1(\xi_1) = \alpha_2(1+\xi_1);\\
&&\alpha'_1(\xi_1) = \alpha'_2(1+\xi_1);
\end{eqnarray}

\end{subequations}
together with the constraint following from the variation of the Lagrange multiplier $\lambda$:
\begin{eqnarray}
&&I(r) = \frac{|\vec{k}|\sqrt{\nu  t} }{y \sqrt{X}}
\end{eqnarray}
\section{Turbulent viscosity and the local limit}\label{TurbVisc}

These five equations, in general, are quite complex, but there is one simplifying property.
In the local limit $N \to \infty$, the remaining effective action at the extremum
    \begin{eqnarray}
    &&N A[\alpha_c, \lambda_c] =\frac{\pi N y^2 X}{2} \nonumber\\
    &&\left(\int_{\xi_1}^{\xi_2} \, d \xi (\alpha'_1(\xi))^2 + \int_{\xi_2}^{1+\xi_1}\, d \xi (\alpha'_2(\xi))^2 \right)
\end{eqnarray}

grows as $N$, unless  both $\alpha_1(\xi) \sim \alpha_2(\xi) \sim N^{-\oh} \to 0$.
In this case, the above constraint can be expanded in $\alpha_1, \alpha_2$. As we show in \cite{MB41}, the leading constant and linear terms both cancel so that the quadratic term remains

    \begin{eqnarray}
\label{quadraticConstraint}
    &&\frac{2|\vec{k}|\sqrt{\nu  t} }{y \sqrt{X}} =\nonumber\\
    &&
    \abs{\int_{\xi_1}^{\xi_2} \frac{d \xi\alpha^2_1(\xi)} { \xi_2-\xi_1} -
\int_{\xi_2}^{1+\xi_1} \frac{d \xi\alpha^2_2(\xi) }{1+ \xi_1-\xi_2}} \nonumber\\
    &&\sim \frac{1}{N}
\end{eqnarray}

This estimate then requires vanishing viscosity in the local limit, at fixed turbulent viscosity
\begin{eqnarray}
   \tilde{\nu} =  \nu N^2   \to\const{}.
\end{eqnarray}

This phenomenon of renormalization of viscosity by a factor of $N^2$ was already observed in our first paper \cite{migdal2023exact}. Our Euler ensemble in the local limit $N\to \infty$ can only solve the inviscid limit of the \NS{} decaying turbulence, with finite $\tilde{\nu}$ acting as a turbulent viscosity.

The desired anomalous dissipation phenomenon takes place in this limit of our theory.
\section{Linearized classical trajectory}\label{ClasTraj}
Returning to the elliptic function solution, we rewrite it in the linearized case at $a_1\sim a_2 \to 0$. This linearization is equivalent to replacing $\sin(\alpha) \to \alpha$ in the differential equation and studying the resulting linear ODE as a boundary problem.
We choose different parametrizations in this linear case
    \begin{eqnarray}
    &&\alpha_1(\xi) = \nonumber\\
    &&a\left(\cos\left(K_1(\xi-\xi_2)\right) + \frac{b}{K_1} \sin\left(K_1(\xi-\xi_2)\right)\right);\\
    &&\alpha_2(\xi) = \nonumber\\
    &&a\left(\cos\left(K_2(\xi-\xi_2)\right) + \frac{b}{K_2} \sin\left(K_2(\xi-\xi_2)\right)\right);\\
    && K_1 = \sqrt{ \frac{r}{\Delta}};\\
    && K_2 = \sqrt{ \frac{r}{\Delta-1}};\\
    && \Delta = \xi_2 - \xi_1;
\end{eqnarray}

In the physical region $0<\Delta <1, r <0$, $K_2$ is real, and $K_1$ imaginary, but the solution stays real.
The matching conditions at $\alpha_1(\xi_2) = \alpha_2(\xi_2), \alpha'_1(\xi_2) = \alpha'_2(\xi_2)$ are identically satisfied with this Anzatz.
The derivative match  $\alpha'_1(\xi_1) = \alpha'_2(1+\xi_1)$ can be solved exactly for $b$
    \begin{eqnarray}
\label{bsol}
   && b = \frac{P}{Q};\\
    && P = \sqrt{\frac{r}{\Delta -1}} \sin \left((1-\Delta ) \sqrt{\frac{r}{\Delta -1}}\right)+\nonumber\\
    &&\sqrt{\frac{r}{\Delta }} \sin \left(\Delta  \sqrt{\frac{r}{\Delta }}\right);\\
    && Q =\cos \left((\Delta -1) \sqrt{\frac{r}{\Delta -1}}\right)-\cos \left(\Delta  \sqrt{\frac{r}{\Delta }}\right)
\end{eqnarray}
The remaining matching condition $\alpha_1(\xi_1) = \alpha_2(1+\xi_1)$ reduces to the root of the function
    \begin{eqnarray}
    &&g(r, \Delta) = \frac{(2 \Delta -1) \sin \left(\sqrt{(\Delta -1) r}\right) \sin \left(\sqrt{\Delta  r}\right)}{\sqrt{(\Delta -1) \Delta }}+\nonumber\\
    &&2 \cos \left(\sqrt{(\Delta -1) r}\right) \cos \left(\sqrt{\Delta  r}\right)-2
\end{eqnarray}

This function has multiple roots, but we are looking for the real root $r_0(\Delta)$ with minimal value of the action at given $\Delta$
\begin{eqnarray}
\label{action}
    A_c(r, \Delta) = \int_{\xi_1}^{\xi_2} d \xi\alpha'_1(\xi)^2+
\int_{\xi_2}^{1+\xi_1} d \xi\alpha'_2(\xi)^2 
\end{eqnarray}
This integral is elementary, but the expression is too long to be presented here. It can be found in the \Mathematica{} notebook\cite{MB41}, where it is used to select the roots $r_0(\Delta)$ of $g(r,\Delta)$,  minimizing $ A_c\left(r_0(\Delta), \Delta\right)$ for a given value of $\Delta \in (0,1)$.

    \pctPDF{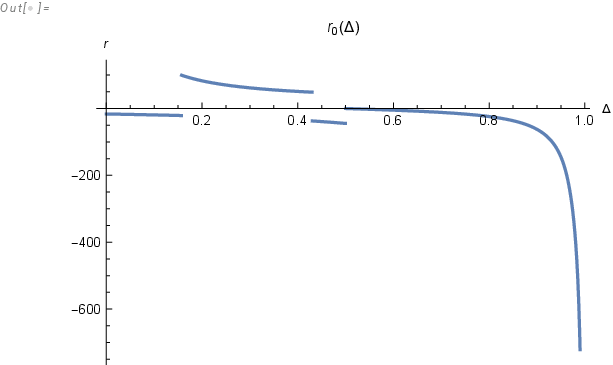}{Log plot of $r_0(\Delta).$}
\pctPDF{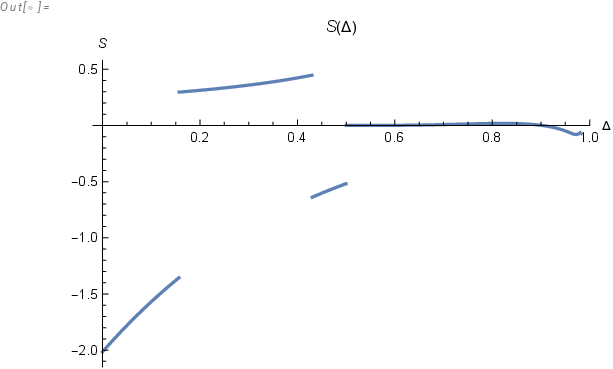}{Plot of $S\left(\Delta\right)$.}

\pctPDF{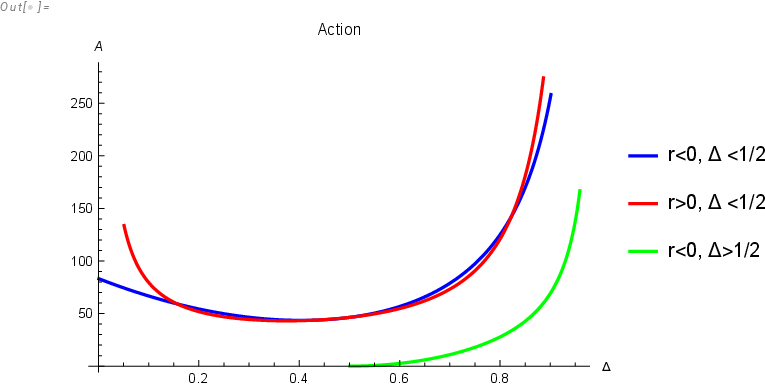}{Plot of $A_c\left(r_0(\Delta), \Delta\right)$ for the three real solutions of the equation $g(r,\Delta) =0$.
At $\Delta = \Delta_1$ and $\Delta = \Delta_2$, the action curves intersect; at $\Delta = \oh$, there is a gap between the lowest action ($=0$) and the lowest of the other two. So, there are second-order phase transitions at $\Delta_1, \Delta_2$ and the first-order phase transition at $\Delta = \oh$.}

This lowest action root is plotted in Fig.\ref{fig::RZeroPlot}. The corresponding value of minimal action $L(\Delta) = A_c\left(r_0(\Delta), \Delta\right)$ is plotted in Fig.\ref{fig::LPlot}.

There are phase transitions at 
\begin{eqnarray}
  &&  \Delta = \left\{\Delta_1, \Delta_2, \Delta_3 \right\};\\
  && \Delta_1 = 0.157143;\\
  && \Delta_2 = 0.43015;\\
  && \Delta_3 = \oh;
\end{eqnarray} 

These branch points in $\Delta$ correspond to the switch of the lowest action solution.
At small positive $\Delta -\oh$
\begin{eqnarray}
    &&r_0(\Delta) \to -48 (\Delta-\oh) -\frac{1536 (\Delta-\oh)^3}{7} + \dots;\\
    && A_0(r_0(\Delta), \Delta) \to 256 (\Delta- \oh)^2 + \dots
\end{eqnarray}

At $\Delta \to 1$ all solutions go to $-\infty $ as
\begin{eqnarray}
    r_n(\Delta) \to -\frac{6}{1-\Delta} + O(1)
\end{eqnarray}
This behavior matches numerical computations in \Mathematica{} \cite{MB41}. 

The constraint \eqref{quadraticConstraint} is also reduced to elementary functions, too lengthy to quote here (see \cite{MB41}).

This constraint yields the quadratic relation for the last unknown parameter $a$ in our solution
\begin{eqnarray}
   && a^2  = \frac{|\vec{k}| \sqrt{\nu t}}{S(\Delta) y \sqrt{X}};
\end{eqnarray}
with universal function $S(\Delta)$ presented in \cite{MB41} and shown in Fig. \ref{fig::SPlot}).

The resulting integral (up to the pre-exponential factor $Q$) is equal to
\begin{eqnarray}
  &&  \int [D \alpha] \alpha'(\xi_1) \alpha'(\xi_2)\delta\left(  \abs{\Im  \vec V(\xi_1,\xi_2)} - |\vec{k}|\sqrt{\nu  t}\right) \propto \nonumber\\
  && Q \VEV{\alpha'(\xi_1) \alpha'(\xi_2)} \exp{ - y \sqrt{X} |\vec{k}| \sqrt{\tilde{\nu} t} \frac{ \pi L(\Delta)}{2|S(\Delta)|}}
\end{eqnarray}
The factor $\VEV{\alpha'(\xi_1)\alpha'(\xi_2)}$ contains two terms :
\begin{eqnarray}
    \VEV{\alpha'(\xi_1)\alpha'(\xi_2)} = \alpha_1'(\xi_1)\alpha_1'(\xi_2) + \VEV{\delta\alpha'(\xi_1)\delta\alpha'(\xi_2)}
\end{eqnarray}
The first term is the contribution of the classical solution we have just found, and the second term comes from Gaussian fluctuations $\delta \alpha(\xi)$ around this solution. 

The classical term is calculable (see \cite{MB41})
    \begin{eqnarray}
    &&\alpha_1'(\xi_1)\alpha_1'(\xi_2)= 
    \frac{|\vec{k}|\sqrt{\tilde{\nu} t}}{N y \sqrt{X}}\frac{J\left(\Delta\right)}{|S(\Delta)|};\\
    &&J\left(\Delta\right)=\left.\frac{r A(r,\Delta) B(r,\Delta)}{\Delta(\Delta -1)  C(r,\Delta)^2}\right|_{r = r_0(\Delta)};\\
    && A(r,\Delta) = \nonumber\\
    &&\Delta  \sin \left(\sqrt{(\Delta -1) r}\right)-\sqrt{(\Delta -1) \Delta } \sin \left(\sqrt{\Delta  r}\right);\\
    && B(r,\Delta) = \Delta  \sin \left(\sqrt{(\Delta -1) r}\right) \cos \left(\sqrt{\Delta  r}\right)-\nonumber\\
    &&\sqrt{(\Delta -1) \Delta } \sin \left(\sqrt{\Delta  r}\right) \cos \left(\sqrt{(\Delta -1) r}\right);\\
    && C(r,\Delta) = \cos \left(\sqrt{(\Delta -1) r}\right)-\cos \left(\sqrt{\Delta  r}\right);
\end{eqnarray}

(see Fig. \ref{fig::JPlot}).

\pctPDF{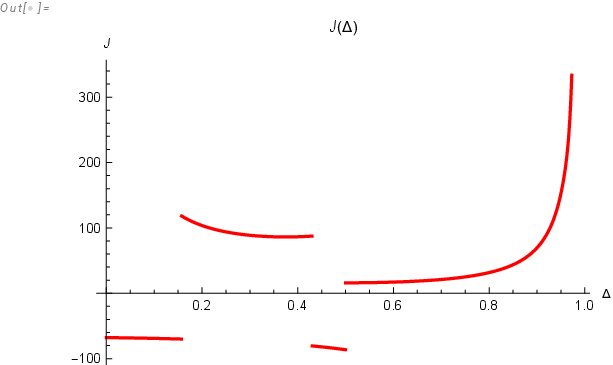}{Plot of universal function $J(\Delta)$.the four corves correspond to four phases (solutions for $r_0(\Delta)$ ).}

The fluctuation term $\VEV{\delta\alpha'(\xi_1)\delta\alpha'(\xi_2)}$  is also proportional to $1/N$, therefore we must keep  this term as well.

As for the pre-exponential factor $Q$ in the saddle point integral, it is given by the functional determinant of the operator $\hat L$ corresponding to linearized effective action \eqref{action} in the vicinity of the saddle point $\lambda_c, \alpha_c(\xi)$.
    \begin{eqnarray}
    &&A[\alpha_c + \delta \alpha, \lambda_c + \delta \lambda] \nonumber\\
    &&\to A[\alpha_c , \lambda_c ] + \oh \tr\left(V^\dagger\mid \hat L \mid V\right);\\
    && V = \left\{\delta\lambda, \delta \alpha\right\};\\
    && Q(\Delta,\tau) = \exp{-\oh \tr \left( \frac{\log \hat L}{\hat L^\alpha}\right)}_{\alpha\to 0};\\
    && \tau = y \sqrt{X};
\end{eqnarray}
The fluctuation correction reduces to the inverse operator $\hat L$, which we compute in the next section.
\pctPDF{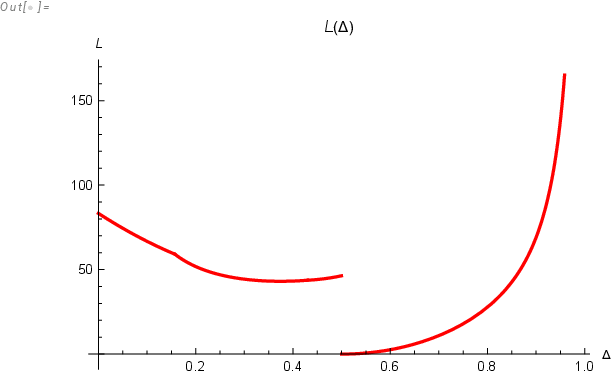}{Log plot of $L(\Delta) =A_c\left(r_0(\Delta), \Delta\right)$. }

Now, we can  reduce multiple sum/integral in \eqref{wwK} to the following
    \begin{eqnarray}
\label{finalCorr}
    &&\VEV{\vec\omega(\vec 0) \cdot \vec \omega(\vec{k})} = \frac{\tilde{\nu}^{\sfrac{3}{2}}H\left(k \sqrt{\tilde{\nu} t}\right)}{\sqrt{t}};\\
    &&H(\kappa) = \frac{1}{\mathcal Z}\sum_{n=1}^\infty\varphi(n) \int_0^{1/n} d \tau \left(\tau^5/n^5- \tau^{10}\right)\nonumber\\
    &&\int_0^1 d \Delta (1-\Delta) G\left(\Delta,\tau,\kappa\right);\\
&&G\left(\Delta,\tau,\kappa\right) = \exp{ -  \frac{\tau \kappa L(\Delta)}{2\pi|S(\Delta)|}}\nonumber\\
&&Q\left(\Delta,\tau\right)\left(\frac{\kappa}{N \tau}\frac{J\left(\Delta\right)}{|S(\Delta)|}+ \VEV{\delta\alpha'(\xi_1)\delta\alpha'(\xi_2)}\right)
\end{eqnarray}
where $\mathcal Z$ is the normalization constant to be determined later.

\section{Functional determinant in the  path integral}\label{FuncDet}

As we have discussed in the previous section,
in the limit $a \to 0$ the classical solution $\alpha_{1,2}(\xi) \propto a \to 0$. 

This observation simplifies the linearized theory 
corresponding to this quadratic form  $\VEV{V\mid\hat L\mid V}$.
First, integrate the fluctuations $\delta \lambda$ of $\lambda $ around the saddle point solution.

The Lagrange multiplier at the saddle point vanishes, as we show in \cite{MB41}
\begin{eqnarray}
    \lambda_0 = \tau r_0(\Delta) I(r_0(\Delta))=0 
\end{eqnarray}
The quadratic term comes from the first derivatives $I_\lambda = \partial_\lambda I, I_r = \partial_r I, \lambda_r = \partial_r \lambda$ , which can be simplified by switching to $\lambda(r) = \tau t I(r)$
\begin{eqnarray}
    A_{\lambda \lambda} = \tau I_\lambda = \frac{\tau I_r}{\lambda_r} = \frac{\tau I_r}{\tau r I_r} = \frac{1}{r}
\end{eqnarray}

The bilinear term $\lambda \delta\alpha $ also simplifies
\begin{eqnarray}
    &&A_{\alpha \lambda}(\delta \alpha) = \imath \tau \delta \lambda F[\delta \alpha];\\
    &&F[\delta \alpha] = \int_{\xi_1}^{\xi_2} \frac{d \xi \delta \alpha(\xi)} { \xi_2-\xi_1} -
\int_{\xi_2}^{1+\xi_1} \frac{d \xi \delta \alpha(\xi) }{1+ \xi_1-\xi_2}
\end{eqnarray}

We can integrate out $\lambda$, producing the extra pre-exponential factor $ Q_\lambda = \sqrt{|r_0(\Delta)|} /\sqrt{N}$.

The bilinear term  in the exponential after $\lambda $ integration leads to the following effective quadratic Action for $\delta \alpha$
\begin{eqnarray}
\label{Aeff}
   && A_{eff}[\delta \alpha] =\frac{\tau^2}{2} \int_{\xi_1}^{1+ \xi_1}\delta \alpha'^2 + \frac{r_0(\Delta) \tau^2}{2} F[\delta\alpha]^2;
\end{eqnarray}
There is a zero-mode $\delta\alpha(\xi) = \const{}$, related to translational invariance of $A_{eff}[\delta \alpha]$. Naturally, this zero-mode must be eliminated from the spectrum when we compute the functional determinant and the resolvent below.

After discarding the zero-mode, this effective action becomes a positive definite functional of $\delta\alpha$ only in the region of $\Delta$ where $r_0(\Delta) >0$, i.e., for $\Delta_1 < \Delta < \Delta_2$.

As we shall see below, the spectrum of fluctuations is positive only in this region. Therefore, we restrict our integration to this region.

The $(\delta \alpha)^2$ term corresponds to the linear eigenvalue equation with $f_{1,2} = \delta \alpha_{1,2}$
\begin{eqnarray}
  && f_{1,2}''(\xi) - \mu_{1,2} F[f] =-\omega f_{1,2}(\xi) ;\\
  && \mu_{1,2} = \left\{ \frac{r}{\Delta}, \frac{r}{\Delta-1}\right\};\\
  && \epsilon = \omega \tau^2;\\
  && r = r_0(\Delta);
\end{eqnarray}
The solution matching with first derivative at $\xi = \xi_2, \xi = (\xi_1,1+ \xi_1)$ 
is built the same way as in \eqref{bsol}. Equations for $f_{1,2}$ being linear homogeneous, we can fix the normalization as $F[f] =1$,

    \begin{eqnarray}
    &&f_1(x)=a \sin (\sqrt{\omega} (\xi_-\xi_2))+ \nonumber\\
    &&B_1 \cos (\sqrt{\omega} (\xi_-\xi_2))+\frac{r}{ \omega \Delta };\\
    &&f_2(\xi)=a \sin (\sqrt{\omega} (\xi_-\xi_2))+ \nonumber\\
    &&B_2 \cos (\sqrt{\omega} (\xi_-\xi_2))+\frac{r}{\omega (\Delta -1)};
\end{eqnarray}

The spectrum $\omega = \omega_n$ is defined by the transcendental equation (the discriminant of this linear system of equations), which we found in \cite{MB41}
    \begin{eqnarray}
\label{spectrumEq}
    && f(\omega_n, \Delta) =0;\\
    && f(\omega, \Delta) = \nonumber\\
    &&(\Delta -1) \Delta  \sqrt{\omega } \sin \left(\frac{\sqrt{\omega }}{2}\right) ((\Delta -1) \Delta  \omega +r)+\nonumber\\
    &&r \cos \left(\frac{1}{2} (1-2 \Delta ) \sqrt{\omega }\right)-r \cos \left(\frac{\sqrt{\omega }}{2}\right);\\
    && r = r_0(\Delta);
\end{eqnarray}
The spectrum is positive in the  interval $\Delta_1 < \Delta < \Delta_2$  where $r_0(\Delta) >0$ so that the solution for $\omega_n(\Delta)$ is stable.
In the following, we only select the stable region with positive $r_0(\Delta)$

\pctPDF{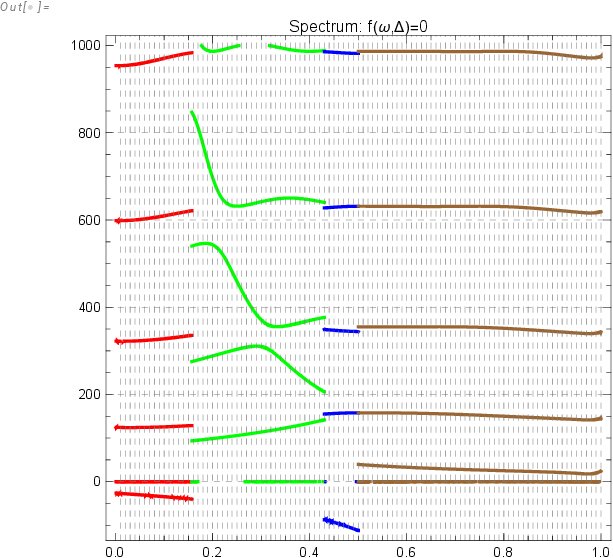}{The first levels of the spectrum satisfying equation $f(\omega, \Delta) =0$. The colored lines correspond to four phases. Red: $ 0 < \Delta < \Delta_1$, Green: $ \Delta_1 < \Delta < \Delta_2$, Blue: $ \Delta_2 < \Delta < \oh$,
Brown: $ \oh < \Delta < 1$. The green zone is left as stable, and others are eliminated because $r_0(\Delta) <0$ in these zones.
Naturally, we eliminate the zero-mode $\epsilon_0=0$ corresponding to translational invariance of the effective Action.}

The functional determinant, resulting from the WKB approximation to the $\alpha$ path integral, would be related to the infinite product of positive eigenvalues $\epsilon_n = \tau^2 \omega_n$, which can be written using  a contour integral 
    \begin{eqnarray}
\label{DetAlpha}
    &&Q_\alpha\left(\Delta,\tau\right) =\prod_{\omega_n > 0} (\tau^2\omega_n)^{-\oh} =\nonumber\\
    &&\left .\exp{ \oh \partial_x \Im\oint_\Gamma \frac{f'(\omega)}{f(\omega)}\frac{d \omega }{2 \pi (\omega\tau^2)^{x}}} \right|_{x\to 0};
 \end{eqnarray}
and the integration contour $\Gamma$ encircles anticlockwise the positive real poles of the meromorphic function $f'(\omega)/f(\omega)$. 
The integral converges at $x> \oh$ and should be analytically continued to $x =0$.

For this purpose, let us introduce another function
\begin{eqnarray}
    \Phi(\omega) =  \frac{f(\omega)}{\cos(\sqrt{\omega}/2)} \omega^{-3/2}
\end{eqnarray}
We show in \cite{MB41} that at large $\omega = \imath y$ this function reaches finite limits
\begin{eqnarray}
   &&\Phi(\imath y) \to  \imath \sign{y} (\Delta -1)^2 \Delta ^2+\frac{(\Delta -1) \Delta  r}{|y| }
\end{eqnarray}

The logarithmic derivative of the original function differs from $\frac{\Phi'(\omega)}{\Phi(\omega)} $ by the following meromorphic function 
\begin{eqnarray}
   && \frac{f'(\omega)}{f(\omega)} -\frac{\Phi'(\omega)}{\Phi(\omega)} = -\frac{\tan \left(\frac{\sqrt{\omega }}{2}\right)}{4 \sqrt{\omega }} + \frac{3}{2 \omega}
\end{eqnarray}
This difference produces a calculable contribution to our integral.
By summing residues of the poles of the tangent, we get
    \begin{eqnarray}
    &&\oint_\Gamma \left(-\frac{\tan \left(\frac{\sqrt{\omega }}{2}\right)}{4 \sqrt{\omega }} + \frac{3}{2 \omega}\right)
    \frac{d \omega }{2 \pi (\omega\tau^2)^{x}} = \nonumber\\
    &&\imath\left(1-2^{2 x }\right) (2 \pi \tau)^{-2 x } \zeta (2 x )
\end{eqnarray}
The derivative at $x =0$ yields a constant
\begin{eqnarray}
  \oh   \partial_x \Im \imath\left(1-2^{2 x }\right) (2 \pi \tau)^{-2 x } \zeta (2 x )\to \frac{\log (2)}{2}
\end{eqnarray}
leading to an irrelevant renormalization of $Q_\alpha(\Delta,\tau) $ by a factor $\sqrt{2}$.

The remaining integral with $ f(\omega) \Rightarrow \Phi(\omega) $ already converges at $\Re x  > -1$, so that we can set $x =0$ and rotate the integration contour $\Gamma$ parallel to the imaginary axis at $\Re \Gamma = \epsilon >0$:
\begin{eqnarray}
\label{RegDet}
&&Q_\alpha\left(\Delta,\tau\right)= \nonumber\\
    && \exp{\oh \Im\int_{\eps - \imath\infty}^{\eps + \imath\infty} \frac{\Phi'(\omega)}{\Phi(\omega)}  \frac{\log (\omega \tau^2) \,d \omega }{2 \pi }}
\end{eqnarray}
The remarkable property of this functional determinant is the factorization of the $\tau$ dependence
\begin{eqnarray}
    &&Q_\alpha(\Delta,\tau)= \tau^{\mu(\Delta)} Q_\alpha(\Delta,1);\\
    \label{muDelta}
    && \mu(\Delta) =  \Im\int_{\epsilon-\imath\infty}^{\epsilon +\imath\infty}  \frac{\Phi'(\omega)}{\Phi(\omega)}   \frac{ \,d \omega }{2 \pi }
\end{eqnarray}

The index $\mu(\Delta)$ has a topological origin and can be computed analytically.

\begin{eqnarray}
    \mu(\Delta) =  \frac{\arg{\Phi(\imath\infty)} - \arg{\Phi(-\imath\infty)}}{2 \pi} = \oh
\end{eqnarray}
Our result for the correlation function is given by \eqref{finalCorr} with
\begin{eqnarray}
    Q\left(\Delta,\tau\right) = Q_\alpha\left(\Delta,1\right) \tau^{\oh} \sqrt{r_0(\Delta)}
\end{eqnarray}
and $Q_\alpha(\Delta,1)$ given by \eqref{RegDet}. All the constant factors we have omitted here are absorbed by the normalization factor $\mathcal Z$, which we determine at the end of the next section.

\section{The fluctuation term in $\alpha'(\xi_1) \alpha'(\xi_2)$}\label{FlucTerm}

The last missing term is the fluctuation contribution to $\alpha'(\xi_1) \alpha'(\xi_2)$. In the Gaussian approximation, valid at $N \to \infty$, this term equals
\begin{eqnarray}
\label{alpaprimG}
    \VEV{\delta\alpha'(\xi_1)\delta\alpha'(\xi_2)} = \frac{-1}{N \tau^2} \left[\partial_\xi \partial_{\xi'} G(\xi,\xi')\right]_{\xi= \xi_1, \xi'= \xi_2}
\end{eqnarray}
where $G(\xi,\xi')$ is a resolvent for the effective quadratic Action \eqref{Aeff}. This resolvent satisfies the equation
    \begin{eqnarray}
   && \partial^2_\xi G(\xi,\xi') - \mu(\xi) F[G] = \delta(\xi-\xi');\\
   && G(\xi_1, \xi') = G(\xi_1+1, \xi') = 0;\\
   && F[G] = \int_{\xi_1}^{\xi_2} \frac{d \xi G(\xi,\xi')} { \xi_2-\xi_1} -\nonumber\\
    &&
\int_{\xi_2}^{1+\xi_1} \frac{d \xi G(\xi,\xi') }{1+ \xi_1-\xi_2};\\
    &&\mu( \xi) = 
    \begin{cases}
        \frac{r}{\Delta} & \text{ if } \xi_1 \le \xi < \xi_2\\
        \frac{r}{\Delta-1} & \text{ if } \xi_2 \le\xi < 1+\xi_1\\
    \end{cases}
\end{eqnarray}
The solution of this equation, matching with the first derivative at $\xi = \xi_2$ is 
    \begin{eqnarray}
    &&G(\xi, \xi')_{\xi_1 \le \xi < \xi_2} =
        A+\frac{| \xi-\xi'| }{2}+B (\xi-\xi_2)+\nonumber\\
    &&\frac{F[G] r (\xi-\xi_2)^2}{2\Delta}\\
    &&G(\xi, \xi')_{\xi_2 \le \xi < 1+\xi_1} =
        A+\frac{| \xi-\xi'| }{2}+B (\xi-\xi_2)+\nonumber\\
    &&\frac{F[G] r (\xi-\xi_2)^2}{2(\Delta-1)};
\end{eqnarray}
The linear functional $F[G]$ on this solution becomes a linear function of these unknown parameters $A, B$.
Two boundary conditions $G(\xi_1, \xi') = G(\xi_1+1, \xi') = 0$ fix these parameters as functions of $\xi_1, \xi_2, \xi'$.

The result derived in \cite{MB41} is too lengthy to present here. The desired quantity
\eqref{alpaprimG} is quite simple
\begin{eqnarray}
    \frac{1}{N \tau^2} \left[\partial_\xi \partial_{\xi'} G(\xi,\xi')\right]_{\xi= \xi_1, \xi'= \xi_2} = \frac{2 (r-6 )}{(r+12)N \tau^2}
\end{eqnarray}

Finally, we get the following correlation \eqref{finalCorr} (absorbing the constant factors in $\mathcal Z$)
\begin{subequations}\label{answer}
      \begin{eqnarray}
&&\nu\VEV{\vec\omega(\vec 0) \cdot \vec \omega(\vec{k})} = \frac{\tilde{\nu}^{\sfrac{3}{2}}H\left(k \sqrt{\tilde{\nu} t}\right)}{\sqrt{t}};\\
    &&H(\kappa) =\int_{\Delta_1}^{\Delta_2} d \Delta (1-\Delta)\displaystyle\sum_{n=1}^\infty\varphi(n) \nonumber\\
    &&\int_0^{1/n} \frac{d \tau }{\tau^{\frac{5}{2}}}\left(\tau^5/n^5- \tau^{10}\right) G\left(\Delta,\tau\kappa\right);\\
    &&G\left(\Delta,x\right) = \frac{Q_\alpha\left(\Delta,1\right)  \sqrt{r_0(\Delta)}}{\mathcal Z}\nonumber\\
    &&\left(x \frac{J\left(\Delta\right)}{S(\Delta)}+ \frac{2 (r_0(\Delta)-6)}{(12+r_0(\Delta))}\right)\exp{ -  \frac{x L(\Delta)}{2\pi S(\Delta) }}
\end{eqnarray}
\end{subequations}

\section{Mellin integral for the energy spectrum and energy decay}
\label{RiemanZeros}
The spectral function $H(\kappa)$ can be computed as follows.
Let us represent the theta function as inverse Mellin transform (at $n>0, \tau >0$)
\begin{eqnarray}
    \theta(1-n \tau) = 
    \int\displaylimits_{+0 -\imath\8}^{+0 +\imath \8} \frac{d p}{2 \pi \I p} n^{-p} \tau^{-p}
\end{eqnarray}
Substituting this representation into our integral in \eqref{answer}, and interchanging summations/integrations, we find

\begin{eqnarray}
    &&H(\kappa) =\int_{\Delta_1}^{\Delta_2} d \Delta (1-\Delta)\int\displaylimits_{+0 -\imath\8}^{+0 +\imath \8}  \frac{d p}{2 \pi \I p} \sum_{n=1}^\infty \frac{\varphi(n)}{n^p}\nonumber\\
&&\int_0^{\infty }\, d t \left(\frac{t^{5/2}}{n^5}-t^{15/2}\right) t^{-p} (A+B t \kappa ) e^{-C t \kappa }   ;\\
    && A =\frac{Q_\alpha\left(\Delta,1\right)  \sqrt{r_0(\Delta)}}{\mathcal Z}\frac{2 (r_0(\Delta)-6)}{(12+r_0(\Delta))};\\
    && B = \frac{Q_\alpha\left(\Delta,1\right)  \sqrt{r_0(\Delta)}}{\mathcal Z} \frac{J\left(\Delta\right)}{S(\Delta)};\\
    && C = \frac{ L(\Delta)}{2\pi S(\Delta) };
\end{eqnarray}

The last integral reduces to Gamma functions and powers of $n$, after which the sum over $n$ reduces to the ratio of two zeta functions

\begin{eqnarray}
 && H(\kappa)=\int_{\Delta_1}^{\Delta_2} d \Delta (1-\Delta)\int_{-2-\I \infty}^{-2 + \I \infty} \frac{d z\, \kappa^z}{2 \pi \I}\nonumber\\
&&\frac{20 C^{z-1} (A C-B z)\zeta \left(z+\frac{15}{2}\right) \Gamma (-z)}{(2 z+7) (2 z+17) \zeta \left(z+\frac{17}{2}\right)}
\end{eqnarray}

According to the Riemann hypothesis, the function $\zeta(w), w= z+\frac{17}{2}$ in the denominator has trivial zeros at negative even points $z =-2,-4,\dots$ and nontrivial zeros at the critical line $\Re w = \oh$.

There is also a contribution from the poles at $z+ \frac{15}{2} =1, (2 z+7)=0,  (2 z+17)  =0$.

We are not going to sum the residues in these poles, but rather integrate along the line $ z = -2 + \I y$, where the integrand exponentially decreases as $ \exp{- \frac{\pi |y|}{2}}$ in both directions.

\pctPDF{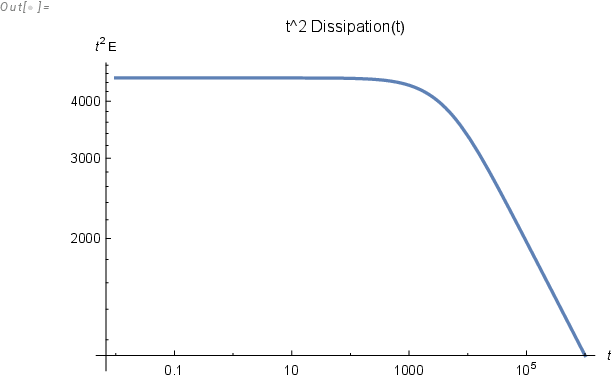}{The log-log plot energy dissipation $ t^2 \mathcal E(t)$ as a function of $t$ has a regime change due to quantum effects. It starts as constant at small t and asymptotically decays as $t^2 \mathcal E \propto t^{-\sfrac{1}{4}}$.}
There are oscillations on top of this exponential decrease. The integrals can be computed with arbitrary accuracy in \Mathematica{} using the "DoubleExponentialOscillatory" integration method.
Each integration takes a fraction of a second. The remaining integral over $\Delta$ was performed by computing the table of (positive) values of integrand $I(\Delta)$ on a grid with step $\delta \Delta  = \frac{\Delta_2 - \Delta_1}{100}$, and interpolating $\log(I(\Delta))$ between these values by fifth order polynomial $P_5(\Delta)$.
Integral of the exponential of this polynomial $\int d \Delta(1-\Delta) \exp{P_5(\Delta)}$
yields good enough accuracy ( floating precision $10^{-8}$).

This process was repeated for all M =1000 values of $\kappa = 0.1 (1000/0.1)^{n/M}, n =0, \dots, M $, corresponding to equidistant log grid.
This table was computed in parallel in \Mathematica{} in \cite{MB44}, which took just a few minutes.
Here is the resulting function  $ H(\kappa) \kappa^{\sfrac{5}{3}}$ (see Fig.\ref{fig::Hk53}). It shows strong deviations from K41 spectrum, in agreement with the DNS at Fig.\ref{fig::DecayingSpectrum}.

The asymptotic behavior of this function is
\begin{eqnarray}
    H(\kappa) \propto \kappa^{-\sfrac{7}{2}}
\end{eqnarray}
The effective index $ \mu(\kappa) =  \dbyd{\log H(\kappa)}{\log\kappa}$ slowly decreases reaching $\mu(\infty) = -\sfrac{7}{2}$, as shown on Fig.\ref{fig::MuIndex}.

The function $F(\kappa)= \int_{\kappa}^\infty x^2 H(x) \, d x$ was obtained by numerical integration of exponential of interpolated function as before, with asymptotic tail $ H(\kappa) \to const{} \,\kappa^{-\sfrac{7}{2}} $ integrated analytically.

The log-log plot of the energy dissipation as a function of time is shown in Fig. \ref{fig::DissipationPlot}, and it is curved in the log-log scale due to the quantum effects (complex poles of Mellin transform at Riemann $\zeta$ function zeros).

\end{document}